\documentclass[
paper=letter,%
numbers=noendperiod,%
captions=nooneline,%
DIV=12%
]{scrartcl}

\usepackage[T1]{fontenc}%
\usepackage{lmodern}%
\usepackage[american]{babel}%
\usepackage{microtype}%

\usepackage[
hyperref,%
table%
]{xcolor}%

\usepackage{scrlayer-scrpage}%

\usepackage{calc}%
\usepackage{rotating}%
\usepackage{amsmath}%
\usepackage{mathtools}%
\usepackage{amssymb}%
\usepackage{amsthm}%
\usepackage{thmtools}%
\usepackage{etoolbox}%
\usepackage{xparse}%
\usepackage{bm}%
\usepackage{bbm}%
\usepackage{enumitem}%
\usepackage[calc]{datetime2}%
\usepackage{graphicx}%
\usepackage{graphbox}%
\usepackage{grffile}%
\usepackage{tikz}%
\usetikzlibrary{arrows.meta, positioning, calc}
\usepackage{wrapfig}%
\usepackage{tabularx}%
\usepackage{siunitx}%
\usepackage{booktabs}%
\usepackage{multirow}%
\usepackage{vruler}%
\usepackage{fancyvrb}%
\usepackage{textcomp}%
\usepackage{listings}%
\usepackage{csquotes}%
\usepackage[
style=authoryear,%
dashed=false,%
hyperref=true,%
useprefix=true,%
maxnames=2,%
uniquelist=minyear,%
maxbibnames=6,%
uniquename=false,%
seconds=true,%
date=iso,%
urldate=iso%
]{biblatex}%
\usepackage[
hypertexnames=false,%
setpagesize=false,%
pdfborder={0 0 0},%
pdfstartview=Fit,%
bookmarksopen=true,%
bookmarksnumbered=true%
]{hyperref}%
\usepackage[hyphens]{xurl}
\hypersetup{breaklinks=true}
\setkomafont{pageheadfoot}{\normalfont\normalcolor\sffamily}%
\setkomafont{pagenumber}{\normalfont\normalcolor\sffamily}%
\automark{section}%
\setkomafont{captionlabel}{\normalfont\normalcolor\sffamily\bfseries}%

\newcommand*{\mysquare}{\rule[0.18em]{0.36em}{0.36em}}
\newcommand*{\mytriangle}{\raisebox{0.12em}{\resizebox{0.48em}{0.48em}{$\blacktriangleright$}}}
\newcommand*{\mybar}{\rule[0.32em]{0.62em}{0.08em}}
\newcommand*{\mydot}{\raisebox{0.14em}{\resizebox{0.44em}{!}{$\bullet$}}}
\setlist{%
  align=left,%
  labelindent=0mm, %
  leftmargin=!,%
  itemindent=0mm, %
  listparindent=\parindent,%
  parsep=0mm,%
  topsep=1mm,%
  itemsep=1mm%
}
\setlist[itemize,1]{label={\mysquare}, labelwidth=\widthof{\mysquare\ }}%
\setlist[itemize,2]{label={\mytriangle}, labelwidth=\widthof{\mytriangle\ }}%
\setlist[itemize,3]{label={\mybar}, labelwidth=\widthof{\mybar\ }}%
\setlist[itemize,4]{label={\mydot}, labelwidth=\widthof{\mydot\ }}%
\setlist[enumerate,1]{label=\arabic*), labelwidth=\widthof{9)}}%
\setlist[enumerate,2]{label=\arabic{enumi}.\arabic*), labelwidth=\widthof{9.9)}}%
\setlist[enumerate,3]{label=\arabic{enumi}.\arabic{enumii}.\arabic*), labelwidth=\widthof{9.9.9)}}%
\setlist[enumerate,4]{label=\arabic{enumi}.\arabic{enumii}.\arabic{enumiii}.\arabic*), labelwidth=\widthof{9.9.9.9)}}%

\allowdisplaybreaks%
\newcommand*{\abstractnoindent}{}%
\let\abstractnoindent\abstract
\renewcommand*{\abstract}{\let\quotation\quote\let\endquotation\endquote
  \abstractnoindent}
\deffootnote[1em]{1em}{1em}{\textsuperscript{\thefootnotemark}}%
\definecolor{blue}{RGB}{58, 95, 205}%
\definecolor{red}{RGB}{205, 41, 144}%
\definecolor{orange}{RGB}{238, 118, 0}%
\definecolor{chocolate}{RGB}{205, 102, 29}%

\lstdefinestyle{input}{
  backgroundcolor=\color{black!12},%
  commentstyle=\itshape\color{black!50},%
  keywordstyle=\bfseries\color{black},%
  stringstyle=\color{black}%
}

\lstdefinestyle{output}{
  backgroundcolor=\color{black!6}%
}

\lstdefinestyle{codestyle}{
  language={},%
  keywords={},%
  otherkeywords={}%
}

\lstnewenvironment{codeinput}[1][]{%
  \lstset{style=input, style=codestyle}
  #1%
}{\vspace{-0.25\baselineskip}}%

\lstnewenvironment{codeoutput}[1][]{%
  \lstset{style=output, style=codestyle}
  #1%
}{\vspace{-0.25\baselineskip}}%

\expandafter\let\csname Sinput\endcsname\relax
\expandafter\let\csname endSinput\endcsname\relax
\expandafter\let\csname Soutput\endcsname\relax
\expandafter\let\csname endSoutput\endcsname\relax

\lstdefinestyle{Rstyle}{
  language=R,%
  keywords={},%
  otherkeywords={}%
}

\lstnewenvironment{Sinput}[1][]{%
  \lstset{style=input, style=Rstyle}
  #1%
}{\vspace{-0.25\baselineskip}}%

\lstnewenvironment{Soutput}[1][]{%
  \lstset{style=output, style=Rstyle}
  #1%
}{\vspace{-0.25\baselineskip}}%

\lstdefinestyle{Cstyle}{
  language=C,%
  keywords={},%
  otherkeywords={}%
}

\lstnewenvironment{Cinput}[1][]{%
  \lstset{style=input, style=Cstyle}
  #1%
}{\vspace{-0.25\baselineskip}}%

\lstnewenvironment{Coutput}[1][]{%
  \lstset{style=output, style=Cstyle}
  #1%
}{\vspace{-0.25\baselineskip}}%

\lstdefinestyle{Bashstyle}{
  language=bash,%
  keywords={},%
  otherkeywords={}%
}

\lstnewenvironment{Bashinput}[1][]{%
  \lstset{style=input, style=Bashstyle}
  #1%
}{\vspace{-0.25\baselineskip}}%

\lstnewenvironment{Bashoutput}[1][]{%
  \lstset{style=output, style=Bashstyle}
  #1%
}{\vspace{-0.25\baselineskip}}%

\lstdefinestyle{LaTeXstyle}{
  language=[LaTeX]{TeX},%
  texcs={},%
  keywords={},%
  otherkeywords={}%
}

\lstnewenvironment{LaTeXinput}[1][]{%
  \lstset{style=input, style=LaTeXstyle}
  #1%
}{\vspace{-0.25\baselineskip}}%

\lstnewenvironment{LaTeXoutput}[1][]{%
  \lstset{style=output, style=LaTeXstyle}
  #1%
}{\vspace{-0.25\baselineskip}}%

\DeclareNameAlias{sortname}{family-given}%
\DefineBibliographyExtras{american}{\DeclareQuotePunctuation{}}%
\renewbibmacro*{volume+number+eid}{%
  \setunit*{\addcomma\space}%
  \printfield{volume}%
  \printfield{number}}
\DeclareFieldFormat*{number}{(#1)}
\DeclareFieldFormat*{title}{#1}%
\DeclareFieldFormat{doi}{%
  \ifhyperref
    {\href{http://dx.doi.org/#1}{\nolinkurl{doi:#1}}}%
    {\nolinkurl{doi:#1}}}%
\renewbibmacro*{in:}{}%
\DeclareFieldFormat{isbn}{ISBN #1}%
\DeclareFieldFormat{pages}{#1}%
\DeclareFieldFormat{url}{\url{#1}}%
\DeclareFieldFormat{urldate}{\mkbibparens{#1}}%
\renewcommand*{\cite}[2][]{\textcite[#1]{#2}}%

\newif\ifstarttheorem
\declaretheoremstyle[%
spaceabove=0.5em,
spacebelow=0.5em,
headfont=\sffamily\bfseries\global\starttheoremtrue,
notefont=\sffamily\bfseries,
notebraces={(}{)},
headpunct={},
bodyfont=\normalfont,
postheadspace=\newline%
]{myMainStyle}
\declaretheorem[style=myMainStyle, numberwithin=section]{definition}%
\declaretheorem[style=myMainStyle, sibling=definition]{proposition}
\declaretheorem[style=myMainStyle, sibling=definition]{lemma}

\declaretheorem[style=myMainStyle, sibling=definition]{remark}
\declaretheorem[style=myMainStyle, sibling=definition]{example}
\declaretheorem[style=myMainStyle, sibling=definition]{algorithm}

\makeatletter
\preto\itemize{%
  \if@inlabel
    \ifstarttheorem
      \mbox{}\par\nobreak\vskip\glueexpr-\parskip-\baselineskip+0.25em\relax\hrule\@height\z@
    \fi%
  \fi%
  \global\starttheoremfalse%
  \def\tempa{proof}%
  \ifx\tempa\mycurrenvir
    \ifstarttheorem
      \mbox{}\par\nobreak\vskip\glueexpr-\parskip-\baselineskip+0.25em\relax\hrule\@height\z@
    \fi%
  \fi%
  \global\starttheoremfalse%
}
\preto\enditemize{\global\starttheoremfalse}
\makeatother

\makeatletter
\preto\enumerate{%
  \if@inlabel
    \ifstarttheorem
      \mbox{}\par\nobreak\vskip\glueexpr-\parskip-\baselineskip+0.25em\relax\hrule\@height\z@
    \fi%
  \fi%
  \global\starttheoremfalse%
  \def\tempa{proof}%
  \ifx\tempa\mycurrenvir
    \ifstarttheorem
      \mbox{}\par\nobreak\vskip\glueexpr-\parskip-\baselineskip+0.25em\relax\hrule\@height\z@
    \fi%
  \fi%
  \global\starttheoremfalse%
}
\preto\endenumerate{\global\starttheoremfalse}
\makeatother

\newcommand{\tb}[2]{\substack{#1\\#2}}

\makeatletter
\NewDocumentCommand{\tmb}{O{0.1mm} O{0.1mm} O{0.88} m m m}{%
  \mathrel{%
    \vbox{\offinterlineskip\m@th
      \ialign{%
        \hfil##\hfil\cr
        $\scriptscriptstyle\text{\scalebox{#3}{#4}}\mathstrut$\cr%
        \noalign{\vspace{#1}}%
        \vtop{%
          \ialign{%
            \hfil##\hfil\cr
            $#5$\cr\noalign{\vspace{#2}}%
            $\scriptscriptstyle\text{\scalebox{#3}{#6}}\mathstrut$\cr%
          }%
        }\cr
      }%
    }%
  }%
}
\makeatother

\makeatletter
\NewDocumentCommand{\tmbc}{O{0.1mm} O{0.1mm} O{0.88} m m m}{%
  \mathrel{%
    \vbox{\offinterlineskip\m@th
      \ialign{%
        \hfil##\hfil\cr
        $\scriptscriptstyle\mathclap{\text{\scalebox{#3}{#4}}}\mathstrut$\cr%
        \noalign{\vspace{#1}}%
        \vtop{%
          \ialign{%
            \hfil##\hfil\cr
            $#5$\cr\noalign{\vspace{#2}}%
            $\scriptscriptstyle\mathclap{\text{\scalebox{#3}{#6}}}\mathstrut$\cr%
          }%
        }\cr
      }%
    }%
  }%
}
\makeatother

\usetikzlibrary{shapes.geometric}

\newcommand*{\T}{^{\top}}

\newcommand*{\isim}{\tmb{\tiny ind.}{\sim}{}}%
\newcommand*{\IN}{\mathbb{N}}

\newcommand*{\IR}{\mathbb{R}}

\newcommand*{\Beta}{\operatorname{Beta}}

\newcommand*{\U}{\operatorname{U}}
\newcommand*{\B}{\operatorname{B}}

\newcommand*{\N}{\operatorname{N}}

\newcommand*{\I}{\mathbbm{1}}
\newcommand*{\rd}{\mathrm{d}}

\newcommand*{\ran}{\operatorname{ran}}

\renewcommand*{\P}{\mathbb{P}}
\newcommand*{\E}{\mathbb{E}}

\newcommand*{\var}{\operatorname{var}}

\newcommand*{\cor}{\operatorname{cor}}

\newcommand*{\VaR}{\operatorname{VaR}}

\newcommand*{\ES}{\operatorname{ES}}

\newcommand*{\logit}{\operatorname*{logit}}
\newcommand*{\logiti}{{\operatorname*{logit}}^{-1}}
\newcommand*{\ARMA}{\operatorname{ARMA}}
\newcommand*{\GARCH}{\operatorname{GARCH}}

\newcommand*{\R}{\textsf{R}}
\newcommand*{\eps}{\varepsilon}

\begin{document}

\thispagestyle{plain}
\begin{center}
  \sffamily
  {\bfseries\LARGE Bernoulli amputation\par}
  \bigskip\smallskip
  {\Large Marius Hofert\footnote{Department of Statistics and Actuarial Science, The University of
      Hong Kong,
      \href{mailto:mhofert@hku.hk}{\nolinkurl{mhofert@hku.hk}}},
    James Jackson\footnote{The Alan Turing Institute, London, UK, \href{mailto:jjackson@turing.ac.uk}{\nolinkurl{jjackson@turing.ac.uk}}},
    Niels Hagenbuch\footnote{F.\ Hoffmann-La Roche AG, Basel, Switzerland, \href{mailto:niels.hagenbuch@roche.com}{\nolinkurl{niels.hagenbuch@roche.com}}}
    \par
    \bigskip
    \today\par}
\end{center}
\par\smallskip
\begin{abstract}
A novel, stochastic approach to amputation, the process of introducing missing values to a complete dataset, is presented. It allows one to construct a wide variety of missingness patterns by only having to specify distributions of missingness indicators as opposed to specifying each missingness pattern manually. Missingness indicators are modeled in a principled way via copulas and Bernoulli margins, thus allowing one to incorporate dependence in missingness patterns.  Besides more classical missingness mechanisms such as missing completely at random, missing at random, and missing not at random, the approach is able to model structured missingness such as block missingness and, via mixtures, monotone missingness, which are patterns of missing data frequently found in real-life datasets. Properties such as joint missingness probabilities or missingness correlation are derived mathematically. The flexibility of the approach in capturing different missingness patterns while only requiring to specify distributional assumptions on missingness indicators is demonstrated with mathematical examples and empirical illustrations in terms of a well-known example dataset of sufficiently small sample size that allows to identify each missing data point visually. Finally, an example application to multivariate financial time series is provided.
\end{abstract}
\minisec{Keywords}
Amputation, copula, missingness indicator matrix, missingness mechanisms, structured missingness
\minisec{MSC2010} 62D10, 62H99, 65C60 %

\section{Introduction}\label{sec:intro}
Owing to the ubiquity of missing data in real-world datasets, \emph{imputation}
algorithms are an important area of research; see \cite{Molenberghs2014},
\cite{LittleRubin2020}, or \cite{Carpenter2023}. The effectiveness of such
algorithms is typically evaluated with simulation studies. However, somewhat
paradoxically, empirical datasets affected by missingness are unsuitable
as inputs for such simulations. This is because neither the missing values' true
underlying (but unobserved) values, nor the mechanism which led to their
missingness, is known. As a result, standard practice in such studies is to
take a complete dataset with known underlying distribution and introduce missing
values to it. This process is called \emph{amputation}; see
\cite{Schouten2018}. Analogously to multiple imputation, we introduce the term
\emph{multiple amputation} when constructing several amputed datasets, which
allows performance to be assessed over replications of a distributional missingness pattern.

In this paper, we address the question of how to create stochastic (including
deterministic) missingness patterns for amputation (rather than which type of
pattern is more appropriate for which dataset and amputation task at hand). Our
proposed approach is based on the stochastic modelling of missingness indicators
\emph{marginally}, i.e., each individually, via Bernoulli random variables
(indicating with the value $1$ the missingness of the respective value) and
\emph{jointly}, i.e., the stochastic dependence among the Bernoulli random
variables, via \emph{copulas}, i.e., joint distribution functions with standard
uniform univariate margins. This allows us to model dependencies among
missingness indicators in a natural and principled way without having
to specify missingness patterns manually.

We assume the reader to be familiar with basic stochastic modelling concepts,
including almost sure equality, equality in distribution, the notion of
multivariate distribution functions, their marginal distribution functions, and
copulas, as well as random elements (random variables, random vectors, and
random matrices; in particular, ``random'' does not necessarily mean independent
and identically distributed (iid) here unless stated). As our modelling approach
is stochastic in nature, we follow classical statistical notation and introduce
each properly on first
appearance. %

There exists a number of different missingness mechanisms in the literature
(e.g., missing completely at random (MCAR), missing at random (MAR), missing not
at random (MNAR), and, more recently, (un)structured
missingness). Section~\ref{sec:missingness:mechanisms} provides respective
definitions and clarifies what we understand under these missingness
mechanisms. To be able to focus on our main contribution, two existing
approaches to amputation are summarized in Section~\ref{sec:lit}, and
proofs are deferred to Section~\ref{sec:proofs}.

The paper is organised as follows. Section~\ref{sec:not:ass} provides important
aspects for understanding how missingness can be modelled in our stochastic
approach, i.e., through random variables via their joint distribution. It also
introduces the notation and selected assumptions we refer to throughout the
paper. Our stochastic approach to amputation is introduced in
Section~\ref{sec:main}. Section~\ref{sec:mtcars} demonstrates its feasibility
and capability in capturing missingness patterns in terms of the publicly
available \R\ dataset \texttt{mtcars} which allows for visual identification of
each missing data point, as well as in terms of the publicly available \R\ dataset \texttt{SP500\_const}
to study the effect on risk measure estimates. Section~\ref{sec:concl} provides concluding remarks.

\section{Understanding missingness, notation, and initial assumptions}\label{sec:not:ass}
\subsection{Probabilistic setup}\label{sec:prob:setup}
Consider the random matrix
$Y=(Y_{i,j})_{i=1,\dots,n,\ j=1,\dots,d}\in\IR^{n\times d}$ %
with $d$-dimensional rows $\bm{Y}_{i,}=(Y_{i,1},\dots,Y_{i,d})$, $i=1,\dots,n$,
and $n$-dimensional columns $\bm{Y}_{,j}=(Y_{1,j},\dots,Y_{n,j})$,
$j=1,\dots,d$; %
note that our results do not depend on the form of $Y$, e.g., $Y$ could also be a
single long vector $\bm{Y}=(Y_1,\dots,Y_{nd})$ as sometimes found in the missingness literature. %
The random matrix $Y$ modelling $n$ observations of dimension $d$ is
considered \emph{complete} in that its realisations produce datasets
without missing values.  A realisation $y$ of $Y$ would ideally be the
starting point of any statistical analysis. Yet, instead of
$y$, a statistician often only observes $x\in\IR_{\ast}^{n\times d}$ for
$\IR_{\ast}=\IR\cup\{\ast\}$, where the symbol $\ast$ is arbitrarily chosen to indicate a missing value;
in statistical software this is typically NA (``not
available''). Thus, there exists a matrix
$m=(\I_{\{x_{i,j}=\ast\}})_{i=1,\dots,n,\ j=1,\dots,d}$ that indicates
missingness in the matrix $x$. In terms of realisations, this describes an
entirely deterministic setup (\textit{ex post}, after realisations
are obtained). To understand how missingness arises and how we can
model it, we need a stochastic setup (\textit{ex ante},
before realisations); confusing the \textit{ex ante} with
the \textit{ex post} has led to ambiguous definitions in the
missingness literature, as pointed out by \cite[Section~3]{Seaman2013}
with various examples.

Given $Y$, we imagine there is a random matrix $M\in\{0,1\}^{n\times d}$, which
indicates with values $1$ the components of (the unobservable, complete) $Y$
that will be missing. The resulting random matrix $X$ has components
$X_{i,j}=\ast$ if $M_{i,j}=1$ and $X_{i,j}=Y_{i,j}$ if $M_{i,j}=0$.
With canonical operations $\ast\cdot 0=0$, $\ast+0=\ast$, and
$\ast\cdot 1=\ast$, we have $X=\ast\,M + Y\odot (1-M)$, where $\odot$ is the
Hadamard product. $M$ is called \emph{missingness indicator matrix}; see also
\cite[p.~8]{LittleRubin2020}. Note that, instead, \cite{meallirubin2015} consider the
\emph{response indicator matrix} or \emph{availability indicator matrix}
$R=1-M = (1-M_{i,j})_{i=1,\dots,n,\ j=1,\dots,d}$. \cite{Rubin1976} and
\cite{Seaman2013} consider the response/availability indicator matrix, but
denote it by $M$ and refer to it as missingness indicator matrix. We indicate
and model missingness (rather than availability) as that is the goal in
amputation, but availability can equally well be modelled \emph{mutatis mutandis}. %

Our main goal is to present a stochastic approach for producing $X$ with realistic
realisations of missingness patterns, including classical ones as well as structured missingness (SM)
such as \ref{SM1} and \ref{SM2}; see Section~\ref{app:SM:def} for these concepts. This is achieved
by constructing conditional distributions $F_{M|Y}$ and thus realisations of
$M$ given $Y$; note that the probability mass function $f_{M|Y}(m\,|\,y)$ of $M$ given
$Y=y$ at $m$ corresponds to $g_{\bm{\phi}}(1-m\,|\,u)$ in the notation of
\cite{Rubin1976} and \cite{Seaman2013}, their $\bm{\phi}$ being a parameter we
will denote by $\bm{\theta}$, and their $u$ representing our $y$.  To include
the case of (subsets of) random variables in $M$ being $0$ or $1$
deterministically (with probability in $\{0,1\}$, so capturing \ref{SM2}), we
allow distributions of entries of $M$ to be \emph{degenerate} (unit jumps).

For each $i=1,\dots,n$ and $j=1,\dots,d$, we know that
$M_{i,j}\sim\B(1,p_{i,j})$, i.e., each entry of $M$ is Bernoulli distributed
with success (so missingness) probability $p_{i,j}\in[0,1]$.  Associated with
the random matrix $M$ is the deterministic matrix
$P=(p_{i,j})_{i,j}\in[0,1]^{n\times d}$ of \emph{marginal missingness
  probabilities}. %
In this context, ``marginal'' refers to the marginal distribution of a single
entry $M_{i,j}$ as part of the random matrix $M$, e.g., the entry with index
pair $(i,j)$ has a (marginal) Bernoulli distribution, characterized by its
distribution function
$F_{i,j}(x)=(1-p_{i,j})\I_{[0,\infty)}(x)+p_{i,j}\I_{[1,\infty)}(x)$, $x\in\IR$,
or the corresponding quantile function
$F_{i,j}^{-1}(u)=\inf\{x\in\IR:F_{i,j}(x)\ge u\}=\I_{(1-p_{i,j},1]}(u)$,
$u\in(0,1]$. The quantile function $F_{i,j}^{-1}$ of $F_{i,j}$ can be used to
construct the stochastic representation $F_{i,j}^{-1}(U_{i,j})$ of $F_{i,j}$
(i.e., $F_{i,j}^{-1}(U_{i,j})\sim F_{i,j}$) which we will use for constructing
stochastic representations for $M$; this result is also known as
\emph{inversion method} for sampling univariate distributions and frequently
applied in statistical software.

The goal of modelling $M$ given $Y$ is aligned with the goal of the
\emph{amputer} (the person responsible for removing values), but
understanding the construction principle of $M$ given $Y$ is also
important for the \emph{imputer} (the person responsible for imputing
values) in order to understand how missingness patterns arise in case
no such knowledge is available otherwise.

\subsection{The treated as observed and treated as missing
  parts}\label{sec:obs:parts}
For amputation, it is convenient to think of the
components of the complete $Y$ and the resulting $X$ to be grouped in two
parts. To this end, let $I=\{1,\dots,n\}\times\{1,\dots,d\}$, and let
$I_Y^{\text{mis}},I_Y^{\text{obs}}\subseteq I$ be a partition of $I$
($I_Y^{\text{mis}}\cup I_Y^{\text{obs}}=I$,
$I_Y^{\text{mis}}\cap I_Y^{\text{obs}}=\emptyset$), %
where $I_Y^{\text{mis}}$ ($I_Y^{\text{obs}}$) denotes the set of indices of those entries of $Y$ to be
treated as missing (observed). We then have the
\emph{treated as missing part}
$Y^{\text{mis}}=(Y_{i,j}^{\text{mis}})_{i=1,\dots,n,\ j=1,\dots,d}$ with
$Y_{i,j}^{\text{mis}}=Y_{i,j}$ if $(i,j)\in I_Y^{\text{mis}}$ and $0$
otherwise, %
and the \emph{treated as observed part}
$Y^{\text{obs}}=(Y_{i,j}^{\text{obs}})_{i=1,\dots,n,\ j=1,\dots,d}$ with
$Y_{i,j}^{\text{obs}}=Y_{i,j}$ if $(i,j)\in I_Y^{\text{obs}}$ and $0$ otherwise.
This gives the decomposition $Y=Y^{\text{mis}}+Y^{\text{obs}}$. Based on
$Y^{\text{mis}}$ and $Y^{\text{obs}}$, the amputer constructs $M$, from which
$X$ results.

Similarly, $X$ can be divided into two parts. With
$I_X^{\text{mis}}=\{(i,j)\in I:X_{i,j}=\ast\}$ and
$I_X^{\text{obs}}=\{(i,j)\in I:X_{i,j}\neq\ast\}$
we obtain a partition of $I$ and thus of $X$ into the
\emph{missing part}
$X^{\text{mis}}=(X_{i,j}^{\text{mis}})_{i=1,\dots,n,\ j=1,\dots,d}$ with
$X_{i,j}^{\text{mis}}=X_{i,j}=\ast$ if $(i,j)\in I_X^{\text{mis}}$ and $0$
otherwise, %
and the \emph{observed part}
$X^{\text{obs}}=(X_{i,j}^{\text{obs}})_{i=1,\dots,n,\ j=1,\dots,d}$ with
$X_{i,j}^{\text{obs}}=X_{i,j}=Y_{i,j}$ if $(i,j)\in I_X^{\text{obs}}$ and $0$
otherwise.
Identifying $0$'s from missing data rather than non-missing $Y$ values
that are $0$ in $X^{\text{obs}}$ can be done with $M$.

In contrast to $Y^{\text{mis}}$ and $Y^{\text{obs}}$, we can also define
$X^{\text{mis}}$ and $X^{\text{obs}}$ depending on $M$ via
$I_X^{\text{mis}}=\{(i,j)\in I:M_{i,j}=1\}$ and
$I_X^{\text{obs}}=\{(i,j)\in I:M_{i,j}=0\}$.
However, $Y^{\text{mis}}$ and $Y^{\text{obs}}$ do not depend on
$M$. %
This is important so as to avoid ambiguous %
definitions that the missing data literature is plagued with; see again
the various examples in \cite[Section~3]{Seaman2013}, partly also addressed in Section~\ref{sec:missingness:mechanisms}, in particular Remark~\ref{rem:everywhere:vs:realised}. %

\begin{remark}[Amputation vs.\ imputation]
  In imputation, even though the imputer thinks of $Y$ as complete, only a
  fraction of values of $Y$ (the non-$\ast$ entries in $X$) are available. In
  contrast, in amputation, the amputer has access to all of
  $Y$, %
  via either its distribution (specified through a distribution function, density, stochastic
  representation, characteristic function, etc.) or a realisation $y$ of
  $Y$. Thinking of $Y$ as having missing entries $\ast$ is confusing at best.
  Missing components of $Y$ would force distributional restrictions on $M$
  which, when combined with $Y$ to form $X$, may result in $X$ not compatible
  with the target missingness pattern of the amputer. For example, thinking of
  $Y_{1,1}$ as almost surely (i.e., with probability 1) missing would force
  $M_{1,1}=1$ almost surely, because if $\P(M_{1,1}=1)<1$, then
  $\P(X_{1,1}\neq\ast)>0$, which contradicts the target $\P(X_{1,1}=\ast)=1$.

  Figure~\ref{fig:concepts} summarises the
  conceptual framework for amputation and imputation; we use lowercase letters
  for the imputation framework as the imputer has one realisation $x$ of $X$
  and, thus $m$ of $M$, available, which typically leads to one (or more)
  estimate(s) $\hat{y}$ of $Y$ by (multiple)
  imputation. %
  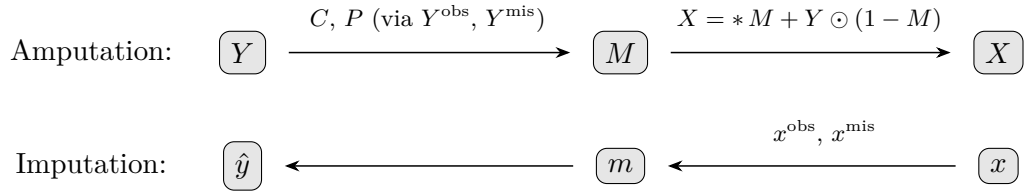
\begin{figure}[htbp]
    \centering
    \begin{tikzpicture}[
      mynode/.style={draw, rectangle, rounded corners, inner sep=1.5mm, fill=gray!20},
      myarrow/.style={->, >={Stealth}, shorten <= 3mm, shorten >= 3mm, line width=0.2mm}
      ]
      \node at (-7,0) {Amputation:};
      \node [mynode] (Y)  at (-5, 0) {$Y$};
      \node [mynode] (M)  at ( 0, 0) {$M$};
      \node [mynode] (X)  at ( 5, 0) {$X$};
      \draw [myarrow] (Y.east)--(M.west);
      \draw [myarrow] (M.east)--(X.west);
      \node [label=above:{\footnotesize $C$, $P$ (via $Y^{\text{obs}}$, $Y^{\text{mis}}$)}] (YobsYmis) at ($ (Y)!0.49!(M) $) {};
      \node [label=above:{\footnotesize $X = \ast\,M + Y\odot(1-M)$}] (one) at ($ (M)!0.5!(X) $) {};
      \node at (-7,-1.5) {Imputation:};
      \node [mynode] (y)  at (-5, -1.5) {$\hat{y}$};
      \node [mynode] (m)  at ( 0, -1.5) {$m$};
      \node [mynode] (x)  at ( 5, -1.5) {$x$};
      \draw [myarrow] (x.west)--(m.east);
      \draw [myarrow] (m.west)--(y.east);
      \node [label=above:{\footnotesize $x^{\text{obs}}$, $x^{\text{mis}}$}] (xobsxmis) at ($ (m)!0.54!(x) $) {};
    \end{tikzpicture}
    \caption{Conceptual framework for amputation and imputation.}
    \label{fig:concepts}
  \end{figure}
  The function $C$ for capturing dependence between missingness
  indicators will be introduced in
  Section~\ref{sec:main:basics}. Note that our work focuses exclusively on
  its contributions to amputation (not imputation).
\end{remark}

\subsection{Initial assumptions}\label{sec:init:ass}
To generate distributions of $M\,|\,Y=y$, as amputers we frequently refer to one
of the following two assumptions (for the exact
assumptions of our approach, see Section~\ref{sec:main}):
\begin{enumerate}[label=(A\arabic*), labelwidth=\widthof{(A2)}]
\item\label{A1} The rows $\bm{M}_{1,},\dots, \bm{M}_{n,}$ of $M$ are
  independent given $Y$.
\item\label{A2} The rows $\bm{M}_{1,},\dots,\bm{M}_{n,}$ of $M$ are
  iid given $Y$.
\end{enumerate}
Under~\ref{A1}, the amputer can still freely choose the entries in the matrix
$P\in[0,1]^{n\times d}$ containing the marginal probabilities of missingness,
whereas, under~\ref{A2}, the rows of $P$
must be equal as vectors, so componentwise, in order to not violate the
assumption of rows being identically distributed. In other words, whilst
under~\ref{A1} $P$ is arbitrary, under~\ref{A2} each row of $P$ must be equal to
$\bm{p}$ for some fixed
$\bm{p}=(p_1,\dots,p_d)\in[0,1]^d$.

Note that~\ref{A2} %
allows to incorporate SM. Under \ref{A2}, in each column, all entries are equally
distributed and one essentially only needs to model one row $\bm{M}$ of
$M$ (a joint distribution of a random vector $\bm{M}$ with $d$ components). In
this case, these $d$ components are referred to as the $d$ margins (or marginal
distributions) of $\bm{M}$ or $M$. If at least one margin of $M$ is degenerate (one
column with index $j$), all $M_{1,j},\dots,M_{n,j}$ indicate missingness (if
$p_j=1$) or all entries indicate no missingness (if $p_j=0$). This will motivate
the introduction of Assumption~\ref{A3} in-between \ref{A1} and \ref{A2} in
Section~\ref{sec:AssumptionsOnDependenceAndMargins}.

So far, our considerations only concerned the marginal distributions
$\B(1,p_{i,j})$, $i=1,\dots,n$, $j=1,\dots,d$, of $M$. In
Section~\ref{sec:main}, we will consider \emph{dependence} amongst
these marginal distributions to obtain a stochastic approach for
constructing missingness indicator matrices $M$ that does not
require specifying specific missingness patterns.

\section{Copula-based multivariate amputation}\label{sec:main}

\subsection{Specifying distributions for random vectors via copulas}\label{sec:main:basics}
To create joint distributions for generating missingness indicator matrices $M$,
we utilise a decomposition of any such distribution into its margins (or
marginal distribution functions) and a \emph{copula}, i.e., a distribution
function with $\U(0,1)$ margins; see \cite{schweizersklar1983} or
\cite{nelsen2006}. %
This result is known as \emph{Sklar's theorem}, see \cite{Sklar1959} or
\cite{sklar1996}, and it consists of a decomposition (the first) and a
composition (the second) part. By the first part, any $d$-dimensional
distribution function $F$ allows for the decomposition
\begin{align}
  F(\bm{x})=C(F_1(x_1),\dots,F_d(x_d)),\quad\bm{x}=(x_1,\dots,x_d)\in\IR^d,\label{eq:sklar}
\end{align}
for some copula $C$ and the margins $F_1,\dots,F_d$ of $F$.  If
$\ran F_j=\{F_j(x_j):x_j\in\IR\}$ denotes the range of $F_j$, $C$ is uniquely
determined on $\prod_{j=1}^d\ran F_j$. We see from~\eqref{eq:sklar} that $C$ is
the function that combines the margins $F_1,\dots,F_d$ to the joint distribution
function $F$, so it contains the information about the dependence between the
margins. By the second part, combining any copula $C$ with margins
$F_1,\dots,F_d$ as in~\eqref{eq:sklar} leads to a proper multivariate
distribution function $F$. Starting from any copula $C$, we can thus combine it
with Bernoulli margins (in the appropriate stochastic way via compositions with
Bernoulli quantile functions) to obtain a joint distribution function $F$ with
such margins. %
This idea lies a the core of our approach to amputation, we will combine copulas
with Bernoulli margins to create a stochastic description of missingness patterns for $M$.

We now list important elementary copula constructions needed later; see
\cite[Sections~7.1.2, 7.4]{mcneilfreyembrechts2015} for several of these models.
\begin{example}[Examples of copulas]\label{ex:copulas}
  \begin{enumerate}
  \item
  The \emph{independence copula} $C^{\Pi}(\bm{u})=\prod_{j=1}^d u_j$ is the
  copula of independent continuous random variables. It has stochastic
  representation $(U_1,\dots,U_d)\sim C^{\Pi}$ for $U_1,\ldots,U_d\isim\U(0,1)$.

  \item
  The \emph{comonotone copula} $C^{\text{M}}(\bm{u})=\min_{j=1,\dots,d}\{u_j\}$
  is the copula of comonotone continuous random variables. It has stochastic
  representation $(U,\dots,U)\sim C^{\text{M}}$ for $U\sim\U(0,1)$.

  \item
  For $d=2$, the \emph{countermonotone copula}
  $C^{\text{W}}(u_1,u_2)=\max\{u_1+u_2-1,0\}$ is the copula of
  countermonotone continuous random variables. It has stochastic
  representation $(U,1-U)\sim C^{\text{W}}$ for $U\sim\U(0,1)$.

  \item
  The \emph{Gauss copula}
  $C^{\text{Ga}}(\bm{u})=\Phi_{P^{\text{Ga}}}(\Phi^{-1}(u_1),\dots,\Phi^{-1}(u_d))$
  is the copula of the $\N_d(\bm{0},P^{\text{Ga}})$ distribution
  function $\Phi_{P^{\text{Ga}}}$ for a correlation matrix
  $P^{\text{Ga}}$ and for $\Phi$ being the $\N(0,1)$ distribution
  function; the correlation matrix $P^{\text{Ga}}$ acts as parameter
  matrix of $C^{\text{Ga}}$. The Gauss copula has stochastic
  representation $(\Phi(W_1),\dots,\Phi(W_d))\sim C^{\text{Ga}}$ for
  $\bm{W}=(W_1,\dots,W_d)\sim \N_d(\bm{0},P^{\text{Ga}})$.

  \item
  The \emph{survival copula} of a copula $C$ is the copula
  $\check{C}(\bm{u})=\sum_{\bm{i}\in\{0,1\}^d}(-1)^{\sum_{j=1}^di_j}C((1-u_1)^{i_1},\dots,(1-u_d)^{i_d})$,
  with stochastic representation $(1-U_1,\dots,1-U_d)\sim\check{C}$ for
  $\bm{U}=(U_1,\dots,U_d)\sim C$.

  \item
  The \emph{convex combination}
  $C(\bm{u})=\lambda C_1(\bm{u})+(1-\lambda)C_2(\bm{u})$,
  $\bm{u}\in[0,1]^d$, $\lambda\in[0,1]$, of two $d$-dimensional
  copulas $C_1,C_2$ is again a copula. It has stochastic
  representation
  $\bm{U}=\I_{\{V\in[0,\lambda]\}}\bm{U}^{C_1}+\I_{\{V\in(\lambda,1]\}}\bm{U}^{C_2}\sim
  C$, with $V\sim\U(0,1)$ and $\bm{U}^{C_k}\sim C_k$, $k=1,2$. One can
  sample convex combinations by sampling $V\sim\U(0,1)$ and,
  independently of $V$, sample $\bm{U}^{C_1}$ if $V\le\lambda$ and sample
  $\bm{U}^{C_2}$ if $V>\lambda$.
  \end{enumerate}
\end{example}

\subsection{Specifying distributions for $M$ via copulas}\label{sec:joint:distr:M}
For constructing missingness indicator matrices $M$, we can utilise
the second part of Sklar's theorem as explained in the last section. The joint distribution function
is now $nd$-dimensional and given by
\begin{align}
  F(\bm{x})=C(F_{1,1}(x_{1,1}),\dots,F_{1,d}(x_{1,d}),F_{2,1}(x_{2,1}),\dots,F_{n,d}(x_{n,d})),\label{eq:FM:most:general}
\end{align}
for all $\bm{x}=(x_{1,1},\dots,x_{1,d},x_{2,1},\dots,x_{n,d})\in\IR^{nd}$, and
some $nd$-dimensional copula $C$, as well as Bernoulli margins
$F_{i,j}(x_{i,j})=(1-p_{i,j})\I_{[0,\infty)}(x_{i,j})+p_{i,j}\I_{[1,\infty)}(x_{i,j})$,
$x_{i,j}\in\IR$, for $p_{i,j}\in[0,1]$. The margin with index $(i,j)$ identifies
the corresponding entry in $M$, so $F$ is the joint distribution function of
$M$. A stochastic representation of $F$ is
\begin{align}
  (F_{1,1}^{-1}(U_{1,1}),\dots,F_{1,d}^{-1}(U_{1,d}),F_{2,1}^{-1}(U_{2,1}),\dots,F_{n,d}^{-1}(U_{n,d})),\label{eq:stoch:rep:most:general}
\end{align}
where $\bm{U}=(U_{1,1},\dots,U_{1,d},U_{2,1},\dots,U_{n,d})\sim C$ and
$F_{i,j}^{-1}(u_{i,j})=\I_{(1-p_{i,j},1]}(u_{i,j})$, $u_{i,j}\in(0,1]$, $i=1,\dots,n$,
$j=1,\dots,d$. Casting these entries into $M$ such that
$\bm{M}_{i,}=(F_{i,1}^{-1}(U_{i,1}),\dots,$ $F_{i,d}^{-1}(U_{i,d}))$,
$i=1,\dots,n$, leads to the stochastic representation
\begin{align}
  M=(\bm{M}_{1,},\dots,\bm{M}_{n,})\T=\begin{pmatrix}
    F_{1,1}^{-1}(U_{1,1}) & \dots & F_{1,d}^{-1}(U_{1,d})\\
    \vdots & & \vdots \\
    F_{n,1}^{-1}(U_{n,1}) & \dots & F_{n,d}^{-1}(U_{n,d})
  \end{pmatrix},\label{eq:stoch:rep:M}
\end{align}
which can be sampled easily to obtain realisations of $M$.
The following algorithm summarises our copula-based
amputation approach for generating $M$ termed \emph{Bernoulli amputation}.
\begin{algorithm}[Bernoulli amputation]\label{alg:Bern:amp}
  Let $K\in\IN$ denote the number of realised missingness matrices $M$ to be simulated.
  Fix the $nd$-dimensional copula $C$ and its parameters, e.g., a Gauss copula
  $C^{\text{Ga}}$ with correlation parameter matrix $P^{\text{Ga}}$. Fix the marginal missingness
  probabilities $P=(p_{i,j})_{i=1,\dots,n,\ j=1,\dots,d}$.
  \begin{enumerate}
  \item Generate $\bm{U}_k=(U_{k,1,1},\dots,U_{k,1,d},U_{k,2,1},\dots,U_{k,n,d})\sim C$, $k=1,\dots,K$.
  \item Compute and return $M_k$ with $(i,j)$th element
    $M_{k,i,j}=\I_{\{U_{k,i,j}>1-p_{i,j}\}}=\I_{\{1-U_{k,i,j}\le p_{i,j}\}}$, %
    $k=1,\dots,K$, $i=1,\dots,n$, $j=1,\dots,d$.
  \end{enumerate}
\end{algorithm}

We now consider examples that show how $C$ can introduce dependence
between the elements of $M$.
\begin{example}[Amputation examples]\label{ex:cop:based:amputation}
  \begin{enumerate}
  \item \emph{Independent homogeneous amputation}. Suppose
  $p_{i,j}=p\in[0,1]$, $i=1,\dots,n$, $j=1,\dots,d$. If $C=C^{\Pi}$,
  then $M$ has stochastic representation (interpreted as a matrix as in \eqref{eq:stoch:rep:M})
  $(\I_{(1-p,1]}(U_{1,1}),$ $\dots,\I_{(1-p,1]}(U_{1,d}),\I_{(1-p,1]}(U_{2,1}),\dots,\I_{(1-p,1]}(U_{n,d}))$
  for $U_{i,j}\isim \U(0,1)$, $i=1,\dots,n$, $j=1,\dots,d$. $M$ thus
  consists of independent entries, each satisfying $\P(M_{i,j}=0)=1-p$
  and $\P(M_{i,j}=1)=p$. In survey statistics, this corresponds to
  item-nonresponse if participants do not respond to single questions
  independently of other questions (same row) and other participants (other rows).
  \item \emph{Independent comonotone homogeneous amputation}. Suppose
  $p_{i,j}=p\in(0,1)$, $i=1,\dots,n$, $j=1,\dots,d$.  If
  $C(\bm{u})=C^{\Pi}(C^{\text{M}}(\bm{u}_1),\dots,C^{\text{M}}(\bm{u}_n))=\prod_{i=1}^n\min_{j=1,\dots,d}\{u_{i,j}\}$,
  then, independently for all rows $i$,
  $\P(\bm{M}_{i,}=(0,\dots,0))=1-p=1-\P(\bm{M}_{i,}=(1,\dots,1))$, so
  that $M$ has stochastic representation
  $(\I_{(1-p,1]}(U_1),\dots,\I_{(1-p,1]}(U_1),\I_{(1-p,1]}(U_2),\dots,$ $\I_{(1-p,1]}(U_n))$
  for $U_1,\dots,U_n\isim \U(0,1)$, where each $U_i$ appears $d$
  times. $M$ thus contains independent rows, each of which consists of
  $0$'s (with probability $1-p$) or $1$'s (with probability $p$). In
  survey statistics, this corresponds to full unit-response or full
  unit-nonresponse, independently of other participants. $C^{\Pi}$ can
  also be replaced by any other copula in this construction to model
  more realistic unit-response/non-response patterns.
  \item \emph{Comonotone set homogeneous amputation}. Let
  $S\subseteq\{1,\dots,n\}\times\{1,\dots,d\}$ be a set of indices of
  $M$ and $p_{i,j}=p\in[0,1]$ for all $(i,j)\in S$. If the marginal
  copula of $U_{i,j}$, $(i,j)\in S$ (the copula obtained by setting
  all components with indices $(i,j)\notin S$ to $1$) is the
  comonotone copula, then
  $\P(M_{i,j}=0\ (i,j)\in S)=1-p=1-\P(M_{i,j}=1\ (i,j)\in S)$, irrespective of those entries in $M$
  with indices $(i,j)\notin S$. This is a case of~\ref{SM1}, and for
  $p\in\{0,1\}$ it is an example of~\ref{SM2}; see Section~\ref{app:SM:def}. A fixed (non-mixed)
  monotone missingness pattern as in
  Definition~\ref{def:monotone:miss} also falls under the latter setup %
  and can serve as an example from survey statistics for such
  missingness patterns.
  \item \emph{Grouped comonotone set homogeneous amputation}. This is a
  generalisation of the last example to multiple sets $S$.  Let
  $S_1,\dots,S_K\subseteq\{1,\dots,n\}\times\{1,\dots,d\}$ be disjoint
  sets, with $U_{i,j}$, $(i,j)\in S_k$, being jointly distributed
  according to the comonotone
  copula. %
  For each $S_k$, let $p_k\in[0,1]$ denote the corresponding
  homogeneous marginal probability of missingness for all components
  with indices $(i,j)\in S_k$. Besides the behaviour within each set
  $S_k$ described in the last example, we can further specify the
  dependence between the $K$ groups by any $K$-dimensional copula
  $C_{S_1,\dots,S_K}$; as before, this would still need to be
  compatible with the copula $C$ on all $nd$ components to lead to a
  proper model ($C_{S_1,\dots,S_K}$ would need to be the marginal
  copula of $C$ for any $K$ components with indices belonging to
  distinct $S_1,\dots,S_K$); see
    Section~\ref{sec:vis:SM} for an illustration. An example from
    survey statistics arises for $K=2$ with
    $S_1=\{(i,j_1):i=1,\dots,n_1\}$ and $S_2=\{(i,j_2):i=n_1+1,\dots,n\}$
    for $j_1\neq j_2$ and $p_1=p_2=1$, where the first $n_1$ units are
    male and do not respond to question $j_1$ about pregnancy, and the
    remaining units are female and do not respond to question $j_2$ on
    prostate cancer, say.
  \end{enumerate}
\end{example}

The following proposition allows one to compute the probability of joint
missingness for any collection of components of $M$, under any
dependence $C$ and any marginal Bernoulli distributions; see Section~\ref{sec:proofs} for its proof.
To this end, note that the $C$-volume $\Delta_{(\bm{a},\bm{b}]}C$ is the probability $\P(\bm{U}\in(\bm{a},\bm{b}])$
for $\bm{U}\sim C$ and $\bm{0}\le\bm{a}\le\bm{b}\le\bm{1}$. %
\begin{proposition}[Joint missingness probabilities]\label{prop:joint:missingness:prob}
  Let $S\subseteq\{1,\dots,n\}\times\{1,\dots,d\}$ be any set of pairs of
  indices and $S^c$ its complement in $\{1,\dots,n\}\times\{1,\dots,d\}$.
  \begin{enumerate}
  \item For $m_S=(\I_{(i,j)\in S})\in\{0,1\}^{n\times d}$ and $\bm{p}_S=(p_{i,j})_{(i,j)\in S}$, one has
    \begin{align*}
      \P(M=m_S)=\Delta_{\left(\tb{\bm{1}-\bm{p}_S,\,\bm{1}}{\bm{0},\,\bm{1}-\bm{p}_{S^c}}\right]}C=\Delta_{\left(\tb{\bm{0},\,\bm{p}_{S}}{\bm{p}_{S^c},\,\bm{1}}\right]}\check{C},
    \end{align*}
    where the interval corresponding to the component $(i,j)$ of $C$ in the first volume
    is understood as $(1-p_{i,j},1]$ if $(i,j)\in S$ and as $(0,1-p_{i,j}]$ if
    $(i,j)\in S^c$, and similarly for the second volume. In particular,
    the expected proportion of missingness is
    $\frac{1}{nd}\sum_{S\subseteq\{1,\dots,n\}\times\{1,\dots,d\}}|S|\,\P(M=m_S)$. %
  \item Focusing only on those
    components with indices in $S$, we have
    \begin{align*}
      \P(M_{i,j}=1\ (i,j)\in S)=\check{C}_S(\bm{p}_S),
    \end{align*}
    where $\check{C}_S$ is the
    survival copula of the marginal copula $C_S$ of $C$ corresponding to all
    indices in $S$. In particular, for two entries $S=\{(i_1,j_1),(i_2,j_2)\}$, one has
    $\P(M_{i_1,j_1}=1,M_{i_2,j_2}=1)=-1+p_{i_1,j_1}+p_{i_2,j_2}+C_{(i_1,j_1),(i_2,j_2)}(1-p_{i_1,j_1},1-p_{i_2,j_2})$,
    where $C_{(i_1,j_1),(i_2,j_2)}$ denotes the marginal copula of $C$
    corresponding to $(M_{i_1,j_1},M_{i_2,j_2})$.
  \end{enumerate}
\end{proposition}
If $C_S$ is \emph{radially symmetric} (i.e., $\check{C}_S=C_S$), then
$\P(M_{i,j}=1\ (i,j)\in S)=C_S(\bm{p}_S)$. This applies to the independence
copula, the comonotone copula, the countermonotone copula, the Gauss copula, and
convex combinations of radially symmetric copulas.

Besides missingness probabilities, the amputer may also be interested in the correlation
between two entries of the random missingness matrix $M$, addressed in the following
proposition; see Section~\ref{sec:proofs} for its
proof. %
\begin{proposition}[Pairwise correlations]\label{pro:pairwise:cor}
  For any $(i_1,j_1),(i_2,j_2)\in \{1,\dots,n\}\times\{1,\dots,d\}$, we have
  \begin{align*}
    \rho_{(i_1,j_1),(i_2,j_2)}=\cor(M_{i_1,j_1}, M_{i_2,j_2})=\frac{\check{C}_{(i_1,j_1),(i_2,j_2)}(p_{i_1,j_1},p_{i_2,j_2})-p_{i_1,j_1}p_{i_2,j_2}}{\sqrt{p_{i_1,j_1}(1-p_{i_1,j_1})p_{i_2,j_2}(1-p_{i_2,j_2})}}.
  \end{align*}
  In particular, $\rho_{(i_1,j_1),(i_2,j_2)}\in[\rho^{\text{min}}_{(i_1,j_1),(i_2,j_2)},\ \rho^{\text{max}}_{(i_1,j_1),(i_2,j_2)}]$ for
  \begin{align*}
    \rho^{\text{min}}_{(i_1,j_1),(i_2,j_2)}&=\frac{C^{\text{W}}(p_{i_1,j_1},p_{i_2,j_2})-p_{i_1,j_1}p_{i_2,j_2}}{\sqrt{p_{i_1,j_1}(1-p_{i_1,j_1})p_{i_2,j_2}(1-p_{i_2,j_2})}},\\
    \rho^{\text{max}}_{(i_1,j_1),(i_2,j_2)}&=\frac{C^{\text{M}}(p_{i_1,j_1},p_{i_2,j_2})-p_{i_1,j_1}p_{i_2,j_2}}{\sqrt{p_{i_1,j_1}(1-p_{i_1,j_1})p_{i_2,j_2}(1-p_{i_2,j_2})}}.
  \end{align*}
\end{proposition}
Concerning the correlation bounds given in Proposition~\ref{pro:pairwise:cor},
the correlation coefficient is well known to reach its extreme values $\{-1,1\}$ if and only if
the two underlying random variables are linearly dependent; see
\cite{embrechtsmcneilstraumann2002}. If they are not linearly dependent, the
range of possible pairwise missingness correlations is the compact interval
$[\rho^{\text{min}}_{(i_1,j_1),(i_2,j_2)},\
\rho^{\text{max}}_{(i_1,j_1),(i_2,j_2)}]$ of
Proposition~\ref{pro:pairwise:cor}; %
this limited range is a known deficiency of the correlation coefficient as a
measure of (only) linear dependence.

\subsection{Assumptions on the dependence and margins}\label{sec:AssumptionsOnDependenceAndMargins}
In Section~\ref{sec:init:ass} we already introduced Assumptions~\ref{A1} and
\ref{A2} and addressed what necessary conditions they imply on the marginal
Bernoulli probabilities. %
We now proceed similarly but with focus on the dependence structure.

Under~\ref{A1}, the distribution function $F$ of $M$ is
$F(\bm{x})=\prod_{i=1}^n
C_i(F_{i,1}(x_{i,1}),\dots,F_{i,d}(x_{i,d}))$, with stochastic
representation as in~\eqref{eq:stoch:rep:most:general} for independent
$\bm{U}_{i,}=(U_{i,1},\dots,U_{i,d})\sim C_i$, $i=1,\dots,n$, where
$C_i$ denotes a copula of $\bm{M}_{i,}$, i.e., a copula of the $i$th
row of $M$; the meaning of ``a copula'' in comparison to ``the
copula'' will become clear in Section~\ref{sec:non:unique}.
Proposition~\ref{prop:joint:missingness:prob} implies that
$\P(M_{i,j}=1\ (i,j)\in S)$ splits into a product of probabilities of
the form as $\check{C}_{S}(\bm{p}_S)$ across different rows. In
particular, for all $i_1\neq i_2$, we have
$\P(M_{i_1,j_1}=1, M_{i_2,j_2}=1)=p_{i_1,j_1}p_{i_2,j_2}$ for all
$j_1,j_2\in\{1,\dots,d\}$, and for all $j_1\neq j_2$, we have
$\P(M_{i,j_1}=1, M_{i,j_2}=1)=\check{C}_i(p_{i,j_1},p_{i,j_2})$ for all
$i=1,\dots,n$.

Under~\ref{A2}, we can model $F$ as
$F(\bm{x})=\prod_{i=1}^n C(F_1(x_{i,1}),\dots,F_d(x_{i,d}))$, where
$F_{1,j}=\ldots=F_{n,j}=F_j$ denotes the (equal) $j$th marginal distribution
function and $C$ is a $d$-dimensional copula for each of $M$'s rows. As such, a
stochastic representation for $M$ consists of iid rows, each of which has
stochastic representation $(F_1^{-1}(U_1),\dots,F_d^{-1}(U_d))$ for
$(U_1,\dots,U_d)\sim C$. Similarly as under~\ref{A1}, under~\ref{A2}
Proposition~\ref{prop:joint:missingness:prob} implies that
$\P(M_{i,j}=1\ (i,j)\in S)$ splits into a product of probabilities. In
particular, for all $i_1\neq i_2$,
$\P(M_{i_1,j_1}=1, M_{i_2,j_2}=1)=p_{j_1}p_{j_2}$ for all
$j_1,j_2\in\{1,\dots,d\}$, and for all $j_1\neq j_2$,
$\P(M_{i,j_1}=1, M_{i,j_2}=1)=\check{C}(p_{j_1},p_{j_2})$ for all $i=1,\dots,n$.

In what follows, we mostly work under the following assumption. %
\begin{enumerate}[label=(A3), labelwidth=\widthof{(A3)}]
\item\label{A3} The rows $\bm{M}_{1,},\dots,\bm{M}_{n,}$ of $M$ are
  independent given $Y$, with equal copulas $C_1=\ldots=C_n=C$.
\end{enumerate}
This assumption
lies between
\ref{A1} and \ref{A2} in the sense that
\ref{A2} implies \ref{A3} and \ref{A3} implies \ref{A1}. Under~\ref{A3},
$F(\bm{x})=\prod_{i=1}^n C(F_{i,1}(x_{i,1}),\ldots,F_{i,d}(x_{i,d}))$, and $M$
has stochastic representation $M=(F_{i,j}^{-1}(U_{i,j}))_{i,j}$ for
$\bm{U}_{i,}=(U_{i,1},\dots,U_{i,d})\isim C$, $i=1,\dots,n$.  In contrast
to~\ref{A2}, assumption \ref{A3} can incorporate \ref{SM2} in terms of block SM
beyond the limitation of full columns being missing almost surely. Another
advantage of~\ref{A3} will become clear in
Section~\ref{sec:non:unique}. Similarly as under~\ref{A1} and~\ref{A2},
under~\ref{A3} Proposition~\ref{prop:joint:missingness:prob} implies that
$\P(M_{i,j}=1\ (i,j)\in S)$ splits into a product of probabilities of the form
as $\check{C}_{S}(\bm{p}_S)$ across different rows, now all based on the same
(row) copula $C$. In particular, for all $i_1\neq i_2$, we have
$\P(M_{i_1,j_1}=1, M_{i_2,j_2}=1)=p_{i_1,j_1}p_{i_2,j_2}$ for all
$j_1,j_2\in\{1,\dots,d\}$, and for all $j_1\neq j_2$, we have
$\P(M_{i,j_1}=1, M_{i,j_2}=1)=\check{C}(p_{i,j_1},p_{i,j_2})$ for all
$i=1,\dots,n$.

For simplicity, when covariates enter multivariate distributions
(e.g., in multivariate time series or regression models), a common
assumption is that the covariates only affect the marginal
distributions, not the underlying copula. In
Sections~\ref{sec:MAR:MNAR} and \ref{sec:mtcars}, we also make this
simplifying assumption, which is:
\begin{enumerate}[label=(S), labelwidth=\widthof{(S)}]
\item\label{S} $C$ does not depend on $Y$.
\end{enumerate}
In other words, for the construction of $M$ under~\ref{S}, $Y$ only
enters $P$, not $C$.

Note that Assumptions~\ref{A1}, \ref{A2}, \ref{A3}, \ref{S} are not necessary to
be made in our amputation approach based on~\eqref{eq:stoch:rep:M}, we merely introduce them to help us
understand the types of missingness patterns we can create, how they arise, and
since amputers would typically introduce them out of convenience to avoid having to specify a large
number of missingness patterns.

\subsection{Uniqueness and influence of margins and
  dependence}\label{sec:non:unique}
As mentioned in Section~\ref{sec:prob:setup}, the margins $F_{i,j}$ of $M$ are
always Bernoulli distributed. In particular, $\ran
F_{i,j}=\{0,1-p_{i,j},1\}$. Having discrete margins poses no problem for applying the
second part of Sklar's theorem, so for constructing joint distributions.
However, according to the first part of Sklar's theorem, the
fact that the margins are discrete implies non-uniqueness of the copula
$C$. This plays a role when we ask which copulas
$C$ actually lead to different $F$ and, thus, can generate stochastically
different $M$.
This section sheds light on the non-uniqueness of $C$, which is useful
information for the amputer in terms of choosing $C$.

A rather extreme case of the influence of discrete margins are degenerate
margins. If all marginal missingness probabilities are $0$ ($1$), then
$M=(0)_{i,j}$ ($M=(1)_{i,j}$) almost surely, irrespective of the underlying
copula. This is a good example for a situation in which $C$ has no influence on
$M$ at all. It is in line with the first part of Sklar's theorem, by which the
copula is only uniquely determined on the product of the ranges of the margins,
i.e., $\{0,1\}^{nd}$. As all copulas share the same values on the boundary of
$[0,1]^{nd}$ (by the so-called groundedness property and since they all have
$\U(0,1)$), there is no restriction on the underlying
copula, i.e., all copulas provide valid models for this setup and lead to the same
(degenerate) distribution of $M$.

Under~\ref{A2} with margin $F_j$ being $\B(1,p_j)$, $p_j\in(0,1)$, Sklar's
theorem implies that $C$ is only uniquely defined on $\prod_{j=1}^d\ran F_j$,
where $\ran F_j=\{0,1-p_j,1\}$. As all copulas share the same values on the
boundary of $[0,1]^d$, the only point where $C$ is uniquely defined
under~\ref{A2} is $\bar{\bm{p}}=(1-p_1,\dots,1-p_d)$. In other words, two
different copulas influence $F$ (and therefore $M$) only if they differ in
$\bar{\bm{p}}$. Thus, we can generate the set of all possible $F$ (and therefore
$M$) by choosing any copula family which reaches all $\bar{\bm{p}}$. By the
Fr\'echet--Hoeffding bounds theorem, see \cite{frechet1951},
$C^{\Pi}(\bar{\bm{p}})\le C(\bar{\bm{p}})\le C^{\text{M}}(\bar{\bm{p}})$ for all
positive lower orthant dependent copulas $C$ (i.e.,
$C(\bm{u})\ge C^{\Pi}(\bm{u})$ for all $\bm{u}\in[0,1]^d$). As such, all
possible joint distributions $F$ in~\eqref{eq:sklar} with $\B(1,p_j)$ margins
and positive lower orthant dependent copula $C$ can be obtained by considering
the convex combination
$C(\bm{u})=\lambda C^{\text{M}}(\bm{u})+(1-\lambda) C^{\Pi}(\bm{u})$ of
Example~\ref{ex:copulas}. Equally well, one can consider a \emph{homogeneous
  Gauss copula} $C^{\text{Ga}}_\rho$, i.e., a Gauss copula $C^{\text{Ga}}$ with
homogeneous parameter matrix $P^{\text{Ga}}_\rho=\rho J+(1-\rho)I$, where
$J=(1)\in\IR^{d\times d}$, $I$ is the identity matrix in $\IR^{d\times d}$, and
$\rho\in[0,1]$ (with $\rho=0$ leading to $C^{\text{Ga}}=C^{\Pi}$, and $\rho=1$
leading to $C^{\text{Ga}}=C^{\text{M}}$).

An advantage of~\ref{A3} over~\ref{A2} is that the $d$-dimensional
copula $C$ of each row of $M$ is uniquely defined on
$\prod_{j=1}^d\{0,1-p_{1,j},\dots,1-p_{n,j},1\}$ %
and not just in the single point $\bar{p}=(1-p_1,\dots,1-p_d)$. Under~\ref{A3}, the more unique missingness
probabilities $p_{1,j},\dots,p_{n,j}$ the $j$th column of $P$
contains, the more points there are at which $C$ is uniquely
determined, and so the stronger the influence of the choice of $C$ on
the distribution $F$ of $M$.

The following algorithm adapts Algorithm~\ref{alg:Bern:amp} to~\ref{A3}.
\begin{algorithm}[Bernoulli amputation under~\ref{A3}]\label{alg:Bern:amp:A3}
  Let $K\in\IN$ denote the number of realised missingness matrices $M$ to be simulated.
  Fix the $d$-dimensional copula $C$ and its parameters, e.g., a Gauss copula
  $C^{\text{Ga}}$ with correlation parameter matrix $P^{\text{Ga}}$. Fix the marginal missingness
  probabilities $P=(p_{i,j})_{i=1,\dots,n,\ j=1,\dots,d}$.
  \begin{enumerate}
  \item Generate independent $\bm{U}_{k,i}=(U_{k,i,1},\dots,U_{k,i,d})\sim C$, $k=1,\dots,K$, $i=1,\dots,n$.
  \item Compute and return $M_k$ with $(i,j)$th element
    $M_{k,i,j}=\I_{\{U_{k,i,j}>1-p_{i,j}\}}=\I_{\{1-U_{k,i,j}\le p_{i,j}\}}$, $k=1,\dots,K$, $i=1,\dots,n$, $j=1,\dots,d$.
  \end{enumerate}
\end{algorithm}

\subsection{Generating MCAR, MAR, and MNAR missingness indicator matrices}\label{sec:MAR:MNAR}
We are now in the position to model how $Y$ can enter $M$.
By letting $p_{i,j}$ depend on values in $Y$, we can generate missingness
indicator matrices $M$ that explicitly exhibit MCAR, MAR or MNAR missingness.
To keep things simple, we assume $Y$ does
not enter $C$, so we work under~\ref{S} and that each $p_{i,j}$ only depends on
$\bm{Y}_{i,}=(Y_{i,1},\dots,Y_{i,d})$, i.e., on variables in the same row of
$Y$, which is in line with~\ref{A1}, \ref{A2}, and \ref{A3}.
We then model $p_{i,j}$ via
\begin{align}
  p_{i,j}=h_{i,j}(\bm{Y}_{i,})\label{eq:p:via:h}
\end{align}
for functions $h_{i,j}:\IR\to[0,1]$, $i=1,\dots,n$, $j=1,\dots,d$.
Equation~\ref{eq:p:via:h} only describes a conceptual relationship between
$\bm{Y}_{i,}$ and $p_{i,j}$, not for all $i,j$ does $p_{i,j}$ have to depend on
all variables $Y_{i,1},\dots,Y_{i,d}$ of $\bm{Y}_{i,}$.  In other words,
$h_{i,j}$ could be constant with respect to some (or even all) of
$Y_{i,1},\dots,Y_{i,d}$.  As such, let $I_{i,j}\subseteq\{1,\dots,d\}$ denote
the set of indices of those $Y_{i,1},\dots,Y_{i,d}$ on which $p_{i,j}$, and thus
$M_{i,j}$, depends. If $I_{i,j}=\emptyset$, $i=1,\dots,n$, $j=1,\dots,d$, then
$M$ adheres to MCAR.  If
$\emptyset\neq I_{i,j}\subseteq\{1,\dots,j-1,j+1,\dots,d\}$, $i=1,\dots,n$,
$j=1,\dots,d$, then $M$ adheres to MAR. And if
$\{j\}\subseteq I_{i,j}\subseteq\{1,\dots,d\}$, $i=1,\dots,n$, $j=1,\dots,d$,
then $M$ adheres to MNAR. Note that the missingness mechanism MNAR differs from
MAR only in whether $j\in I_{i,j}$.

To better distinguish various setups, we refer to $|I_{i,j}|$ as the
\emph{degree}. Under MNAR, a particularly challenging setup for imputation is
degree $1$ MNAR with $I_{i,j}=\{j\}$, which we refer to as \emph{self-MNAR} (in
this case, $Y_{i,j}$ and only $Y_{i,j}$ influences its own demise; if
$M_{i,j}=1$, this is hard to detect). If MNAR is of degree at least $2$, we
speak of \emph{group-MNAR}. Example relationships for the marginal probability of missingness $p_{i,j}$
are {\small%
  \begin{align*}%
    p_{i,j}&\ \ \ \ \tmbc{MCAR}{=}{degree $0$}\ \ \ \ c_{i,j}\in[0,1],\\
    p_{i,j}&\ \ \ \ \tmbc{MAR}{=}{degree $d\!-\!1$}\ \ \ \ h_{i,j}(Y_{i,1},\dots,Y_{i,j-1}, Y_{i,j+1},\dots, Y_{i,d})\in[0,1],\\
    p_{i,j}&\ \ \ \ \tmbc{MNAR}{=}{degree $d$}\ \ \ \ h_{i,j}(Y_{i,1},\dots,Y_{i,d})\in[0,1],\\
                   &\ \ \ \ \ \,\vdots\\
    p_{i,j}&\ \ \ \ \tmbc{MNAR}{=}{degree $1$}\ \ \ \ h_{i,j}(Y_{i,j})\in[0,1],
  \end{align*}}%
where the two MNAR cases are group-MNAR (degree $d$) and self-MNAR (degree $1$), respectively.
Finally, the amputer uses $P=(p_{i,j})$ as marginal probabilities of missingness
to enter Bernoulli amputation as in Algorithm~\ref{alg:Bern:amp} or,
under~\ref{A3}, Algorithm~\ref{alg:Bern:amp:A3}, and determine $X$ via $X=\ast\,M + Y\odot(1-M)$
as already described previously.

A non-trivial decision that the amputer has to make is choosing appropriate
functions $h_{i,j}$, $i=1,\dots,n$, $j=1,\dots,d$. A simple MNAR example of
degree $d$ for $Y_{i,j}>0$ is
$h_{i,j}(\bm{Y}_{i,})=c_{i,j}\I_{\{Y_{i,j}=Y_{i,(d)}\}}$ for $c_{i,j}\in[0,1]$,
where $Y_{i,(1)}\le\dots\le Y_{i,(d)}$ denote the order statistics of the
components of $\bm{Y}_{i,}$, which leads to the missingness probability
$p_{i,j}=c_{i,j}$ for the index $j$ belonging to the largest $Y_{i,j}$ among
$Y_{i,1},\dots,Y_{i,d}$, and $0$ (not missing) otherwise. Another idea is that
$h_{i,j}$ could be or depend on $Y_{i,j}/Y_{i,(d)}$.  Whether such a
relationship is suitable depends on the data and amputation problem in mind.

Another choice, considered by \cite{Schouten2018} and \cite{schoutenvink2021}, are
transformed linear functions of the form
  \begin{align}
     p_{i,j}=h_{i,j}(\bm{Y}_{i,})=\logiti\biggl(\beta_{i,j;0}+\sum_{k\in I_{i,j}}\beta_{i,j;k} Y_{i,k}\biggr)=\frac{1}{1+e^{-(\beta_{i,j;0}+\sum_{k\in I_{i,j}}\beta_{i,j;k} Y_{i,k})}},\label{eq:regr}
  \end{align}
  where $\logiti(x)=1/(1+\exp(-x))$, $x\in\IR$, is the standard logistic
  distribution function, so $h_{i,j}$ depends on
  $\bm{\beta}_{i,j}=(\beta_{i,j;0},(\beta_{i,j;k})_{k\in I_{i,j}})$,
  $i=1,\dots,n$,
  $j=1,\dots,d$. %
  The intercept coefficient
  $\beta_{i,j;0}$ allows to adjust the missingness probability $p_{i,j}$
  overall, with larger $\beta_{i,j;0}$ leading to a larger $p_{i,j}$.  Similar
  for the other coefficients, if $Y_{i,j}$ is positive, larger $\beta_{i,j;k}>0$
  lead to larger $p_{i,j}$.  Increasing $|\beta_{i,j;k}|$ increases the
  influence of $Y_{i,k}$ on $p_{i,j}$.

  The following lemma can be helpful for choosing the coefficients in terms of
  the range of the resulting marginal missingness probabilities; see
  Section~\ref{sec:mtcars} for an application of the following result,
  Section~\ref{sec:uni:amp} for possible assumptions to further simplify the
  choice of $\bm{\beta}_{i,j}$ and Section~\ref{sec:proofs} for the proof of the
  following result.
\begin{lemma}[Probability-implied coefficients]\label{lem:ran:beta}
  For a row $i\in\{1,\dots,n\}$, column $j\in\{1,\dots,d\}$, and constants
  $c_{\text{min}}<c_{\text{max}}$, suppose
  $Y_{i,k}\in[c_{\text{min}},c_{\text{max}}]$ almost surely for all
  $k\in I_{i,j}$ for some $I_{i,j}\subseteq\{1,\dots,d\}$. Let $p\in(0,1)$ be a
  target probability of missingness for $p_{i,j}$ and $\eps>0$ such that
  $[p-\eps,p+\eps]\subseteq(0,1)$. If $\beta_{i,j;k}>0$ for all $k\in I_{i,j}$
  and
  \begin{align*}
    \sum_{k\in I_{i,j}}\beta_{i,j;k}=\frac{\logit(p+\eps)-\logit(p-\eps)}{c_{\text{max}}-c_{\text{min}}},\quad
    \beta_{i,j;0}=\logit(p-\eps)-c_{\text{min}}\sum_{k\in I_{i,j}}\beta_{i,j;k},
  \end{align*}
  then $p_{i,j}\in[p-\eps,p+\eps]$. In particular, if
  $\beta_{i,j;k}$ are equal for all $k\in I_{i,j}$, then each
  $\beta_{i,j;k}$ can be chosen as
  $(\logit(p+\eps)-\logit(p-\eps))/(|I_{i,j}|(c_{\text{max}}-c_{\text{min}}))$ to guarantee that
  $p_{i,j}\in[p-\eps,p+\eps]$.
\end{lemma}
Lemma~\ref{lem:ran:beta} can help avoiding pitfalls in the sense
of ending up in extreme cases addressed in
Section~\ref{sec:non:unique}. For example, an amputer may unintentionally choose
a too extreme coefficient, leading to $p_{i,j}\approx 0$ or
$p_{i,j}\approx 1$ and, thus, to (non-)missing blocks (SM), even
though this may not be desired for the amputation task at hand.
The amputer should also keep Section~\ref{sec:non:unique} in mind on the impact
of the choice of marginal missingness probabilities $p_{i,j}$ on the uniqueness
of $C$ and its influence on $M$. As depicted in Figure~\ref{fig:concepts}, a
probabilistic model for $M$ is always the result of a suitable choice of both $C$ and
$P$.

\section{Empirical illustration}\label{sec:mtcars}

\subsection{Dataset \texttt{mtcars01}}
For illustrations of Bernoulli amputation, we consider the publicly
available dataset \texttt{mtcars} from the base \R\ package
\texttt{datasets}. The dataset contains $n=32$ observations (automobiles; the
rows) of $d=11$ variables (the columns) describing aspects of automobile design
and performance. To illustrate and better understand Bernoulli amputation, we
range-transform the $d$ columns (with the function
$h(x_{i,j})=\frac{x_{i,j}-x_{(1),j}}{x_{(n),j}-x_{(1),j}}$, $i=1,\dots,n$,
$j=1,\dots,d$, where $x_{(1),j}\le \dots \le x_{(n),j}$ denote the order
statistics of $x_{1,j},\dots,x_{n,j}$) and sort all rows according to increasing
miles per gallon (\texttt{mpg}; the first column). We refer to this new dataset
with columns in $[0,1]$ as \texttt{mtcars01}. Figure~\ref{fig:mtcars01} shows
this colour-coded dataset, with darker colours corresponding to larger values in
$[0,1]$, next to brief explanations of all variables.
\begin{figure}[htbp]
  \begin{minipage}[t]{0.4\textwidth}
    \vspace{0mm}
    \includegraphics[width=\textwidth]{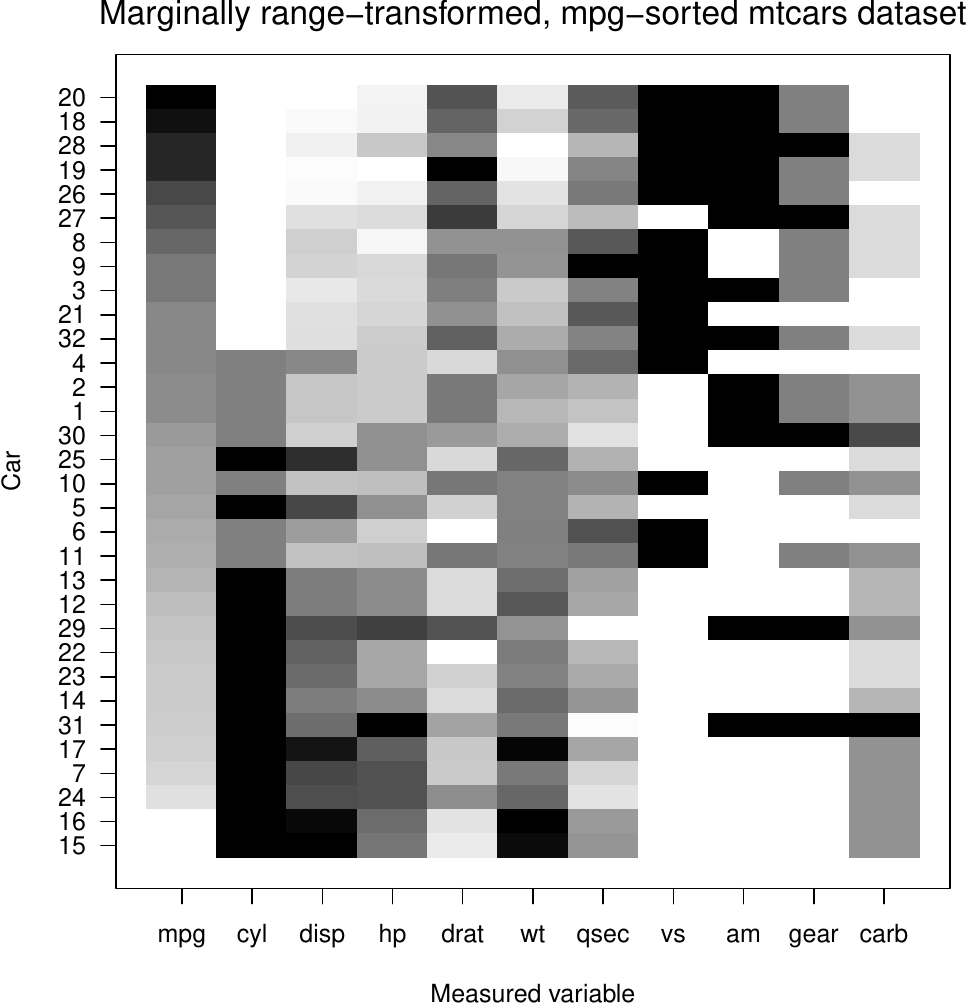}
  \end{minipage}
  \hfill
  \begin{minipage}[t]{0.56\textwidth}
    \vspace{5mm}
    \footnotesize%
    \begin{tabular}{ll}
      \toprule
    Variable & Meaning (range or values)\\
    \midrule
    \texttt{mpg} & Miles per gallon (\texttt{10.4}--\texttt{33.9})\\
    \texttt{cyl} & Number of cylinders (\texttt{4}, \texttt{6}, \texttt{8})\\
    \texttt{disp} & Displacement in $\text{in}^3$ (\texttt{71.1}--\texttt{472.0})\\
    \texttt{hp} & Gross horse power (\texttt{52}--\texttt{335})\\
    \texttt{drat} & Rear axle ratio (\texttt{2.76}--\texttt{4.93})\\
    \texttt{wt} & Weight in lb/1000 (\texttt{1.513}--\texttt{5.424})\\
    \texttt{qsec} & 1/4 mile time in seconds (\texttt{14.5}--\texttt{22.9})\\
    \texttt{vs} & Cyl.\ config (\texttt{0} = V-shaped, \texttt{1} = straight)\\
    \texttt{am} & Transmission (\texttt{0} = autom., \texttt{1} = manual)\\
    \texttt{gear} & Number of forward gears (\texttt{3}, \texttt{4}, \texttt{5})\\
    \texttt{carb} & Number of carburetors (\texttt{1}, \texttt{2}, \texttt{3}, \texttt{4}, \texttt{6}, \texttt{8})\\
    \bottomrule
    \end{tabular}
    \end{minipage}
    \caption{Marginally (per variable) range-transformed \texttt{mtcars} dataset
    sorted according to the first variable \texttt{mpg} (left), and
    explanation of the variables in \texttt{mtcars} (right).}
  \label{fig:mtcars01}
\end{figure}
\begin{remark}[About the data]
  We deliberately chose such a small-dimensional dataset so that single (missing)
  entries can still be visually identified and so that we can showcase a variety
  of missingness patterns Bernoulli amputation is able to capture; for a larger
  dataset in the context of an application, see Section~\ref{sec:app:VaR:ES}. Similarly,
  parameters of all experiments are chosen to allow for visual identification of
  their effect. As our approach only involves random variate generation (sampling)
  via the stochastic representation~\eqref{eq:stoch:rep:M}, it can easily be
  applied with different parameters and to datasets of much larger size and
  dimension. Sampling can easily be accomplished with the \R\ package
  \texttt{copula} of \cite{copula}, which is also used in our experiments in this
  section; see \cite{hofertkojadinovicmaechleryan2018} for more details. Also note
  that Bernoulli amputation applies irrespectively of the type of data $M$ is
  applied to.
\end{remark}

\subsection{Structured missingness}\label{sec:vis:SM}
Our first example demonstrates the ability of Bernoulli amputation to create
SM. The top row of Figure~\ref{fig:mtcars01:smiley} shows three
realisations of the amputed $X$. %
\begin{figure}[htbp]
  \centering
  \includegraphics[width=0.32\textwidth]{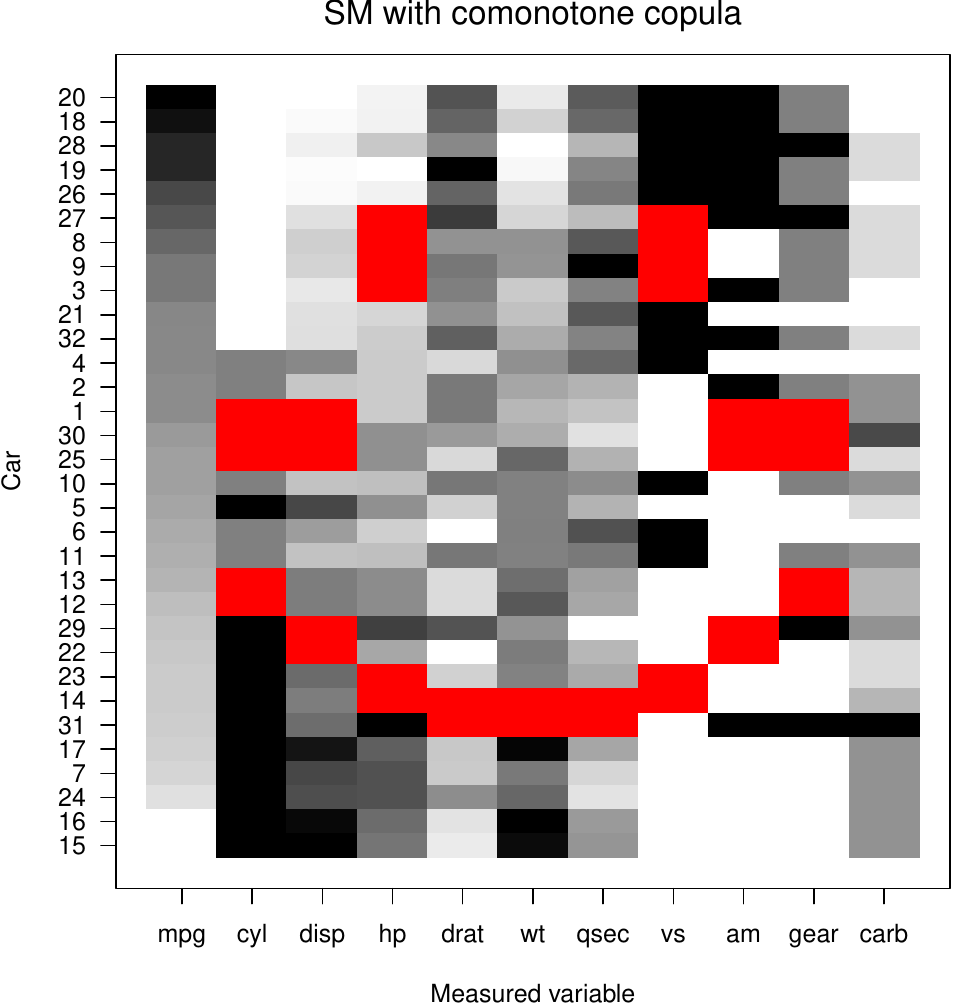}
  \hfill
  \includegraphics[width=0.32\textwidth]{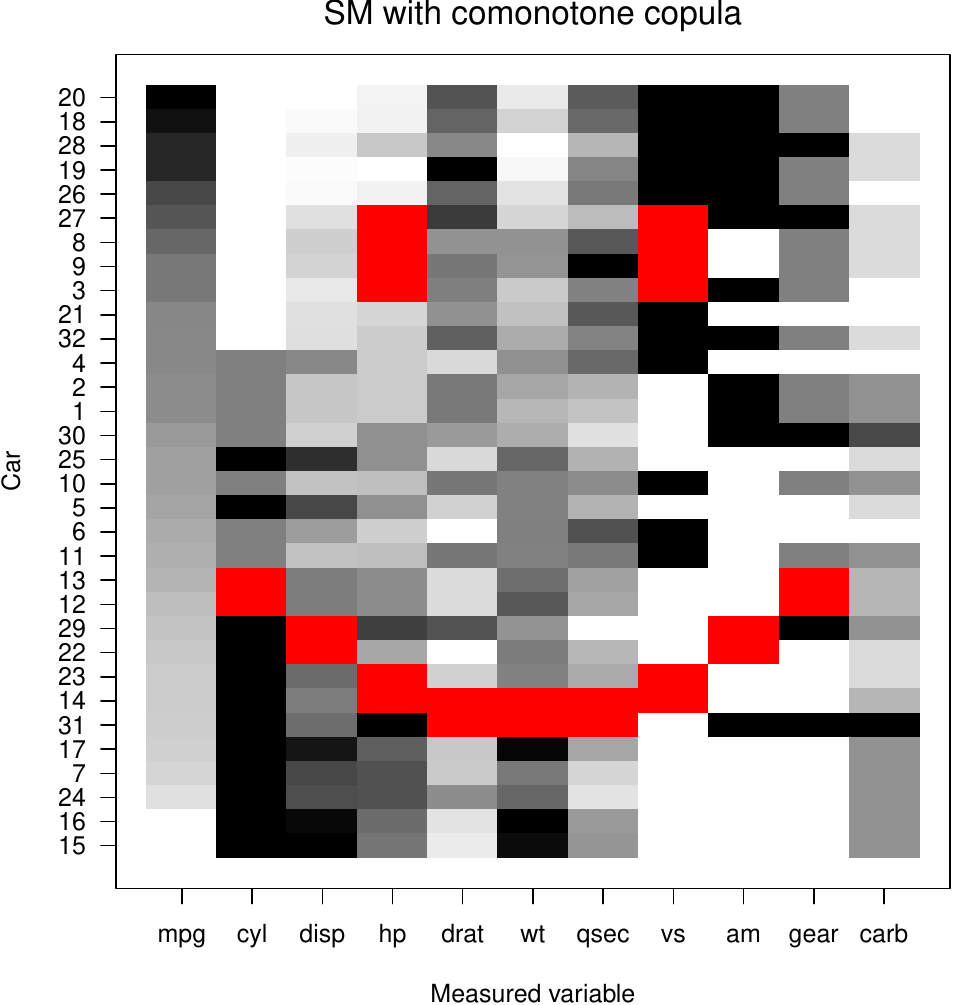}
  \hfill
  \includegraphics[width=0.32\textwidth]{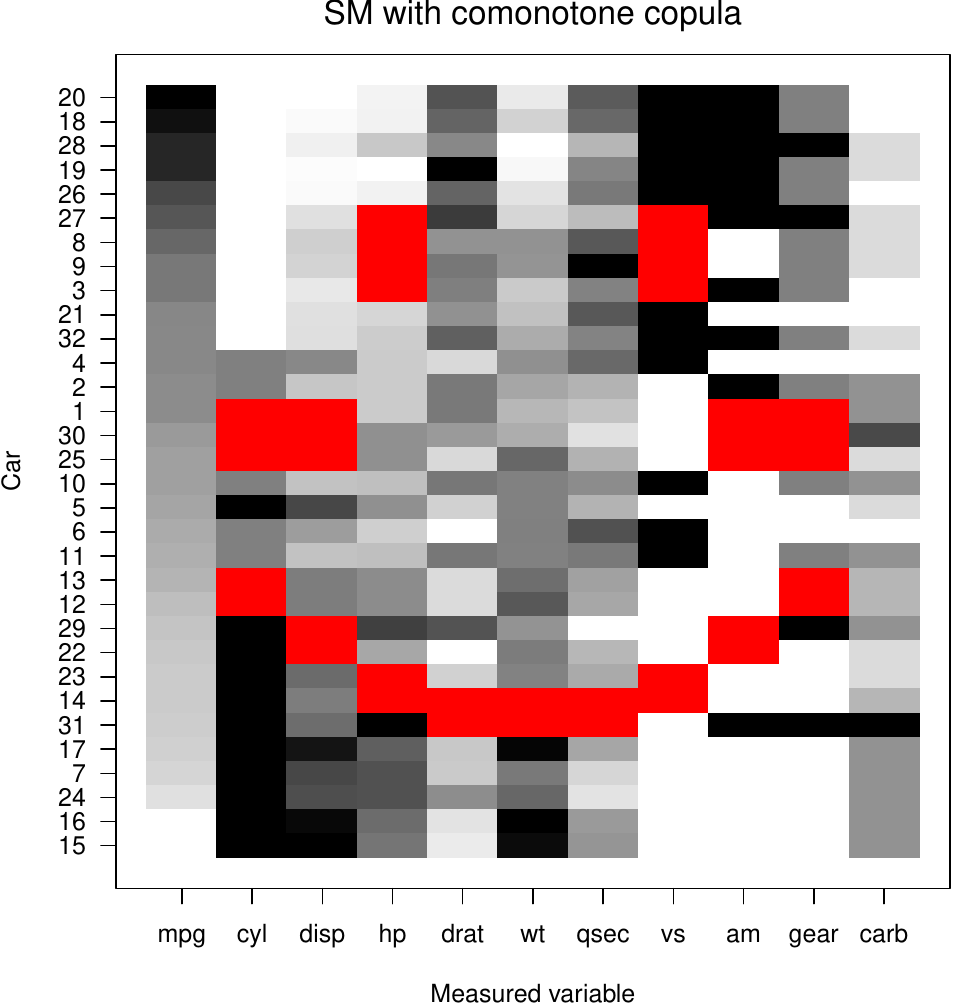}\\[4mm]
  \includegraphics[width=0.32\textwidth]{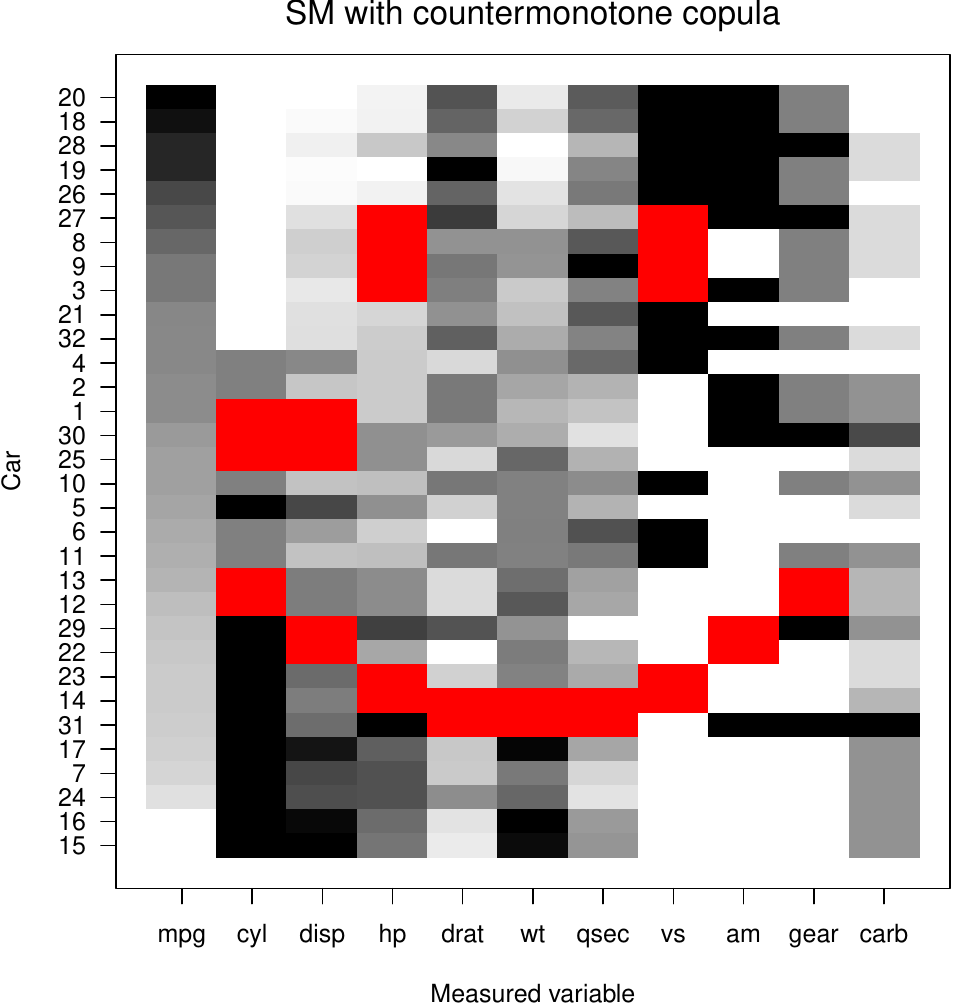}
  \hfill
  \includegraphics[width=0.32\textwidth]{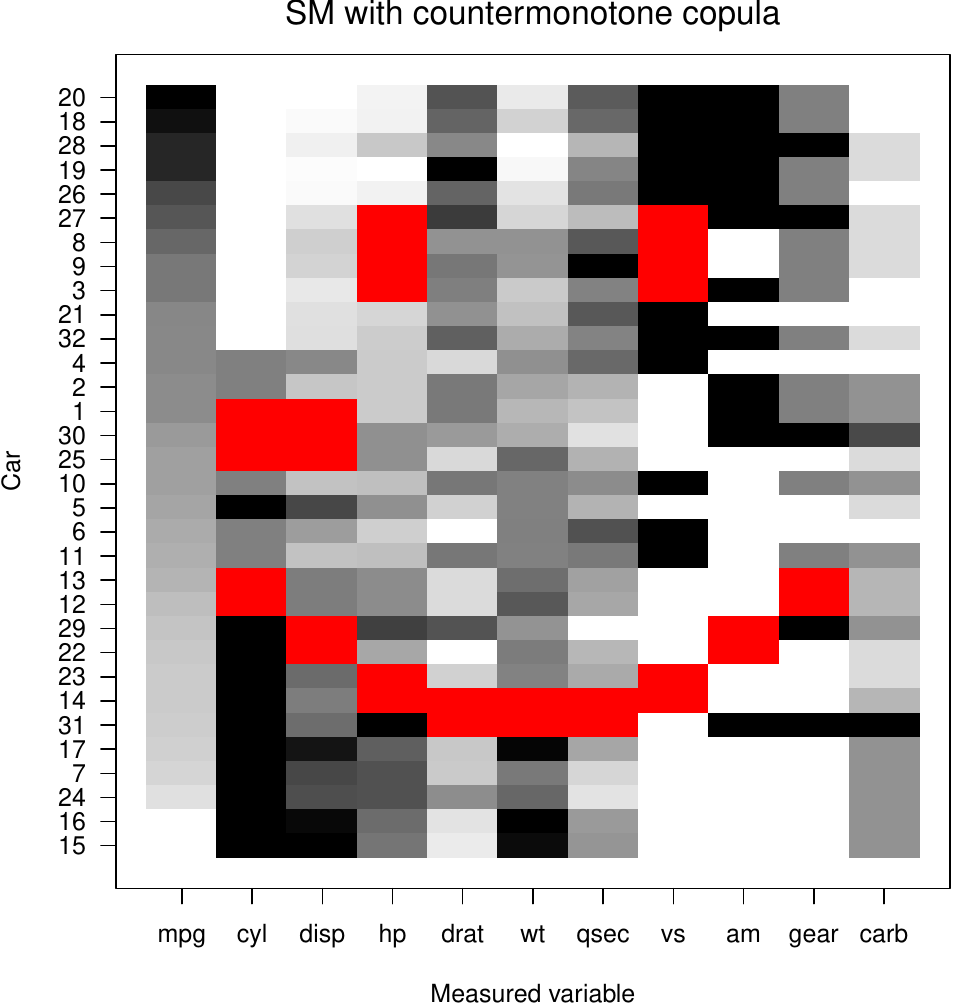}
  \hfill
  \includegraphics[width=0.32\textwidth]{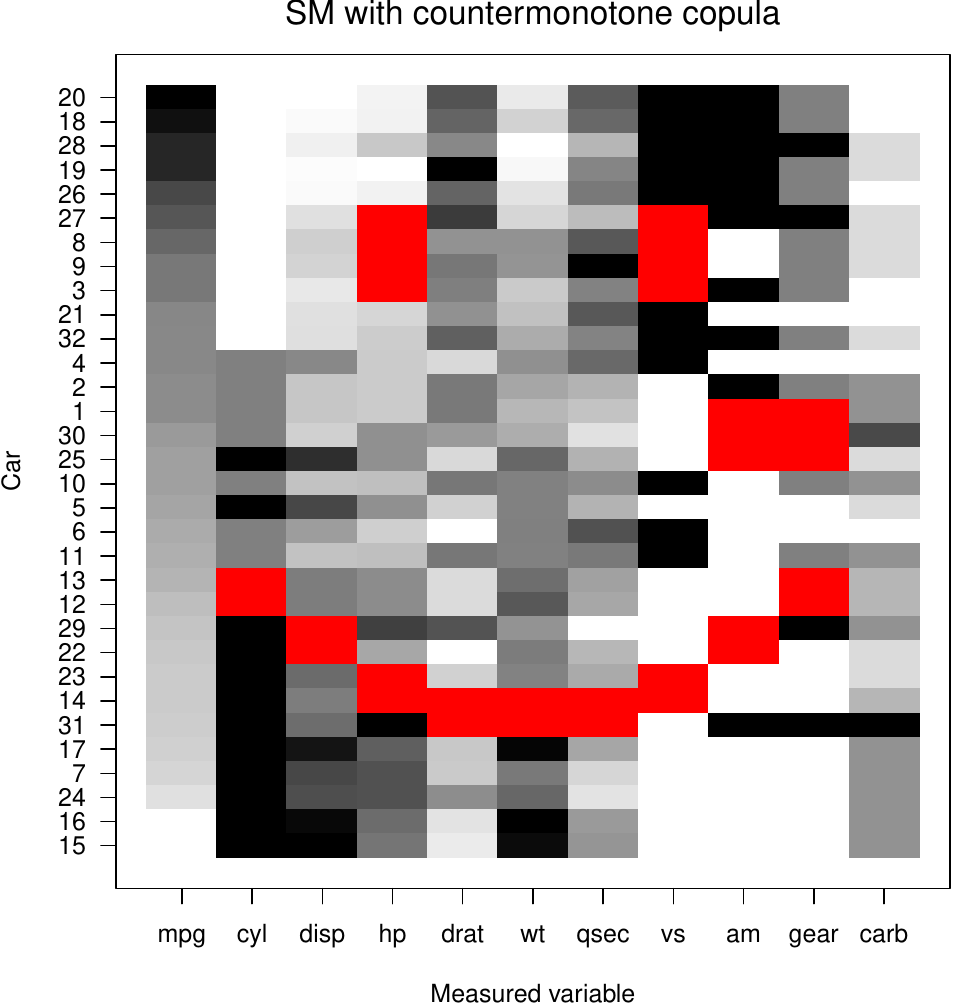}
  \caption{Realisations of SM amputed \texttt{mtcars01} dataset. The entries
    for eyes and mouth of the smiley are always missing ($p_{i,j}=1$),
    those for the blush (red cheeks) are comonotone (top row) or
    countermonotone (bottom row), and all other entries are always
    observed ($p_{i,j}=0$).  Each of the possible two outcomes of the
    blush (top row, both or none; bottom row, precisely one) is chosen
    to appear with probability $p_{i,j}=1/2$.}
  \label{fig:mtcars01:smiley}
\end{figure}
Cells in red indicate missing values. The missingness pattern shows a smiley,
which is blushing (red on both cheeks) with probability $1/2$. This example of
SM falls under the setup of Example~\ref{ex:cop:based:amputation}, with $K=3$
groups. The set $S_1$ contains all indices related to the eyes and mouth, with
corresponding homogeneous marginal probability $p_1=1$ (always missing). The set
$S_2$ contains all indices related to the blush, with corresponding homogeneous
marginal probability $p_2=1/2$ (appearing only in the first and third
realisation shown). The remaining set $S_3$ with $p_3=0$ contains the indices of
all remaining entries (never
missing). %
Any copula with marginal copula corresponding to the indices in $S_2$ being the
comonotone copula (both cheeks either blush or not) can serve as copula $C$
here; see Section~\ref{sec:non:unique}.  In particular, the dependence between
all components that are always missing and all components that are never missing
is irrelevant (as $p_1,p_3\in\{0,1\}$ for all such components), so we can take
$C=C^{\text{M}}$ (in this case,
$C_{S_1,S_2,S_3}(u_1,u_2,u_3)=\min\{u_1,u_2,u_3\}$) or, equally well, $C$ being
a product of $C^{\text{M}}$ and $C^{\Pi}$, where the former contains all
components with indices in $S_2$ and the latter contains all remaining
components (in this case, $C_{S_1,S_2,S_3}(u_1,u_2,u_3)=u_1u_2u_3$). Note how we
only specified distributions, not the missingness patterns themselves, which is
typically much easier to achieve by the amputer and explores the space of
possible missingness patterns much better than a finite set of specific
patterns.

The bottom row of Figure~\ref{fig:mtcars01:smiley} is also based on
Example~\ref{ex:cop:based:amputation}. $M$ is now constructed based on four
groups, namely the facial structure $S_1$ with $p_1=1$ (always missing), the
left cheek $S_2$ with $p_2=1/2$, the right cheek $S_3$ with $p_3=1/2$, and all
other components $S_4$ with $p_4=0$ (never missing). In this case, one can
choose $C_{S_1,S_2,S_3,S_4}(u_1,u_2,u_3,u_4)=W(M(u_1, u_2), M(u_3, u_4))$, which
implies that the two cheeks $S_2,S_3$ are countermonotone, i.e., precisely one
cheek of the smiley shows a blush, and each side appears with probability
$p_2=p_3=1/2$.

Figure~\ref{fig:mtcars01:smiley} also demonstrates that Bernoulli amputation can
not only handle stochastic missingness, but also deterministic missingness
since, by construction, the latter appears as a special case for degenerate
Bernoulli margins, i.e., for $p_{i,j}=0$ (never missing) or $p_{i,j}=1$ (always
missing).

Another form of SM is monotone missingness, shown in
Figure~\ref{fig:mtcars01:monotone}.
\begin{figure}[htbp]
  \centering
  \includegraphics[width=0.32\textwidth]{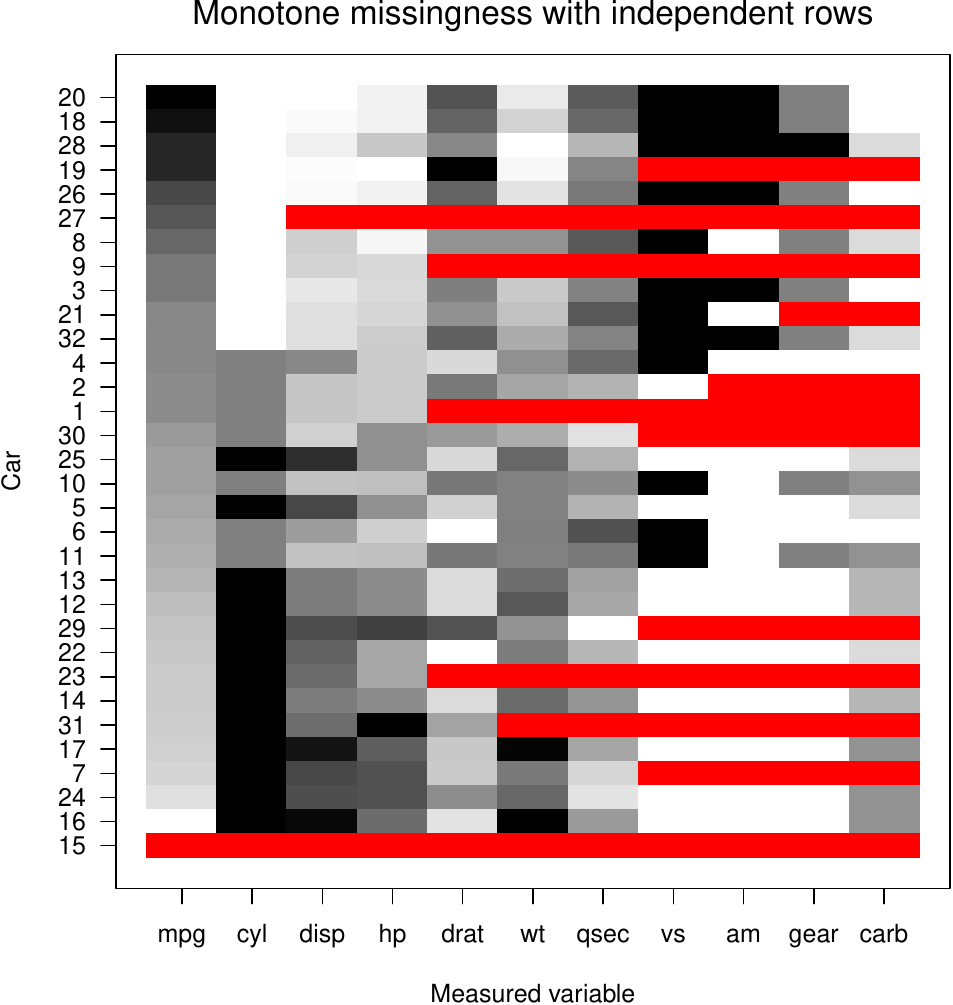}
  \hfill
  \includegraphics[width=0.32\textwidth]{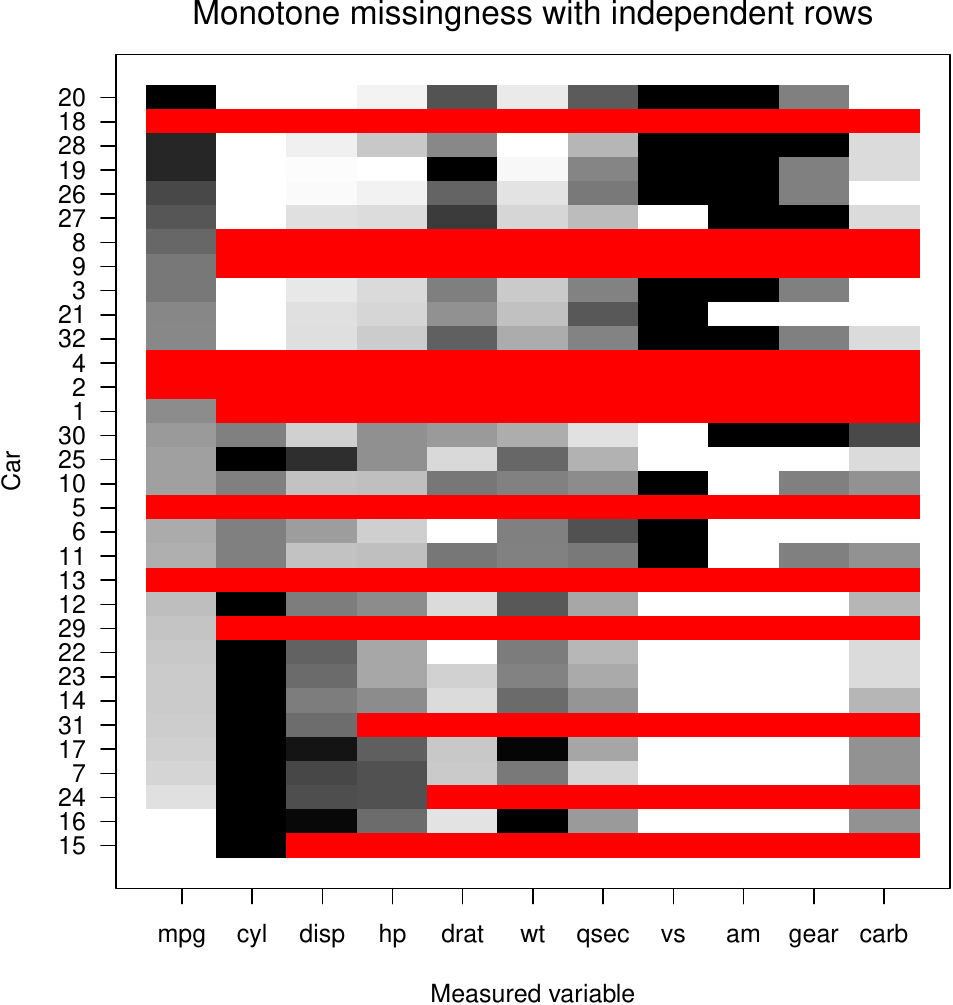}
  \hfill
  \includegraphics[width=0.32\textwidth]{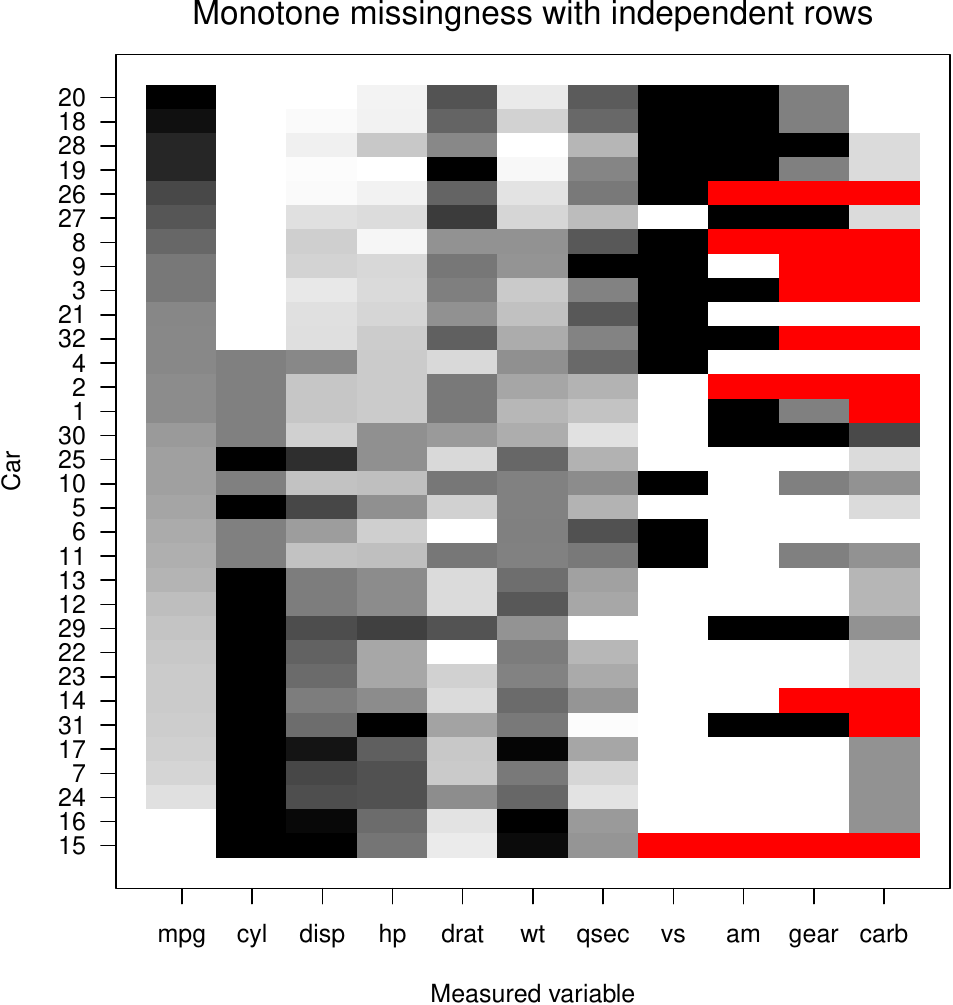}\\[4mm]
  \includegraphics[width=0.32\textwidth]{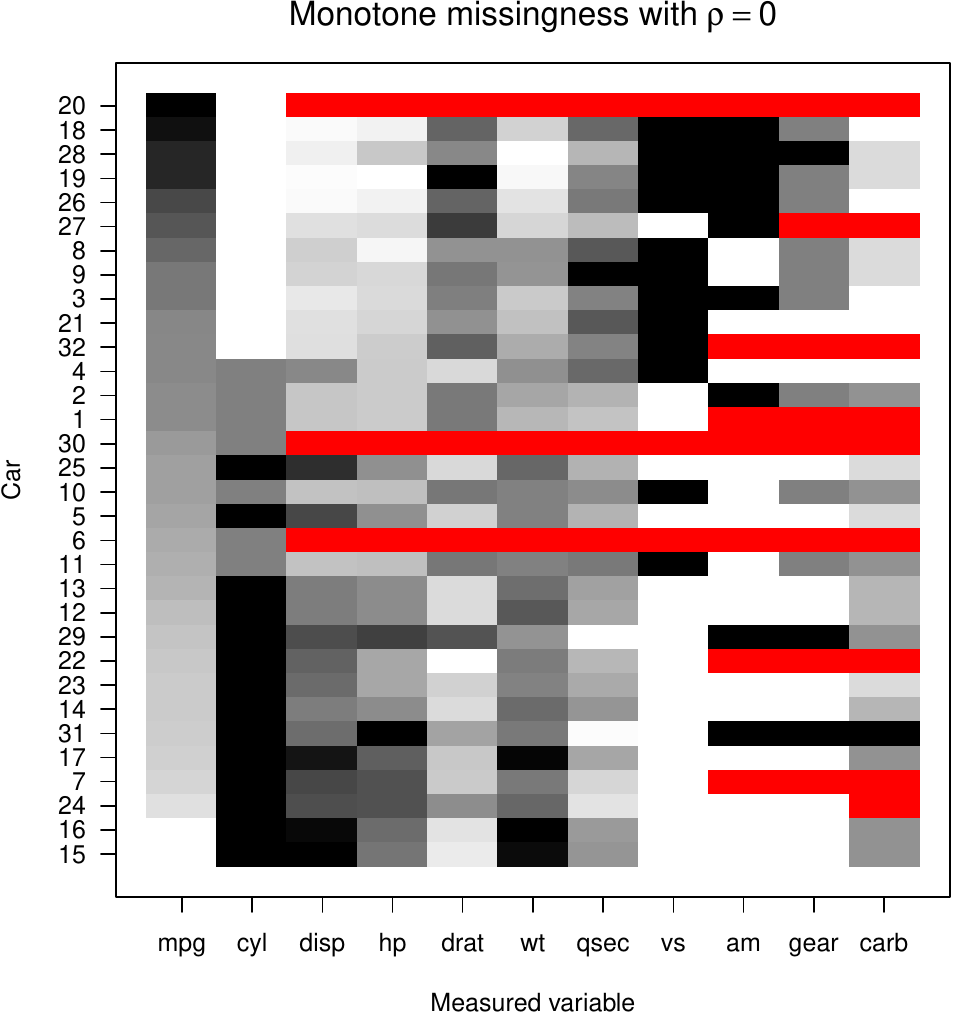}
  \hfill
  \includegraphics[width=0.32\textwidth]{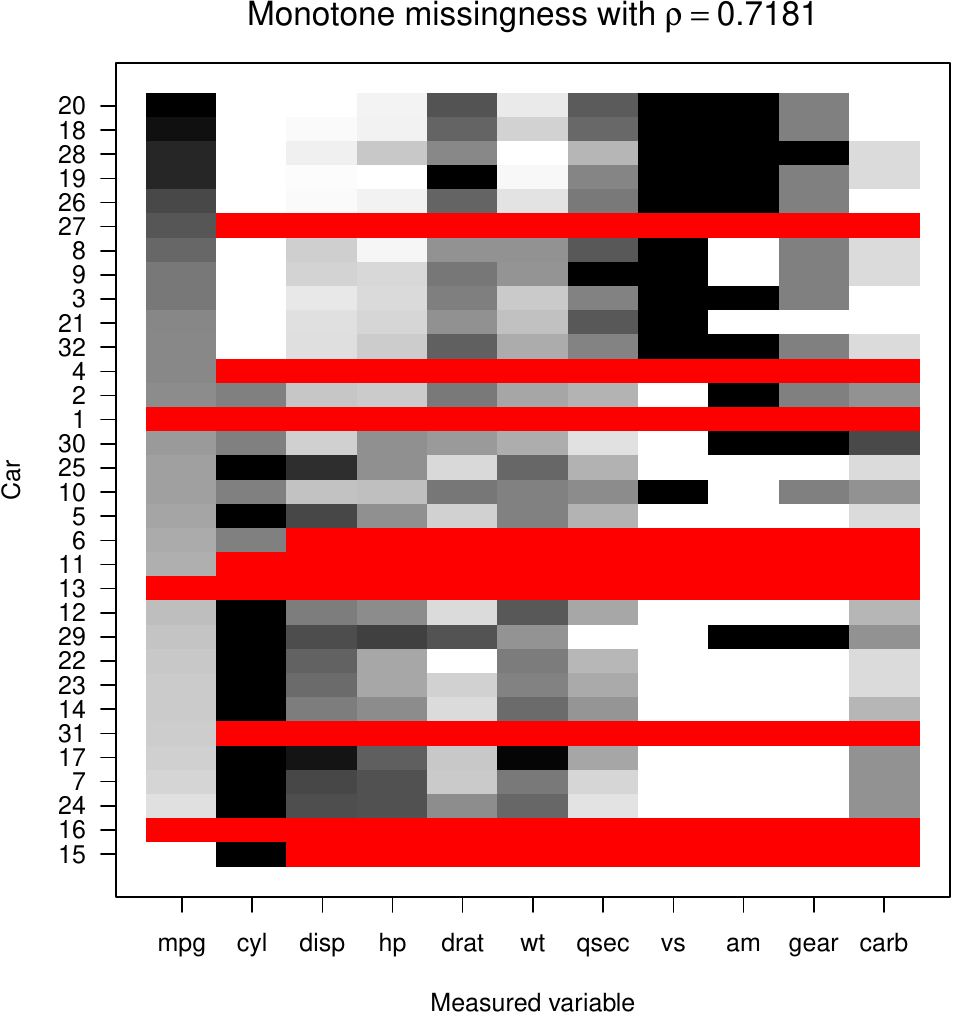}
  \hfill
  \includegraphics[width=0.32\textwidth]{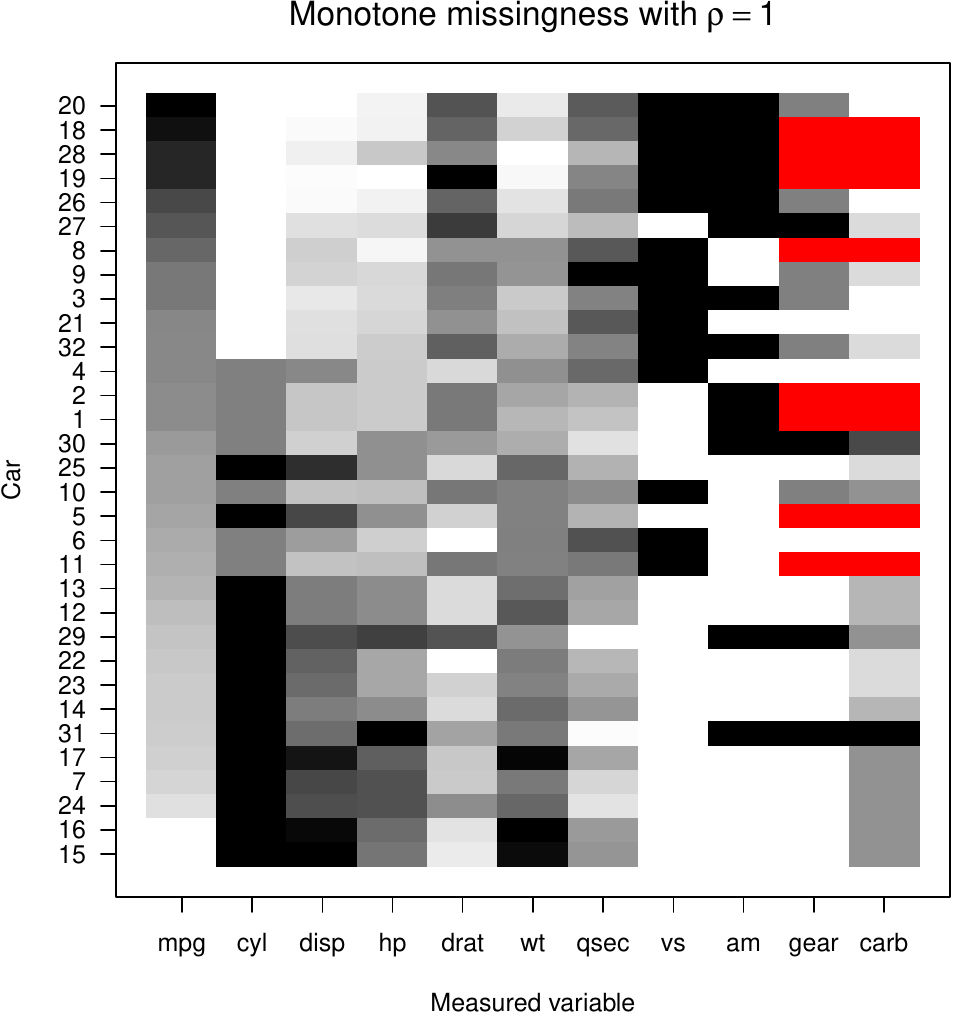}
  \caption{Realisations of monotone-SM amputed \texttt{mtcars01} dataset via
    $M=(\I_{\{j>J_i\}})_{i=1,\dots,n,\ j=1,\dots,d}$ for $J_i=d$ with
    probability $2/3$ (no missingness) and
    $J_i=\lceil d
    F_{\Beta(\alpha,\beta)}^{-1}(U_i)\rceil-1\in\{0,\dots,d-1\}$ with
    probability $1/3$ (missingess), where the parameters of the
    $\Beta(\alpha,\beta)$ distribution are $\alpha=\beta=1$ (the
    $\U(0,1)$ distribution; left), $\alpha=1$, $\beta=4$ (decreasing
    density; center), and $\alpha=4$, $\beta=1$ (increasing density;
    right). In the top row, $U_1,\dots,U_n\isim\U(0,1)$ and in the
    bottom row $(U_1,\dots,U_n)\sim C^{\text{Ga}}_{\rho}$ with
    $\rho=0$ (independence copula; left), $\rho=0.7181$ (center), and
    $\rho=1$ (comonotone copula; right).}
  \label{fig:mtcars01:monotone}
\end{figure}
In both rows we consider realisations of
$M=(\I_{\{j>J_i\}})_{i=1,\dots,n,\ j=1,\dots,d}$ for $J_i=d$ with probability
$2/3$ (leading to no missingness in such rows $i$) and
$J_i=\lceil d F_{\Beta(\alpha,\beta)}^{-1}(U_i)\rceil-1\in\{0,\dots,d-1\}$ with
probability $1/3$ (leading to missingess in such rows $i$), where the parameters
of the $\Beta(\alpha,\beta)$ distribution are $\alpha=\beta=1$ (left),
$\alpha=1$, $\beta=4$ (center), and $\alpha=4$, $\beta=1$ (right). The three
beta distributions have a flat density (standard uniform), decreasing density,
and increasing density, respectively, which is reflected in the random (left),
early (center), and late (right) starting point of the monotone missingness
pattern in each row. In the top row, $U_1,\dots,U_n\isim\U(0,1)$ (independent
rows), whereas in the bottom row, $(U_1,\dots,U_n)\sim C^{\text{Ga}}_{\rho}$
with $\rho=0$ (independence copula; left), $\rho=0.7181$ (center), and $\rho=1$
(comonotone copula; right); the value $\rho=0.7181$ was chosen for comparability
with later examples, see the following section for details. Note that both plots
in the left column are realisations of the same model yet exhibit different
missingness patterns, an advantage of a stochastic approach to amputation as it
allows to focus on principles of missingness patterns rather than specifying
them manually. We see that the dependence between the rows controls how
scattered across the columns the starting points of the monontone missingness
patterns in each row are. In the top row and the bottom left plot, the monotone
missingness patterns start independently of each other across rows, whilst in
the bottom center they tend to be closer together across different rows
(starting early due to the decreasing density of the $\Beta(1,4)$
distribution). In the bottom right they start rather late due to the increasing
density of the $\Beta(4,1)$ distribution, and in the same column due to the
comonotonicity of $(U_1,\dots,U_n)$.

\begin{remark}[Advantage over existing approaches]
  Neither of the examples presented so far can be reproduced with the multivariate
  amputation approach by \cite{Schouten2018}; see also
  Algorithm~\ref{alg:scen:approach1} in Section~\ref{sec:multi:amp}. This is
  obvious for Figure~\ref{fig:mtcars01:smiley} and the bottom row of
  Figure~\ref{fig:mtcars01:monotone} due to the dependence between the rows, which
  contradicts the shuffling of rows in Step~\ref{alg:scen:approach1:permutation}
  of Algorithm~\ref{alg:scen:approach1}. Concerning the top row of
  Figure~\ref{fig:mtcars01:monotone}, consider the first plot. Although such fixed
  patterns (those we actually see as realizations in each row in this plot) can be
  specified in Step~\ref{alg:scen:approach1:patterns} of
  Algorithm~\ref{alg:scen:approach1}, these are the only patterns we would see
  applied to rows in multiple amputation (the patterns are
  fixed), %
  whereas the stochastic nature of Bernoulli amputation allows for different
  patterns to emerge, for example one with $J_i=d-3$ or $J_i=d-1$ currently not
  appearing in the actual realization displayed in the first plot in the top row
  of Figure~\ref{fig:mtcars01:monotone}. The stochasticity of Bernoulli amputation
  (even including deterministic missingness as a special (degenerate) case)
  is a main advantage as it allows for realised missingness patterns the amputer
  has not explicitly specified and thus is more convenient from the amputer's
  perspective.
\end{remark}

\subsection{Missingness mechanisms MCAR, MAR, and MNAR}
We now provide a visual illustration of Algorithm~\ref{alg:Bern:amp:A3},
showing how the parameter $\rho$ of the underlying homogeneous Gauss copula
$C^{\text{Ga}}_\rho$ can be used to produce different strengths of
multivariate missingness in \texttt{mtcars01}.

The plots in Figure~\ref{fig:mtcars01:mcar} show MCAR missingness patterns when
amputing values for $\rho=0$ (left), $\rho=0.7181$ (center; chosen for it
implies, by Proposition~\ref{pro:pairwise:cor}, a correlation of about $0.5$
between the entries in $M$), and $\rho=1$ (right) for the same (homogeneous)
marginal missingness probability of $p=1/3$ (top) and $p=1/5$ (bottom).
\begin{figure}[htbp]
  \centering
  \includegraphics[width=0.32\textwidth]{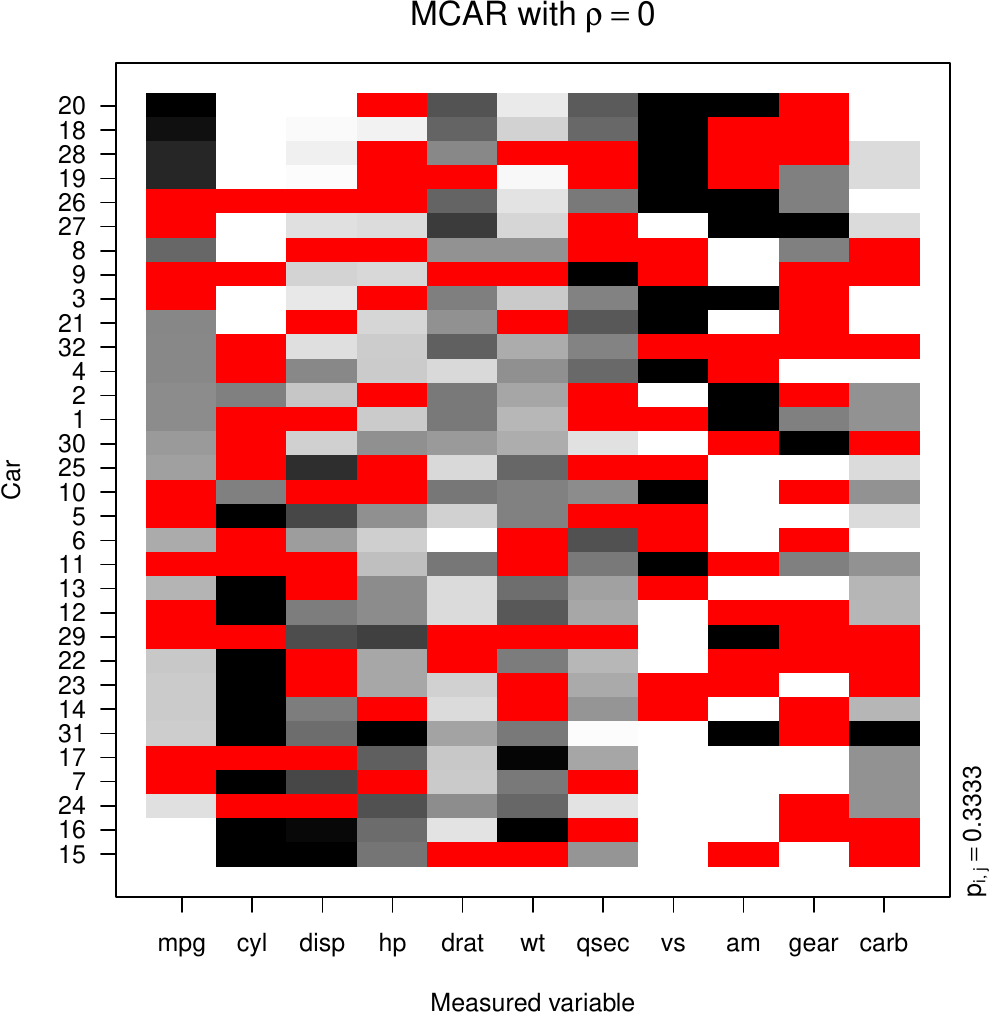}
  \hfill
  \includegraphics[width=0.32\textwidth]{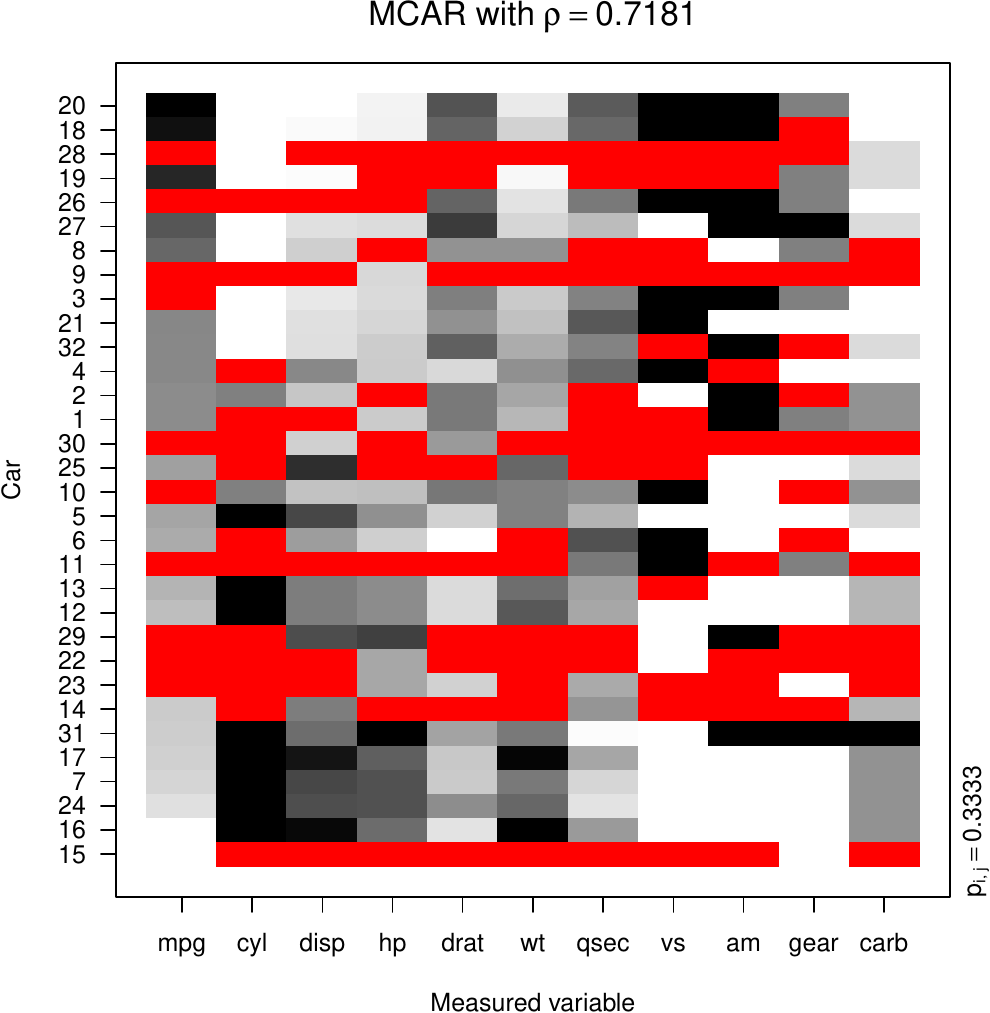}
  \hfill
  \includegraphics[width=0.32\textwidth]{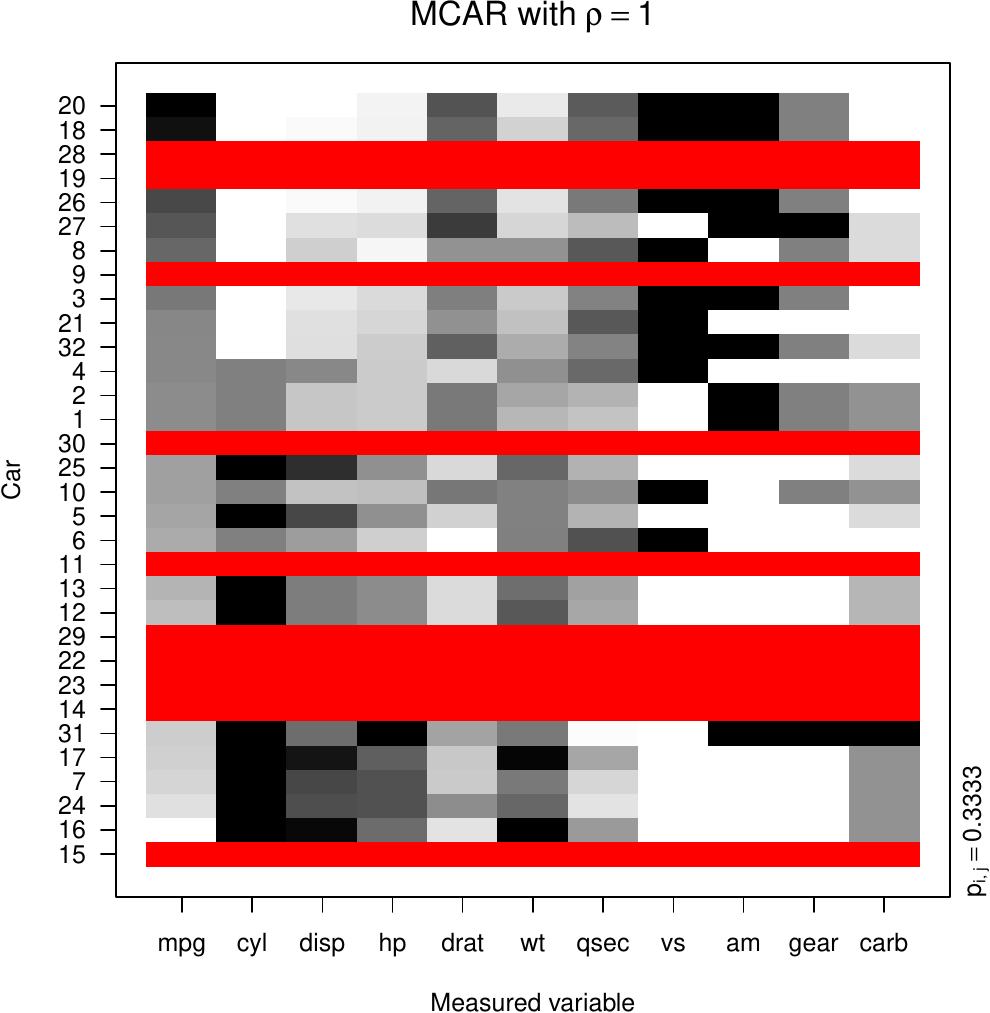}\\[4mm]
  \includegraphics[width=0.32\textwidth]{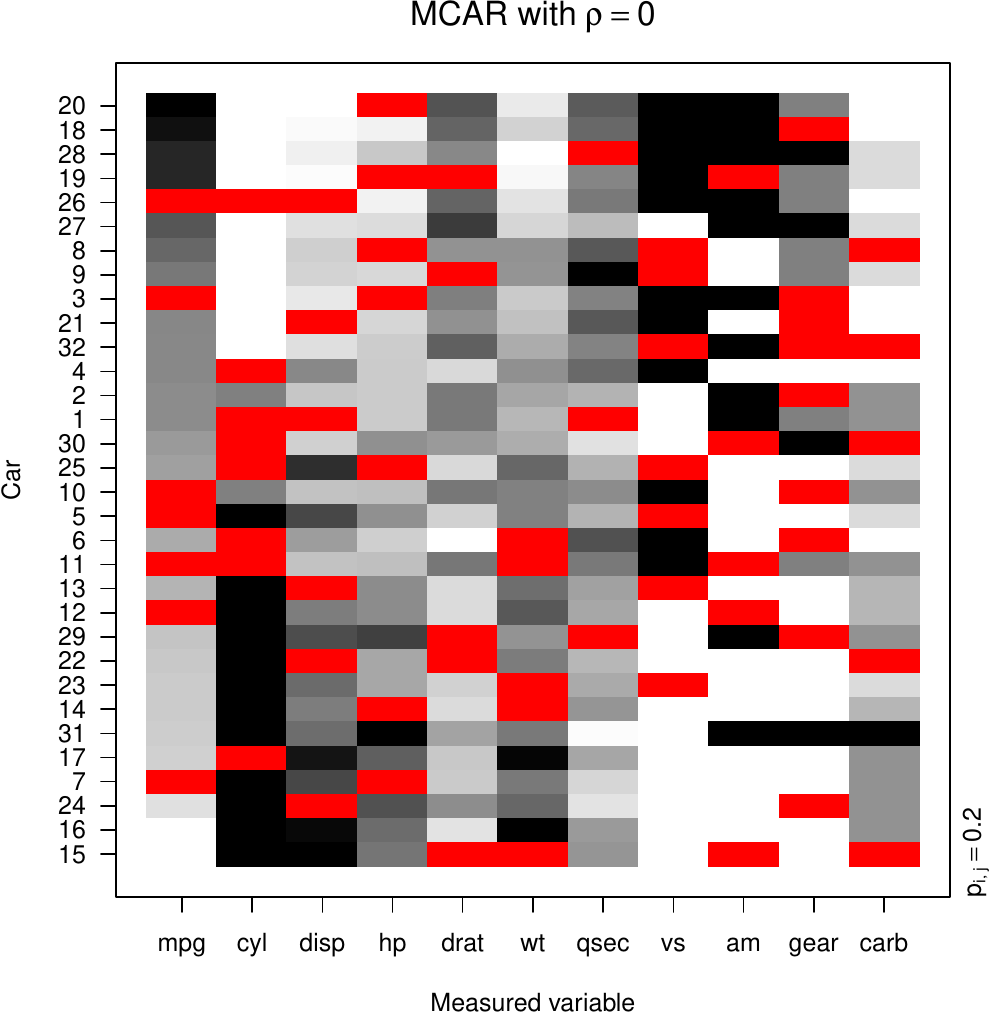}
  \hfill
  \includegraphics[width=0.32\textwidth]{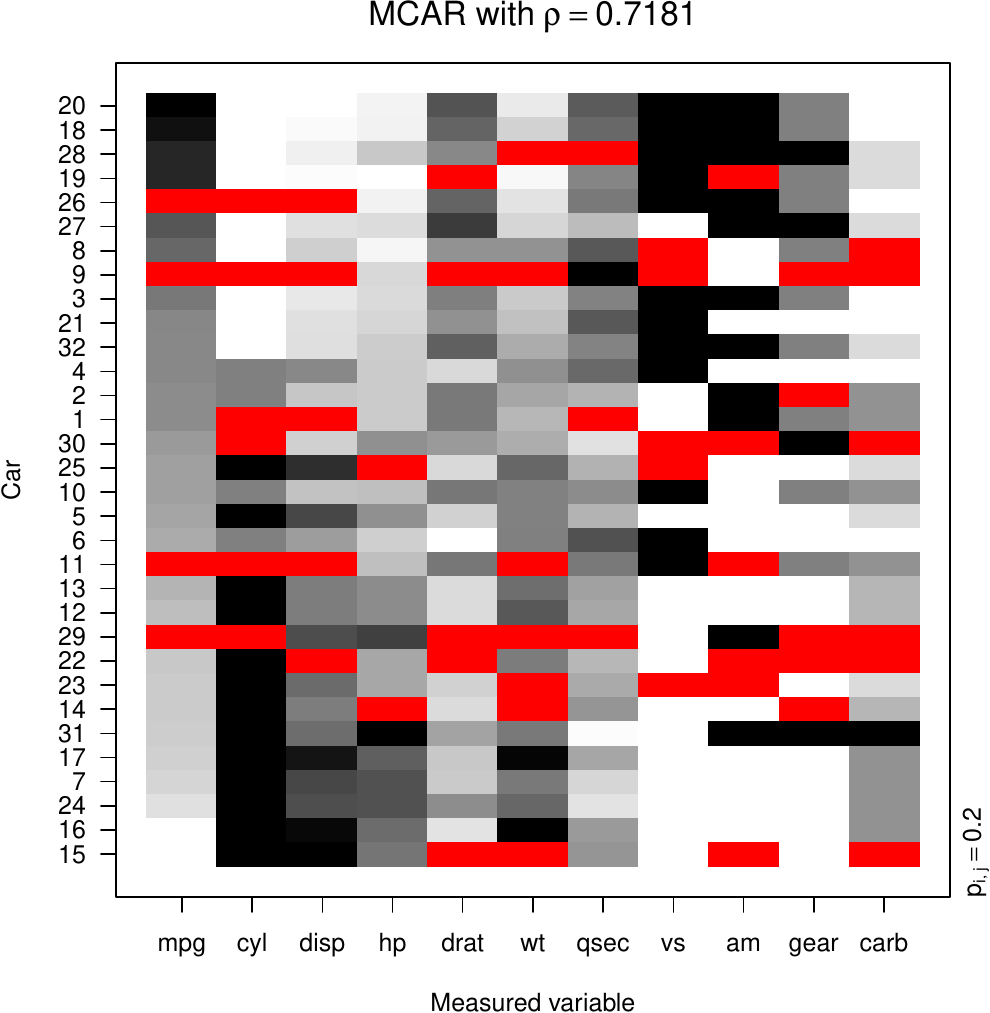}
  \hfill
  \includegraphics[width=0.32\textwidth]{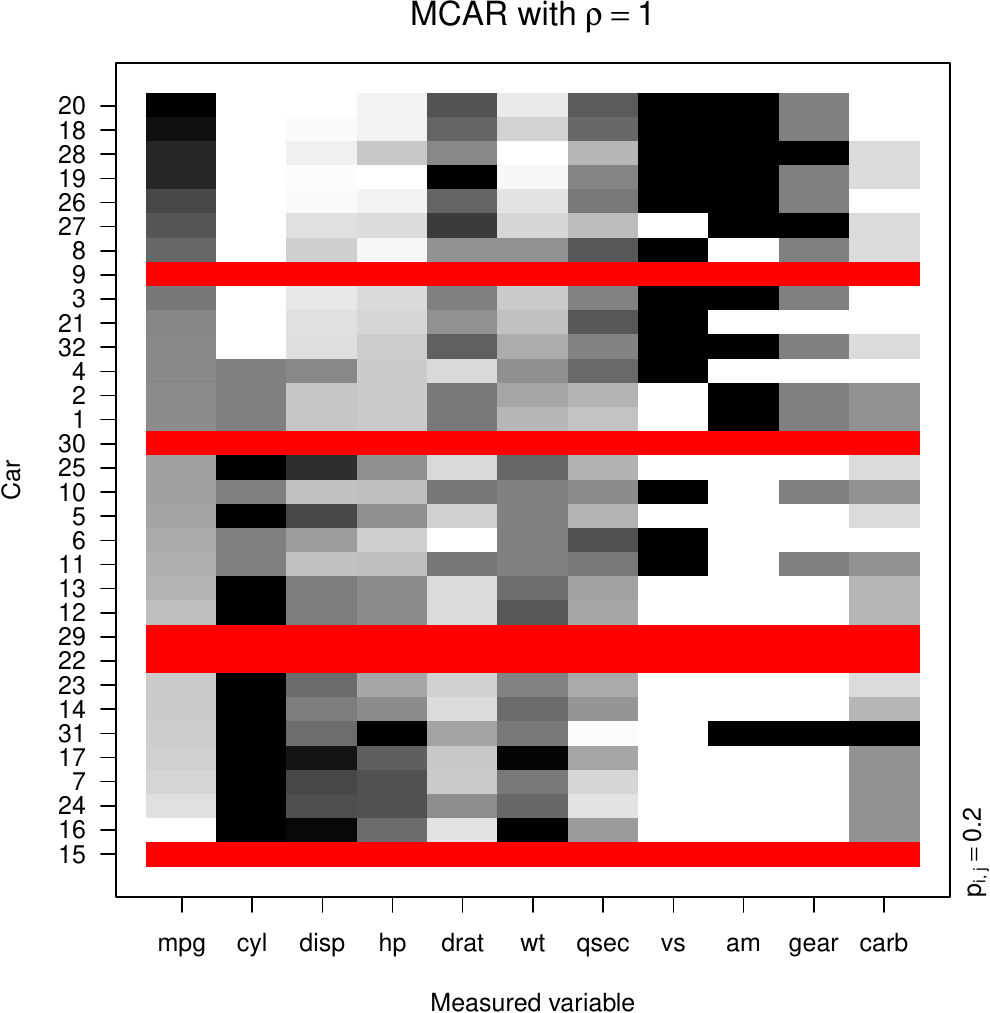}
  \caption{MCAR amputed \texttt{mtcars01} dataset according to
    Algorithm~\ref{alg:Bern:amp:A3} with different strengths $\rho$
    of dependence of the underlying Gauss copula $C^{\text{Ga}}_\rho$,
    implying pairwise correlations of $0$ (independence copula; left),
    $0.5$ (center), and $1$ (comonotone copula; right) between the
    entries in $M$. The homogeneous marginal missingness probabilities
    are $p=1/3$ (top) and $p=1/5$ (bottom).}
  \label{fig:mtcars01:mcar}
\end{figure}
Comparing the two rows of plots, the influence of $p$ on the overall
proportion of missingness is clearly visible. Furthermore, as $\rho$
increases, there is more structure to the missingness pattern. When
$\rho=0$ ($C=C^{\Pi}$), missing values occur randomly across the
dataset, whilst for $\rho=1$ ($C=C^{\text{M}}$), rows of the amputed
dataset are necessarily entirely missing or entirely complete. Under
$\rho=1$, the probability of joint missingness of a whole row is about
$1/3$ in the top right plot and about $1/5$ in the bottom right plot,
but for $\rho=0$ (in the two plots of the left column), these
probabilities are only about $1/3^{11}\approx 5.65/10^6$ and
$1/5^{11}\approx 2.048/10^8$, respectively.

The plots in Figure~\ref{fig:mtcars01:mar} show MAR missingness
patterns for the same $\rho$ values as in Figure~\ref{fig:mtcars01:mcar}.
\begin{figure}[htbp]
  \centering
  \includegraphics[width=0.32\textwidth]{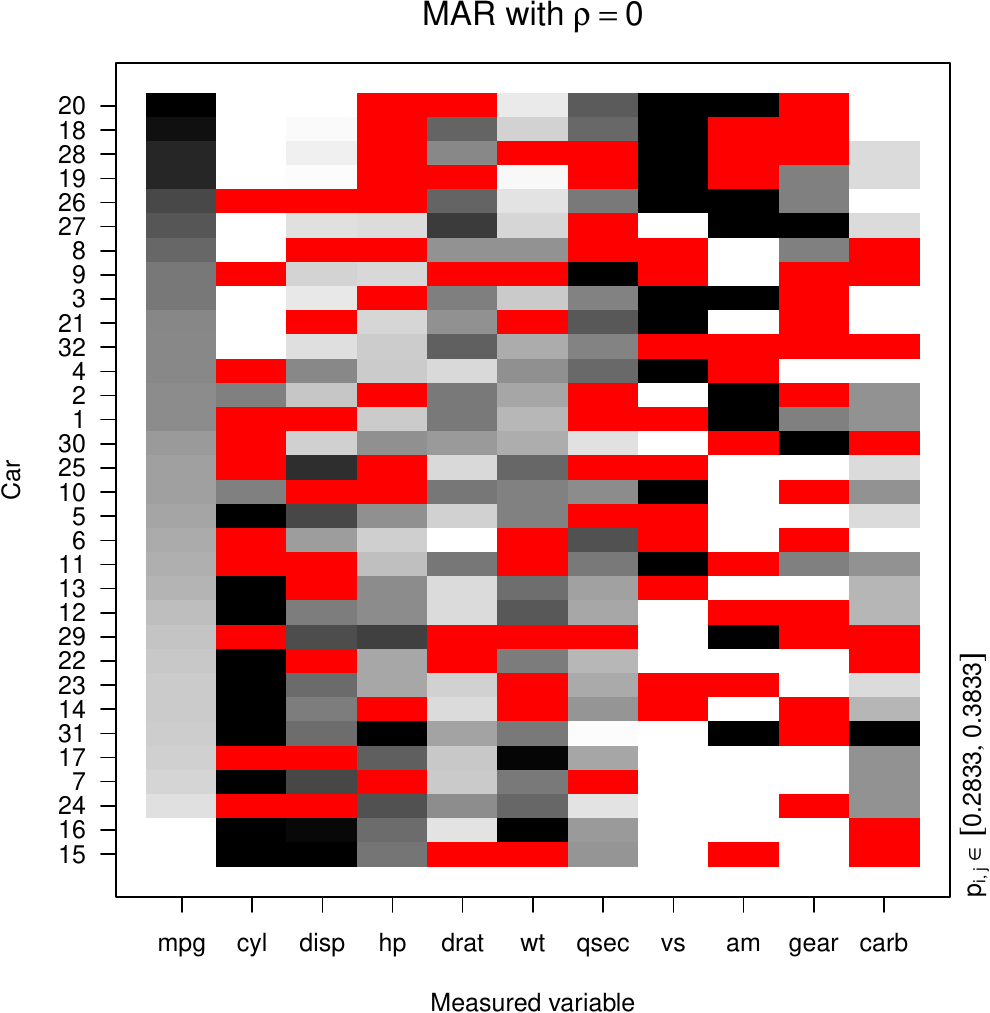}
  \hfill
  \includegraphics[width=0.32\textwidth]{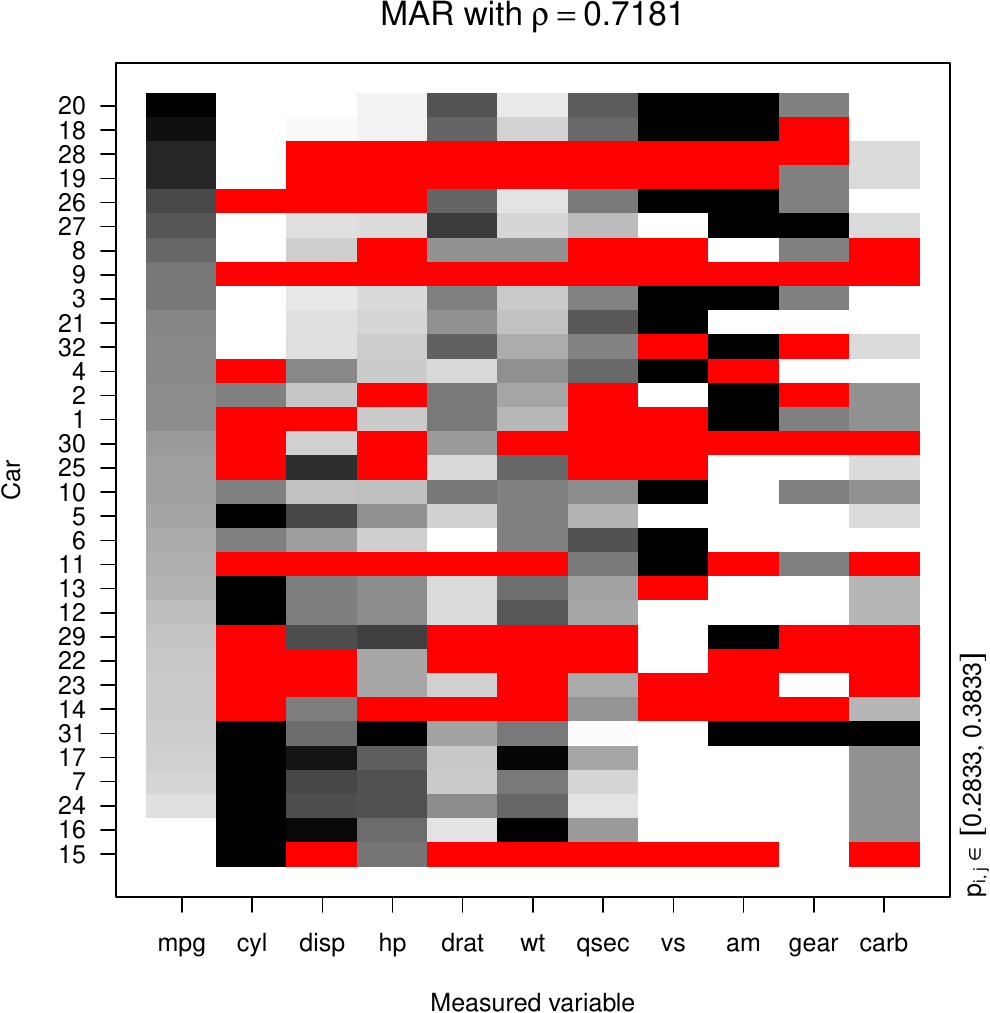}
  \hfill
  \includegraphics[width=0.32\textwidth]{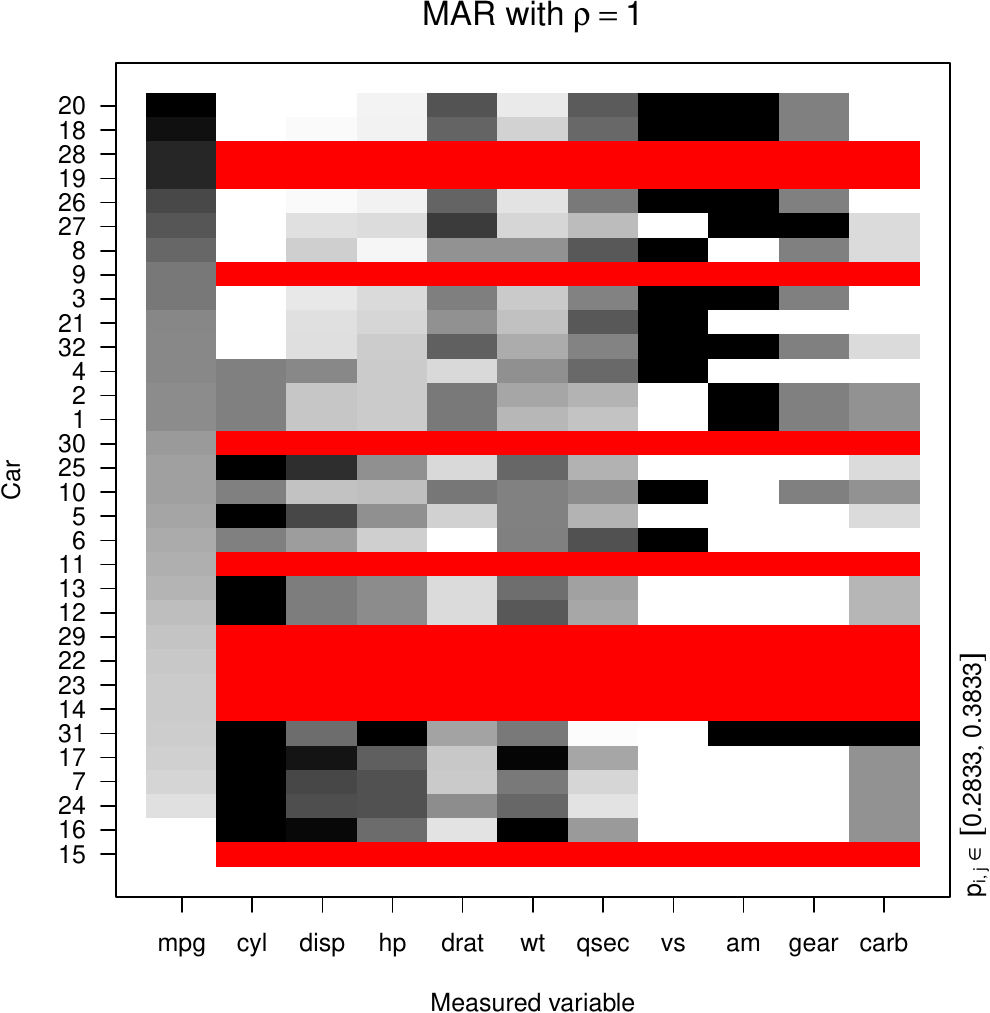}\\[4mm]
  \includegraphics[width=0.32\textwidth]{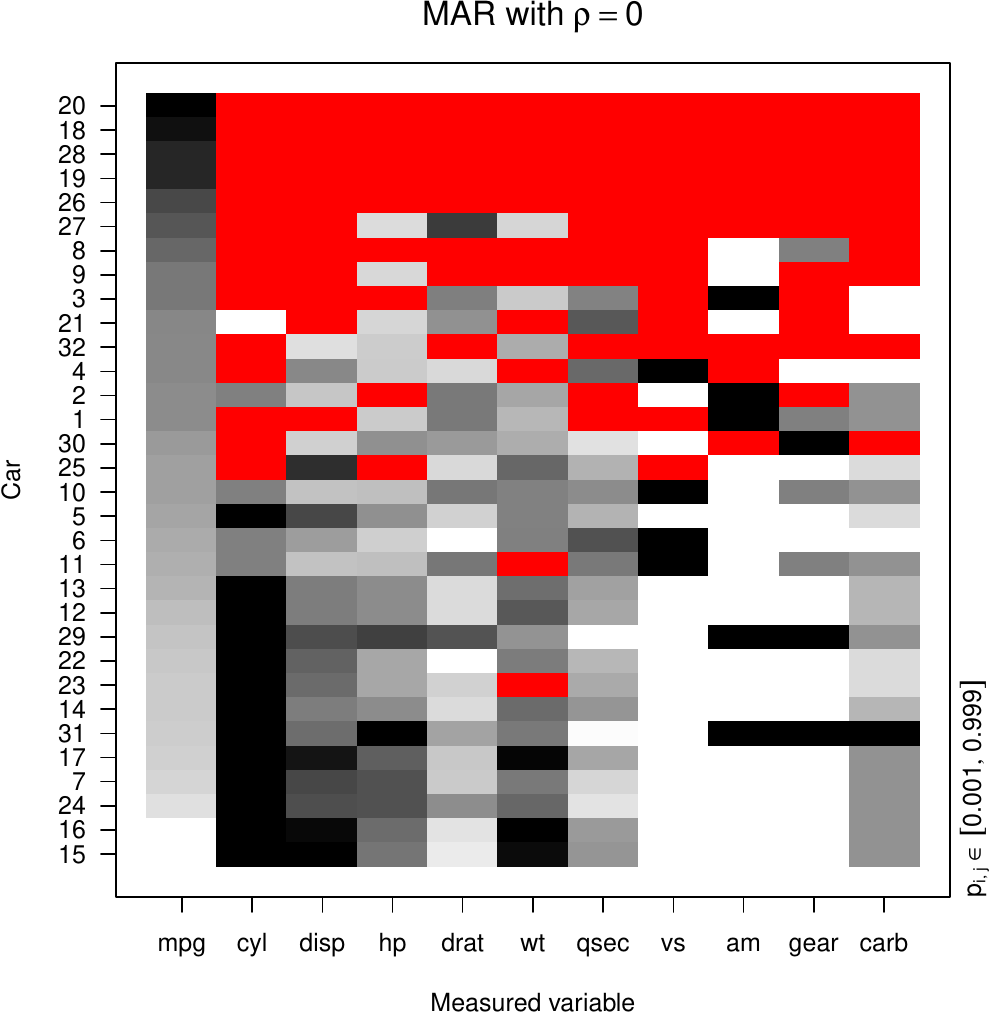}
  \hfill
  \includegraphics[width=0.32\textwidth]{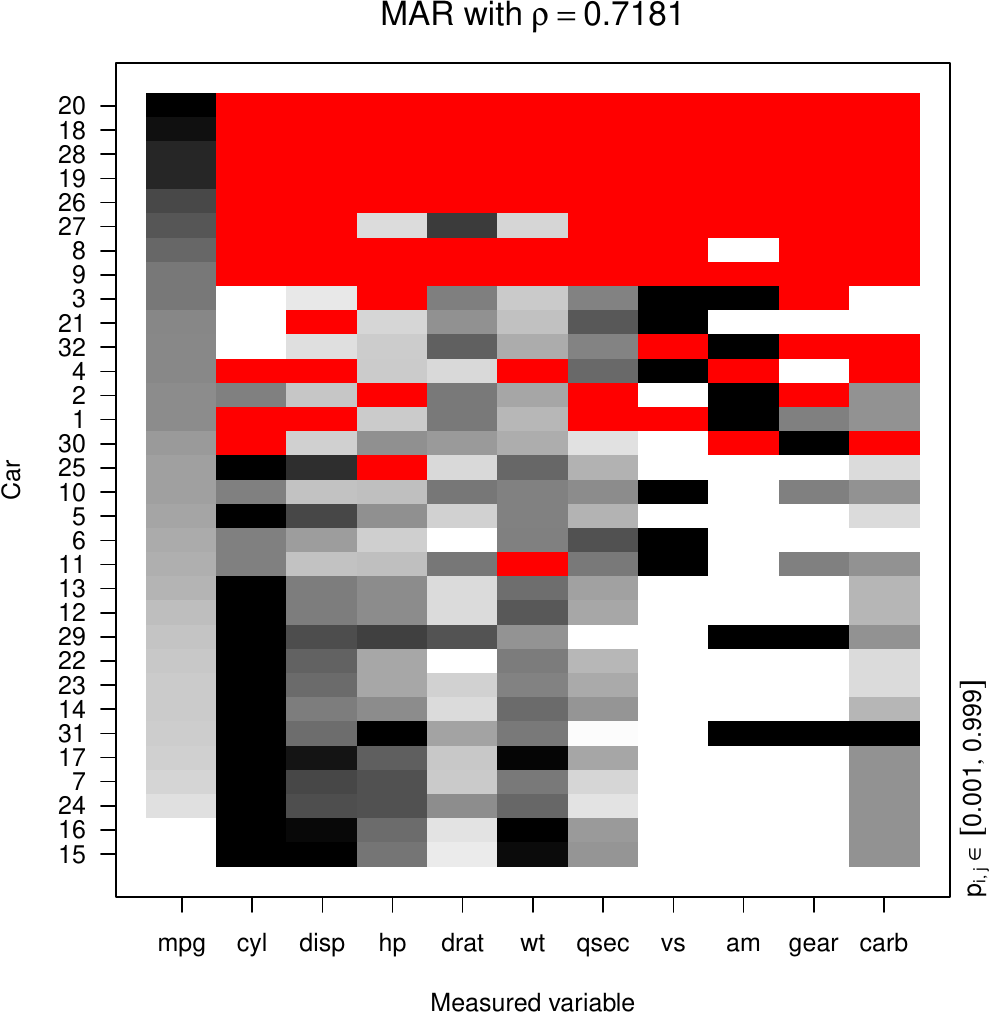}
  \hfill
  \includegraphics[width=0.32\textwidth]{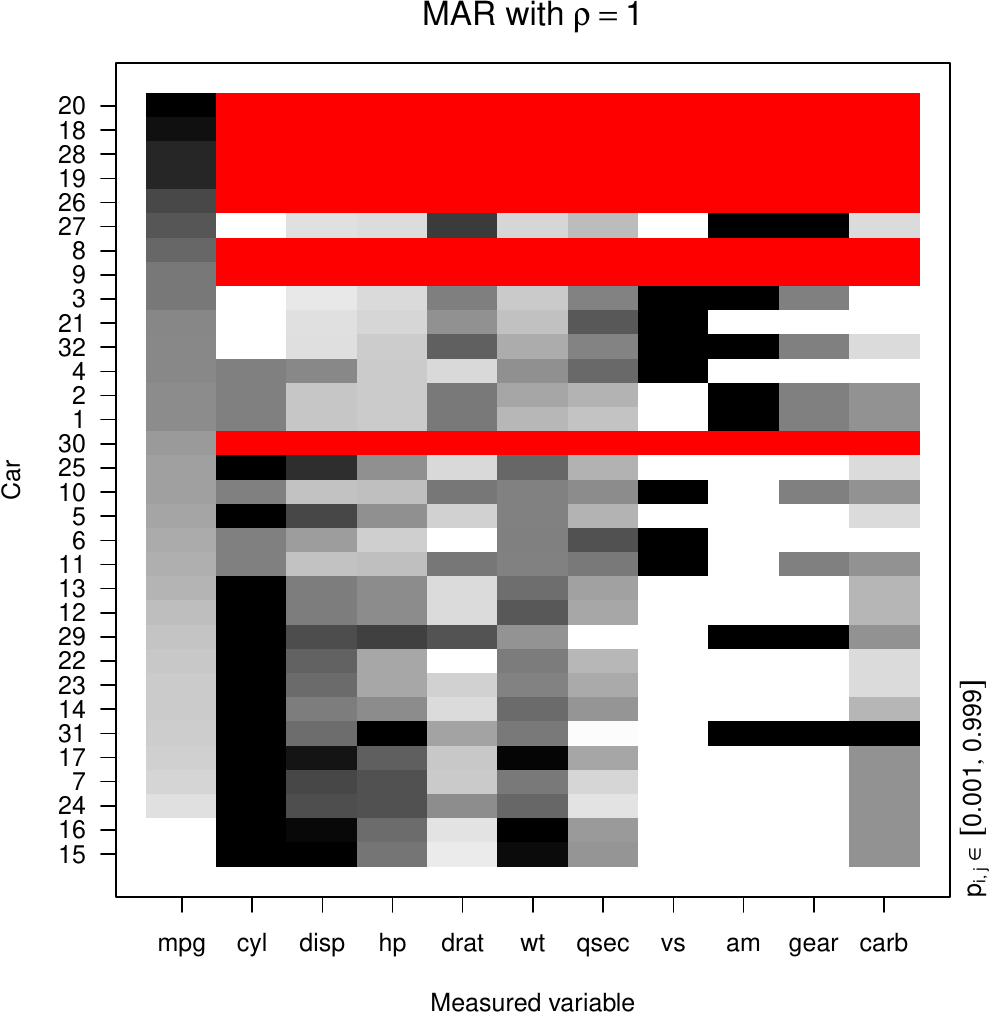}
  \caption{MAR amputed \texttt{mtcars01} dataset according to
    Algorithm~\ref{alg:Bern:amp:A3} with different strengths $\rho$ of
    dependence of the underlying Gauss copula $C^{\text{Ga}}_\rho$ (left, center, right) as in Figure~\ref{fig:mtcars01:mcar}.
    For all $i=1,\dots,32$, the marginal missingness probabilities
    $p_{i,j}$ are $p_{i,1}=0$ (no missingness) and, for
    $j=2,\dots,11$, $p_{i,j}\in [1/3-0.05, 1/3+0.05]$ (top row) and
    $p_{i,j}\in [0.001, 0.999]$ (bottom row), where $p_{i,j}$,
    $j=2,\dots,11$, depends on $Y_{i,1}$ via~\eqref{eq:regr} for equal
    $\beta_{i,j;0}$ and equal $\beta_{i,j;1}$, %
    determined by Lemma~\ref{lem:ran:beta}.}
  \label{fig:mtcars01:mar}
\end{figure}
To easily see the influence of MAR missingness, the marginal
missingness probabilities $p_{i,j}$ are chosen as $p_{i,1}=0$ (no
missingness) and, for $j=2,\dots,11$,
$p_{i,j}\in [1/3-0.05, 1/3+0.05]$ (top row) and
$p_{i,j}\in [0.001, 0.999]$ (bottom row). Each $p_{i,j}$ depends on
$Y_{i,1}$ via~\eqref{eq:regr} for equal $\beta_{i,j;0}$ and equal
$\beta_{i,j;1}$ (only two coefficients), determined by
Lemma~\ref{lem:ran:beta}. In this case, the sets $I_{i,j}$
in~\eqref{eq:regr} are $I_{i,1}=\emptyset$ and
$I_{i,2}=\ldots=I_{i,11}=\{1\}$ for all $i=1,\dots,32$. In the top
row, we see how the dependence changes in terms of $\rho$, but the
effect of MAR is not particularly strong. Although $Y_{20,1}=1$ is
large, $M_{20,j}=1$ only for few $j$. In the bottom row, where
$p_{i,j}\in [0.001, 0.999]$ (with larger values indicating larger
missingness probabilities), we can clearly
see that missingness patterns for variables $Y_{i,2},\dots,Y_{i,11}$
appear more often if $Y_{i,1}$ is large.

We also see the effect on the implied dependence the probabilities
$p_{i,2},\dots,p_{i,11}$ have. Large $Y_{i,1}$ lead to large
$p_{i,2},\dots,p_{i,11}$ and a high probability that multiple
$Y_{i,2},\dots,Y_{i,11}$ are amputed, thus creating a missingness
pattern across columns $2$ to $11$ even in the case of
$\rho=0$; see Proposition~\ref{prop:joint:missingness:prob} where we
addressed this effect of a large probability of simultaneous missingness
(including the extreme cases $p_{i,j}\in\{0,1\}$).

Similarly, the plots in Figure~\ref{fig:mtcars01:mnar} show MNAR
missingness patterns for the same $\rho$ values as before and with
$p_{i,j}\in [1/3-0.05, 1/3+0.05]$ (top and middle row) and
$p_{i,j}\in [0.001, 0.999]$ (bottom row).
\begin{figure}[htbp]
  \centering
  \includegraphics[width=0.32\textwidth]{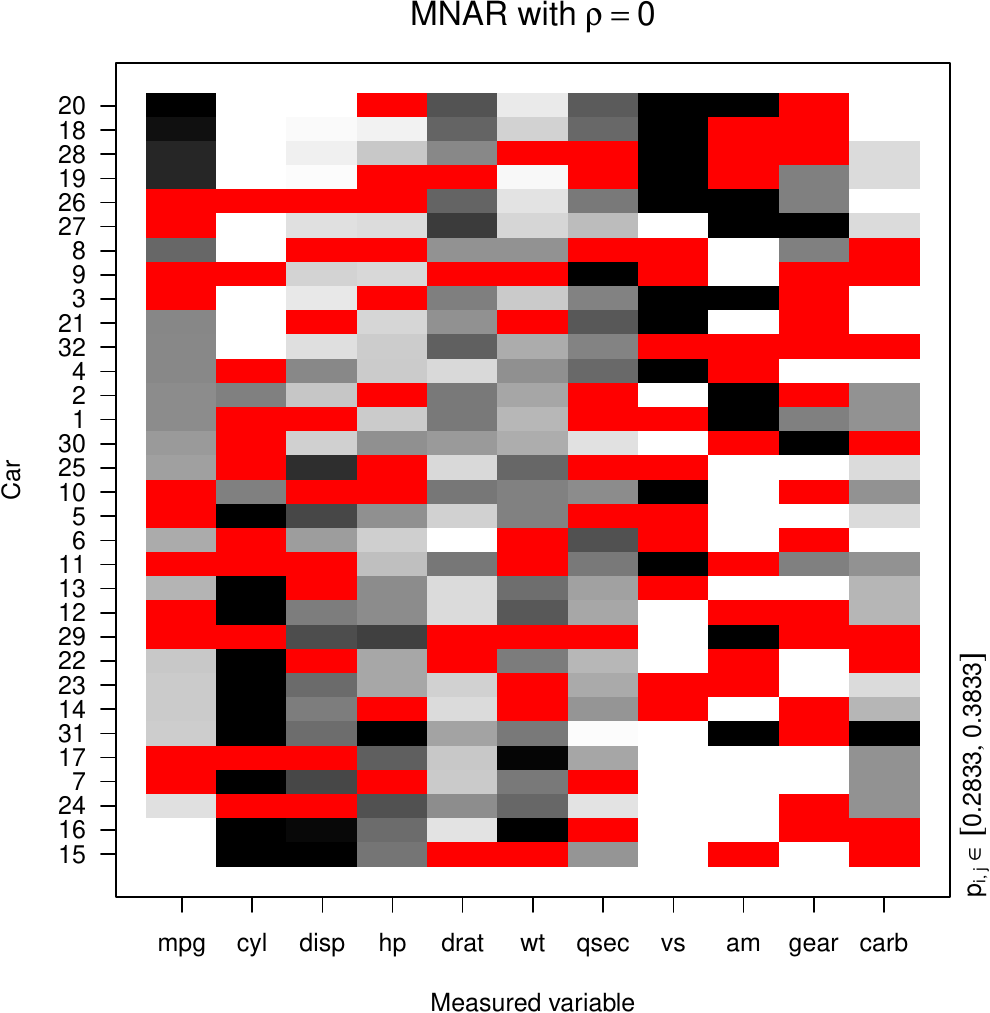}
  \hfill
  \includegraphics[width=0.32\textwidth]{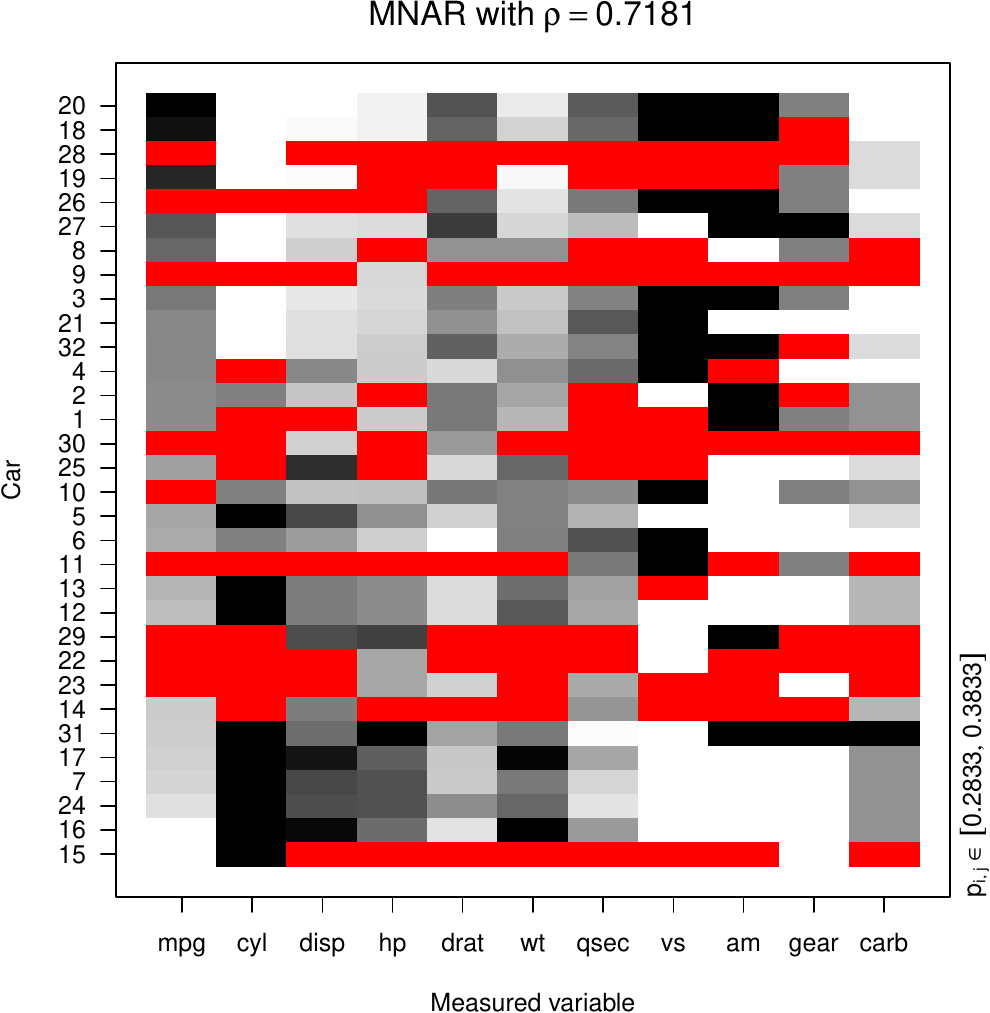}
  \hfill
  \includegraphics[width=0.32\textwidth]{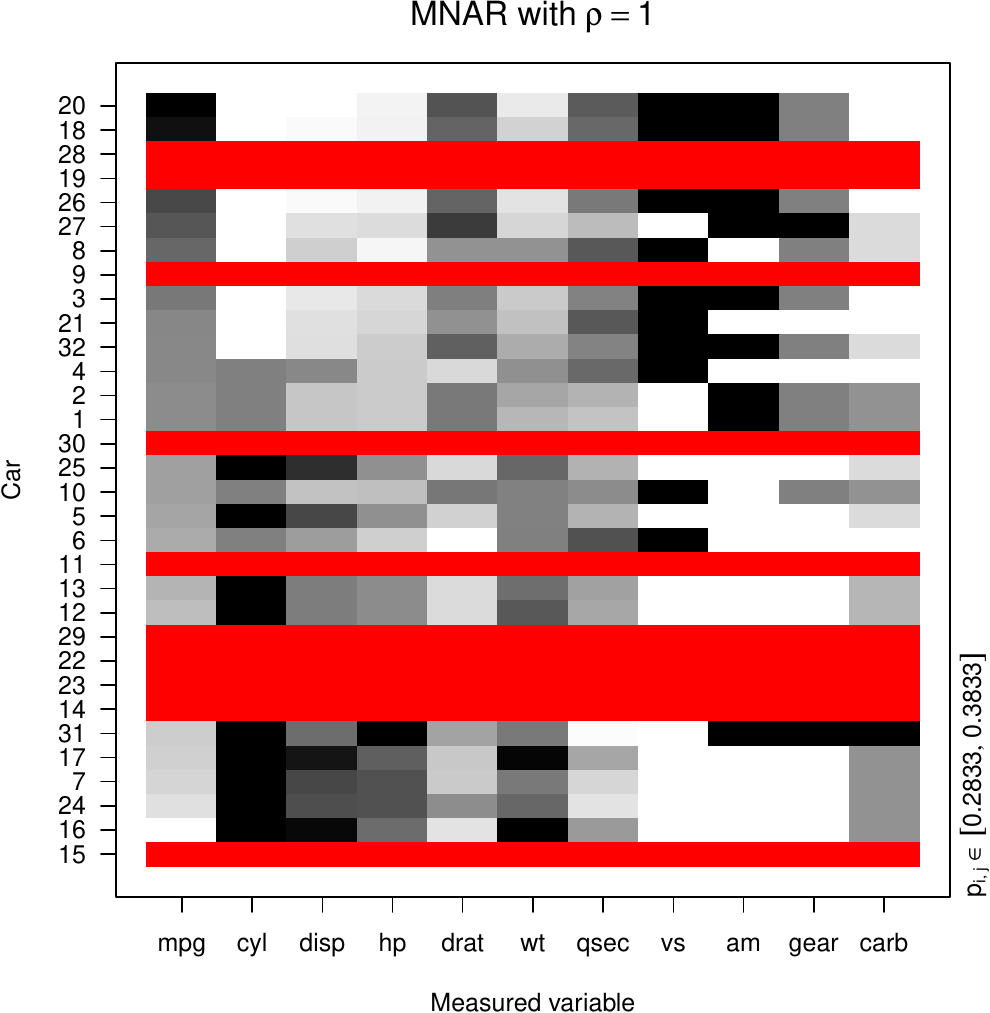}\\[4mm]
  \includegraphics[width=0.32\textwidth]{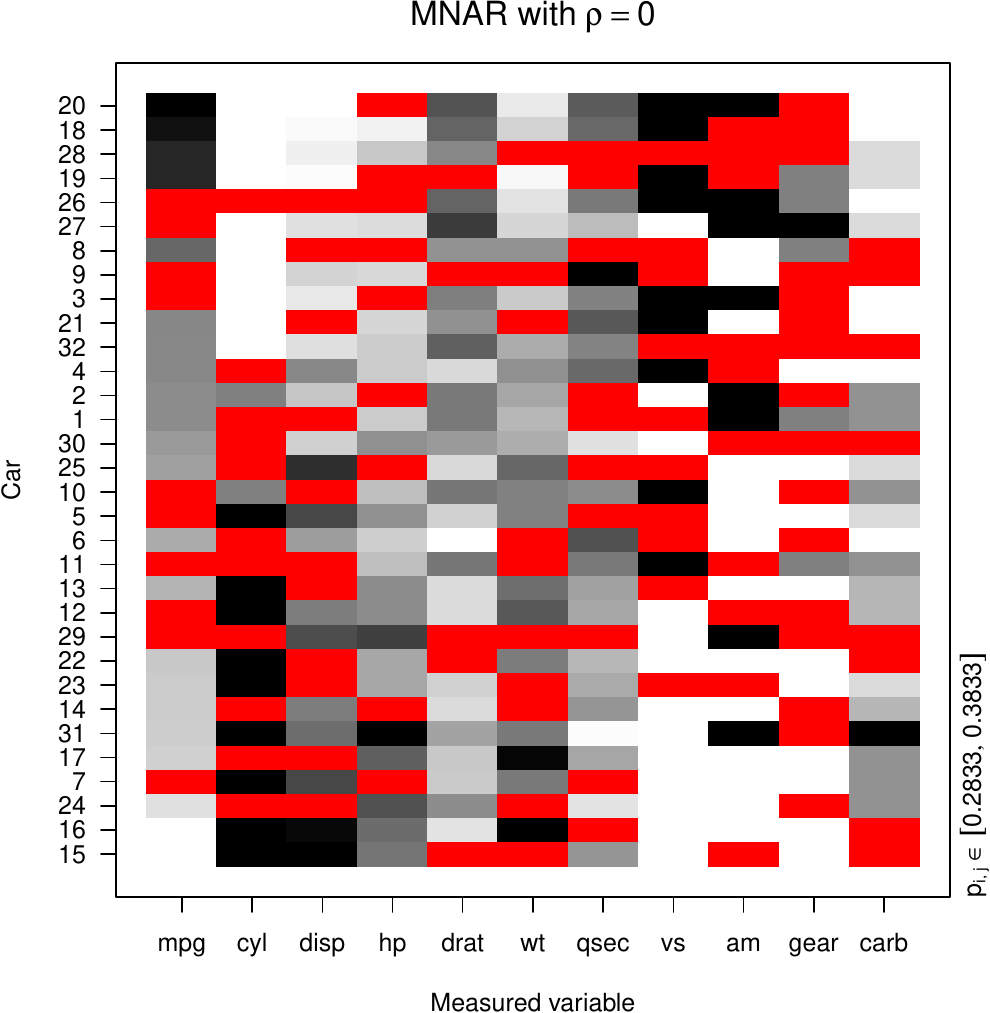}
  \hfill
  \includegraphics[width=0.32\textwidth]{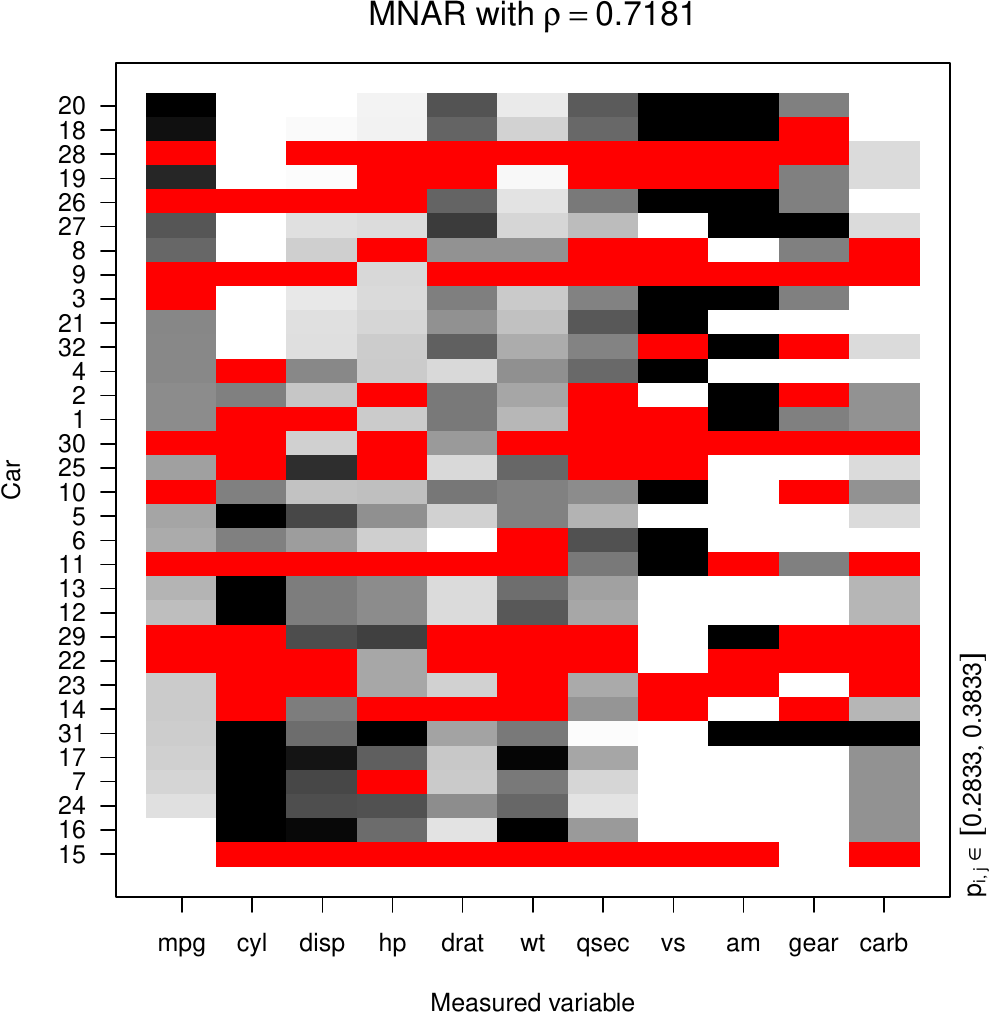}
  \hfill
  \includegraphics[width=0.32\textwidth]{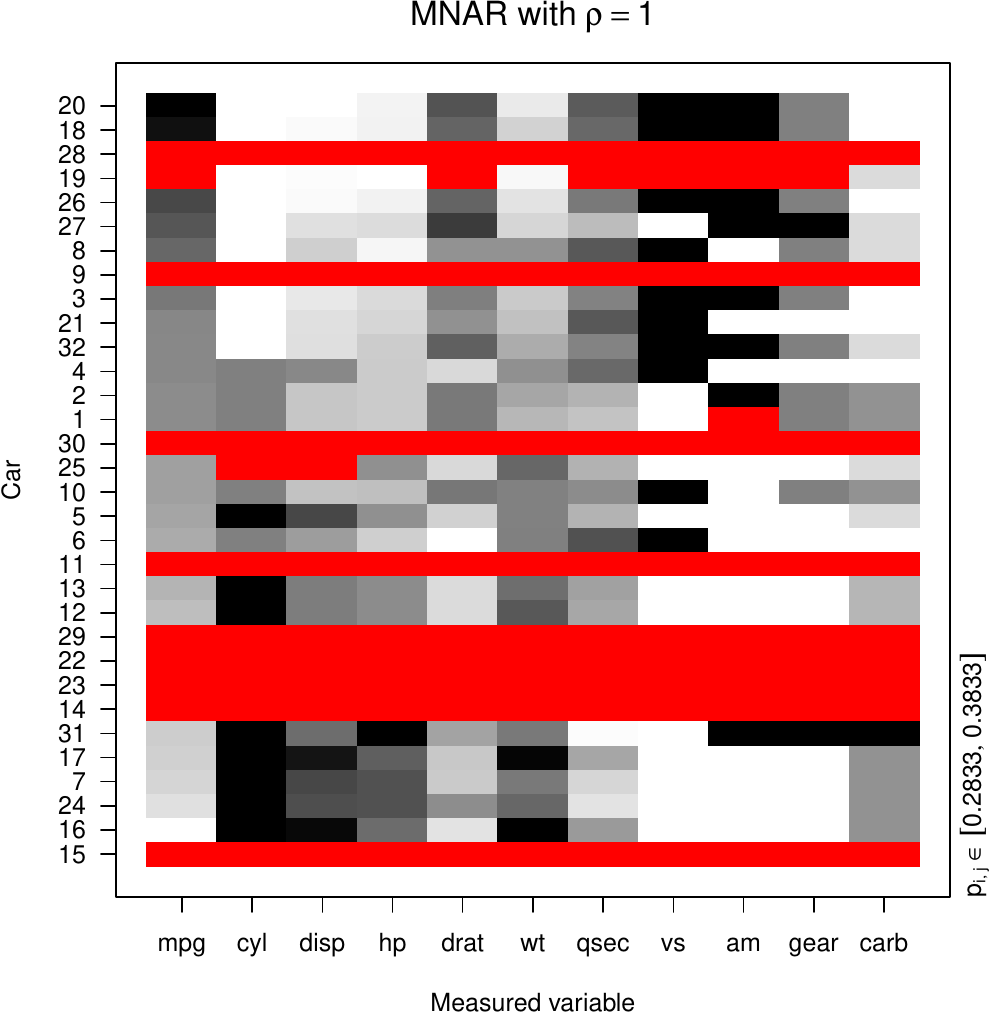}\\[4mm]
  \includegraphics[width=0.32\textwidth]{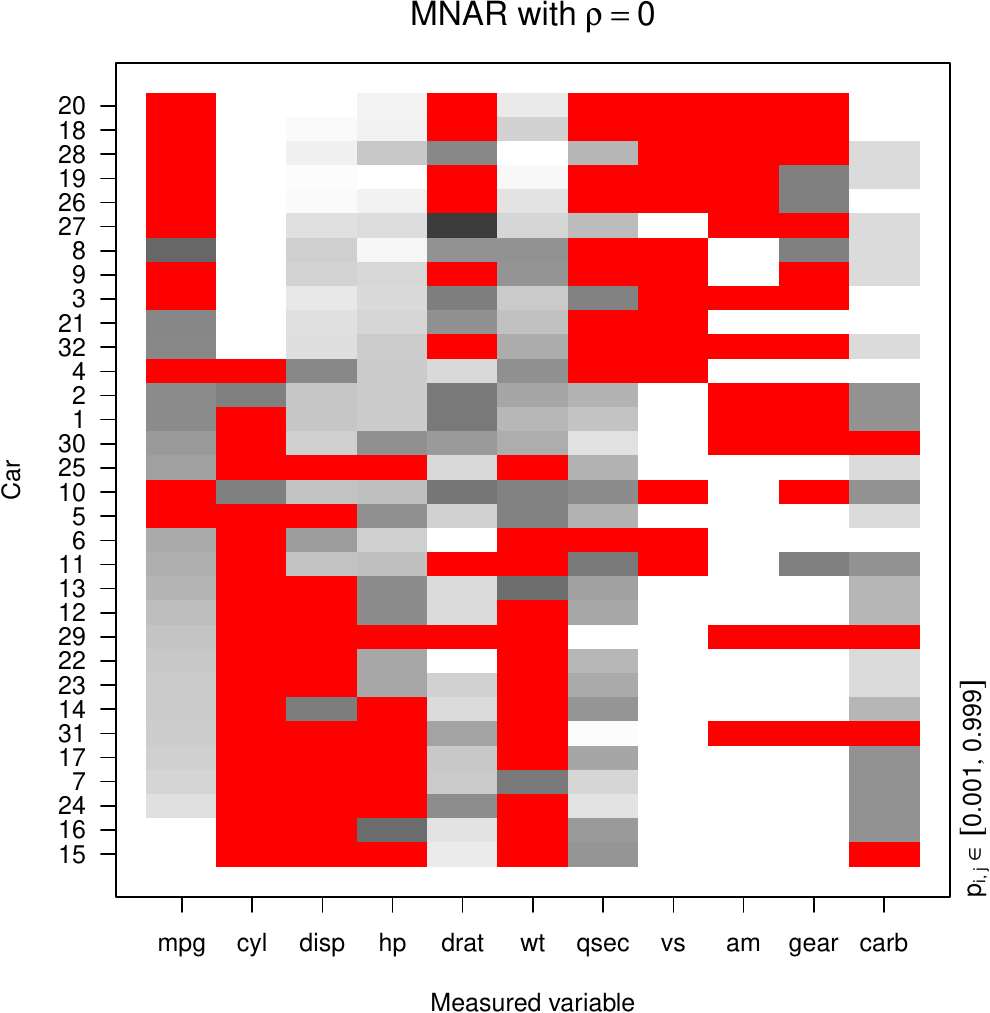}
  \hfill
  \includegraphics[width=0.32\textwidth]{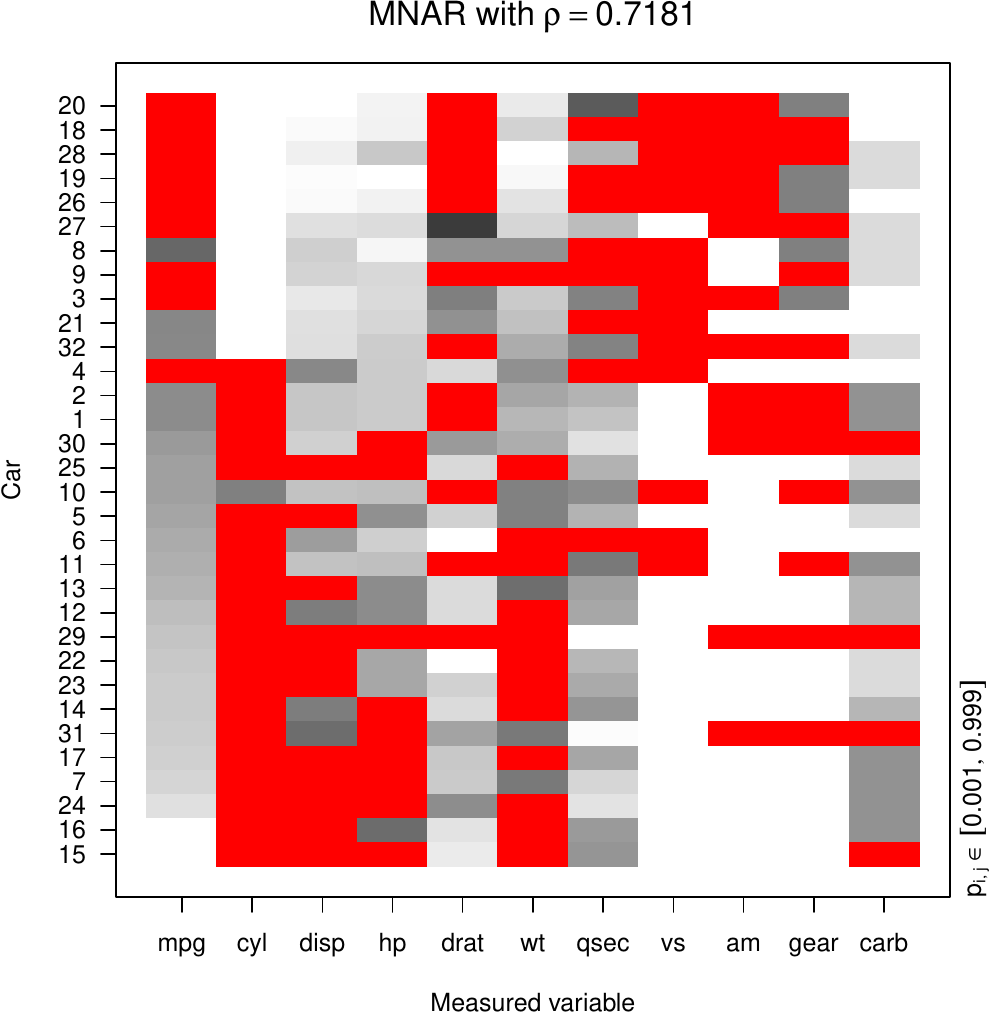}
  \hfill
  \includegraphics[width=0.32\textwidth]{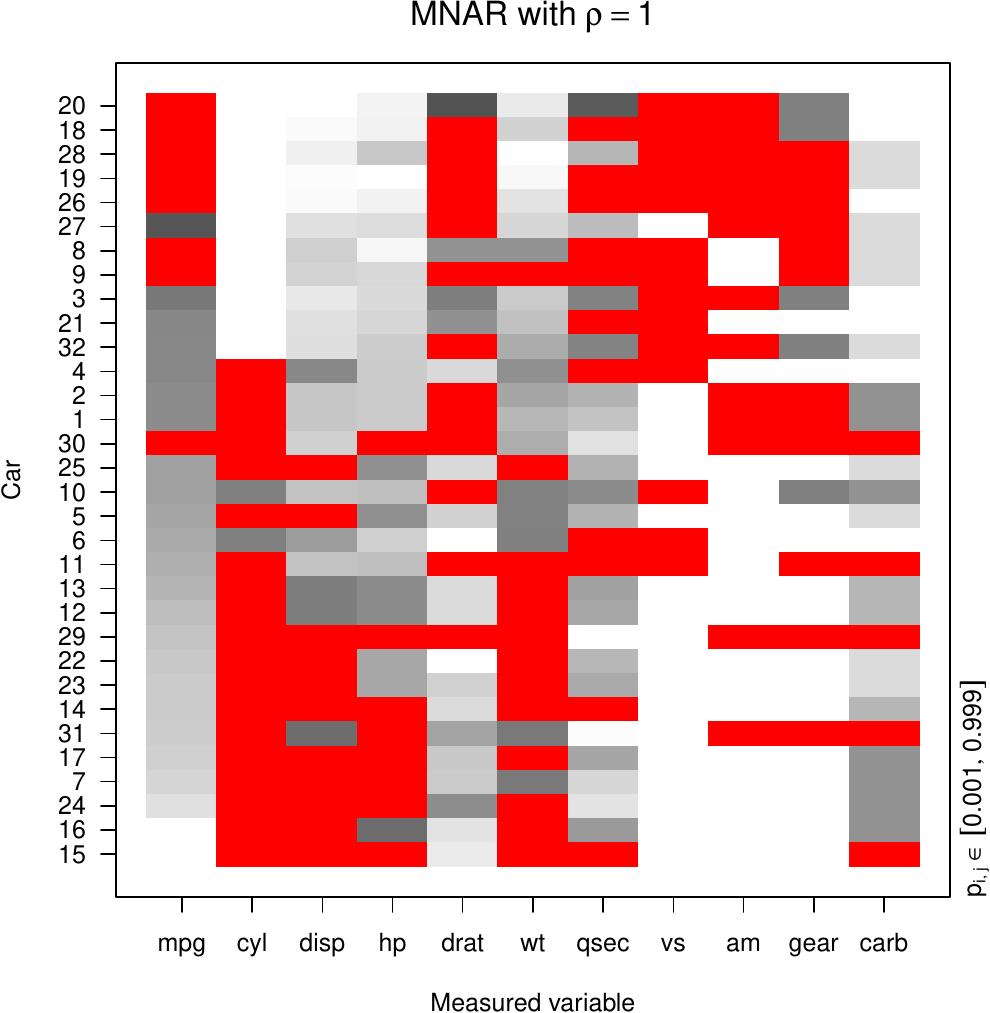}
  \caption{MNAR amputed \texttt{mtcars01} dataset according to
    Algorithm~\ref{alg:Bern:amp:A3} with different strengths $\rho$
    of dependence of the underlying Gauss copula $C^{\text{Ga}}_\rho$
    (left, center, right) as in Figure~\ref{fig:mtcars01:mcar}.
    For all $i=1,\dots,32$ and $j=1,\dots,11$, the marginal
    missingness probabilities $p_{i,j}$ are
    $p_{i,j}\in [1/3-0.05, 1/3+0.05]$ (top and middle row) and
    $p_{i,j}\in [0.001, 0.999]$ (bottom row), where $p_{i,j}$ depends
    on $Y_{i,1},\dots,Y_{i,11}$ (group-MNAR; top row) or $Y_{i,j}$
    (self-MNAR; middle and bottom row) via~\eqref{eq:regr} for
    equal $\beta_{i,j;0}$ and equal
    $\beta_{i,j;1}$, %
    determined by Lemma~\ref{lem:ran:beta}.}
  \label{fig:mtcars01:mnar}
\end{figure}
The top row shows group-MNAR missingness, where each
$Y_{i,1},\dots,Y_{i,11}$ contributes equally to $p_{i,j}$
via~\eqref{eq:regr}. As before, we chose equal $\beta_{i,j;0}$ and
equal $\beta_{i,j;1}$ (only two coefficients), determined by
Lemma~\ref{lem:ran:beta}; in this case, the sets $I_{i,j}$
in~\eqref{eq:regr} are $I_{i,1}=\ldots=I_{i,11}=\{1,\ldots,11\}$ for
all $i=1,\dots,32$. The middle and bottom row show self-MNAR
missingness, again with equal $\beta_{i,j;0}$ and equal
$\beta_{i,j;1}$; in this case, $I_{i,j}=\{j\}$ for all $i=1,\dots,32$,
$j=1,\ldots,11$. Owing to the requirement
$p_{i,j}\in [1/3-0.05, 1/3+0.05]$ in the top and middle row, the
self-MNAR vs.\ group-MNAR effect is not very pronounced. Similarly, as
in the MAR case, the bottom row allows for
$p_{i,j}\in [0.001, 0.999]$, and we immediately see the self-MNAR
effect as most dark cells are set to missing now at the expense of the
influence of the dependence (parameter $\rho$); see
Section~\ref{sec:non:unique}. %

\subsection{Application to stress testing of risk measures}\label{sec:app:VaR:ES}
We now turn to an application from the realm of quantitative risk management;
see \cite{mcneilfreyembrechts2015} for more details on the presented
concepts. Financial firms are required to hold a certain amount of capital,
known as \emph{risk capital}, to absorb potential future losses on exposures
subject to market, credit, and operational risk. A \emph{loss distribution} is
the distribution function $F_L$ of the (random) loss $L$ the firm potentially
faces over a predetermined future time horizon. A \emph{risk measure} maps a
loss distribution to a real number, interpreted as the amount of capital to put
aside to account for future losses. Two widely used risk measures are
\emph{value-at-risk} at \emph{level} $\alpha\in(0,1)$, given by
$\VaR_\alpha=F_L^{-1}(\alpha)=\inf\{x\in\IR:F_L(x)\ge\alpha\}$, and
\emph{expected shortfall} at \emph{level} $\alpha\in(0,1)$, given by
$\ES_\alpha=\frac{1}{1-\alpha}\int_{\alpha}^1 F_L^{-1}(u)\,\rd u$.  A financial
firm is naturally interested in reducing its risk capital, as holding any amount
beyond the unknown amount of what is actually needed is expensive and could
instead be used to make investments.

As a concrete setup, suppose a financial firm holds a portfolio exposed to
market risk, with its components subject to losses from market movements. The
firm is therefore naturally interested in the effect that avoiding large losses
has on its risk capital estimates. To investigate this effect, we consider
$d=434$ constituents of the S\&P~500 available in the \R\ package
\texttt{qrmdata} from 2001-01-01 to 2015-12-31; the components are selected to
have less than 10\% missing data and completed by linear interpolation (and
extension at the boundaries). To handle serial dependence, $n=3772$ negative
logarithmic returns are formed, componentwise deGARCHed by fitting
$\ARMA(1,1)-\GARCH(1,1)$ models with standard normal
innovations, %
and the standardized residuals extracted. The standardized residuals serve as
iid observations for the firm's analysis. The population loss to be modeled here is
$L=\sum_{j=1}^d Z_j$, where $Z_j$ denotes a random variable from the innovation
distribution of the $j$th component, $j=1,\dots,d$. The sample losses are
$L_i=\sum_{j=1}^d Z_{i,j}$, $i=1,\dots,n$, where $Z_{i,j}$ are the standardized residuals
obtained from deGARCHing. Based on $L_1,\dots,L_n$, $\VaR_\alpha$ and $\ES_\alpha$
can be nonparametrically estimated via
\begin{align*}
  \widehat{\VaR}_\alpha &= \hat{F}_{L,n}^{-1}(\alpha) = L_{(\lceil n\alpha\rceil)},\\
  \widehat{\ES}_\alpha  &= \frac{1}{n-\lceil n\alpha\rceil}\sum_{i=\lceil n\alpha\rceil+1}^n L_{(i)},
\end{align*}
where $\hat{F}_{L,n}$ denotes the empirical distribution function of $L_1,\dots,L_n$ and
$L_{(1)},\dots,L_{(n)}$ the order statistics of $L_1,\dots,L_n$.

To study the effect of a reduced capital requirement, the firm introduces
missingness into the standardized residuals $Z_{i,j}$ (effectively replacing
selected ones by $0$ for estimating $\VaR_\alpha$, $\ES_\alpha$ here). For each
$i=1,\dots,n$ such that $L_i>\hat{F}_{L,n}^{-1}(\alpha)$, so for each joint
observation such that the corresponding total loss takes a realization in its
right tail, we replace $Z_{i,j}$ in the above formulas by
$Z_{i,j}'=Z_{i,j}\I_{\{U_{i,j}>p\}}$, where $U_{i,j}\isim\U(0,1)$ and the
probability of missingness $p$ runs in
$p\in\{\frac{k}{200}:k=0,\dots,200\}$. The corresponding total losses are
then $L_i'=\sum_{j=1}^d Z_{i,j}'$, $i=1,\dots,n$. Computing $\widehat{\VaR}_\alpha$
and $\widehat{\ES}_\alpha$ from $L_1',\dots,L_n'$ allows us to gradually study, as $p$
increases, the decrease in capital requirements; in the context of Bernoulli
amputation, this corresponds to an MNAR case of degree $d$.

Figure~\ref{fig:SP500:VaR:ES:99} shows plots of $\widehat{\VaR}_\alpha$ and
$\widehat{\ES}_\alpha$ computed from $L_1',\dots,L_n'$ as functions of $p$.
\begin{figure}[htbp]
  \centering
  \includegraphics[width=0.48\textwidth]{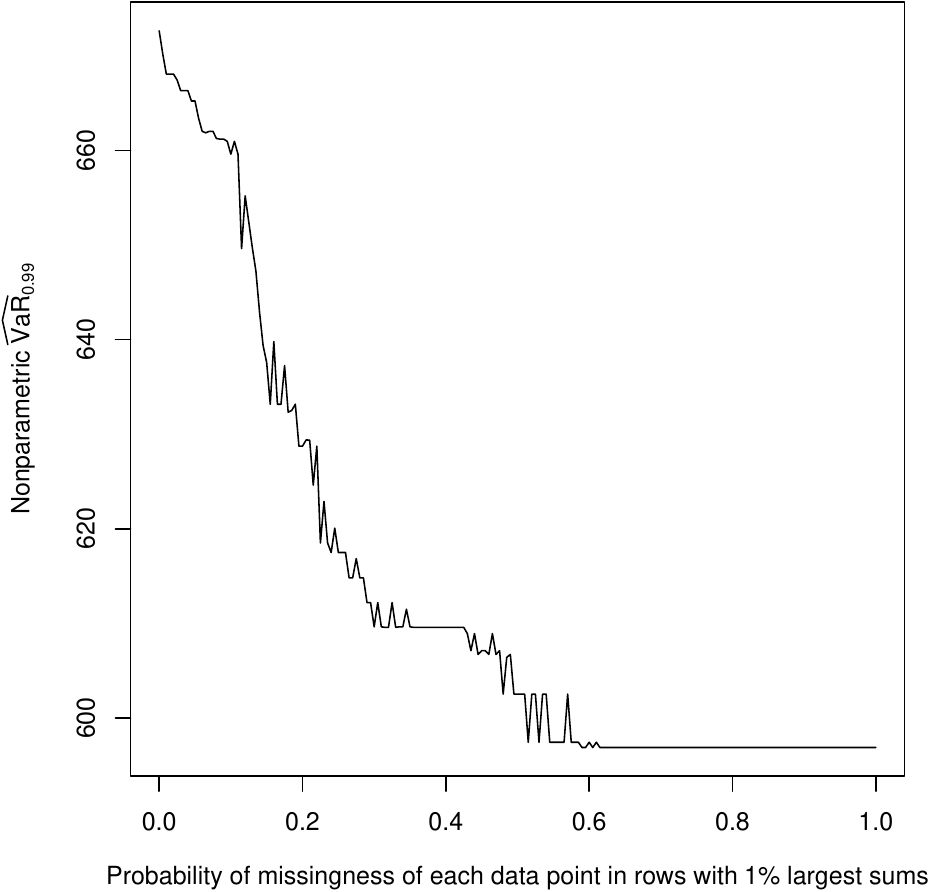}
  \hfill
  \includegraphics[width=0.48\textwidth]{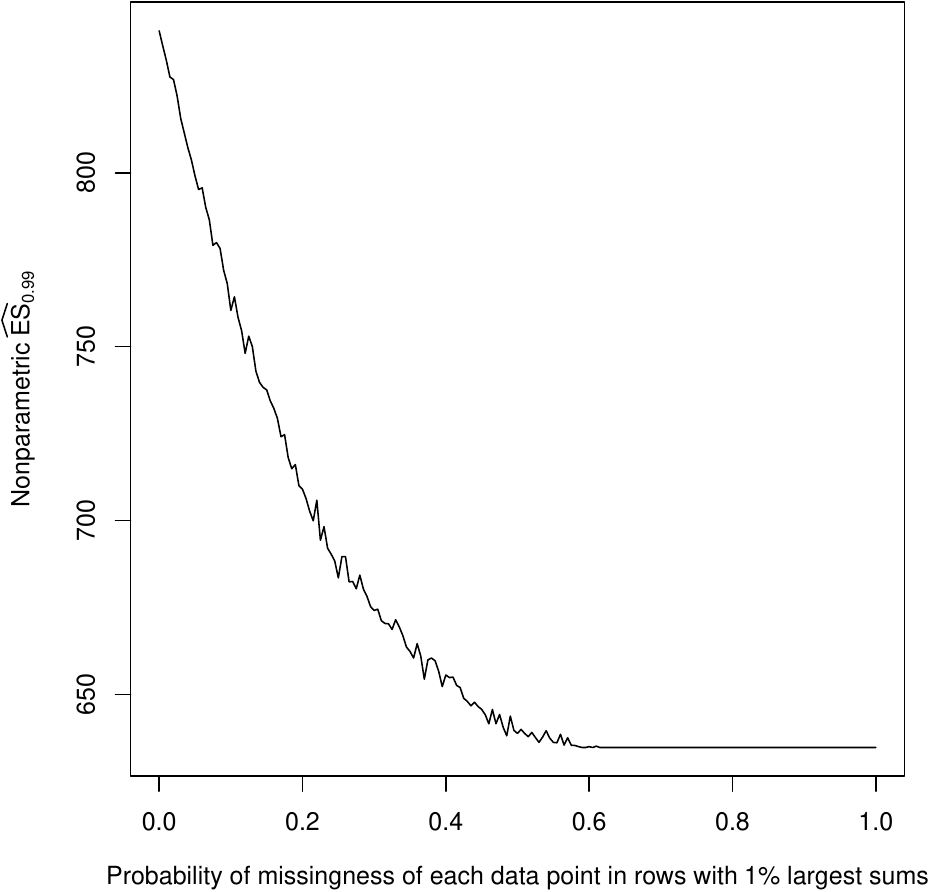}
  \caption{Effect on nonparametric $\VaR_{0.99}$ (left) and $\ES_{0.99}$ (right)
    estimates when varying the probability of MNAR missingness of
    entries %
    in rows with large sums (in their 1\% tail) of standardized residuals after
    deGARCHing of $n=3772$ negative log-returns (from 2001-01-01 to 2015-12-31) of
    $d=434$ constituents of the S\&P~500 dataset \texttt{SP500\_const} of the \R\
    package \texttt{qrmdata}.}
  \label{fig:SP500:VaR:ES:99}
\end{figure}
We can clearly see the effect on the risk capital estimates when individual
losses leading to larger total losses are gradually omitted more and more.
We also see that the drop is more substantial (25\% vs 11\%) for $\widehat{\ES}_{0.99}$
than for $\widehat{\VaR}_{0.99}$, since $\ES_{0.99}$ (as integral over the tail region)
is more heavily affected by changes of the distribution in the tail. Finally,
let us note again, that such a missingness pattern as used in this application cannot
be replicated with existing amputation approaches.

\section{Conclusion}\label{sec:concl}
We introduced Bernoulli amputation, a stochastic approach to amputation based on
Bernoulli margins and copula dependence, for generating multivariate missingness
in complete datasets. At its core, Bernoulli amputation allows to construct a
dependent missingness indicator matrix $M\in\{0,1\}^{n\times d}$ via copulas and
Bernoulli $\B(1,p_{i,j})$ distributed $(i,j)$th margins, $i=1,\dots,n$,
$j=1,\dots,d$.

Besides its ability to capture MCAR, MAR, and MNAR missingness patterns, one
main advantage of Bernoulli amputation is to allow for structured missingness in
a principled manner. Random monotone missingness patterns can also be covered,
namely via mixtures. The stochastic nature of Bernoulli amputation can be
particularly useful to apply in simulation studies to evaluate and assess
imputation methods' effectiveness in repeated experiments under the same
distributional assumptions but with varying missingness patterns; especially
under multiple amputation, an amputer cannot easily construct multiple
missingness matrices manually such that they can be viewed as realisations of
the same distribution of $M$ given $Y$. Also, Bernoulli amputation allows the
amputer to focus on structural decisions (dependence, margins,
who-influences-whom under MAR and MNAR) rather than many specific (and thus
stochastically non-exhausting) missingness patterns. Furthermore, Bernoulli
amputation is straightforward to implement with already existing software and we
demonstrated the approach in terms of examples and illustrations based on
a publicly available dataset.

We could also derive mathematical quantities such as joint missingness
probabilities of any collection of entries in $M$, as well as their correlation,
under Bernoulli amputation. These results can be helpful in setting up
reasonable parameter choices for Bernoulli amputation. Besides that, we
recommend to conduct pilot studies based on manageable sample sizes and plot the
corresponding $M$ in order to tune the distribution of $M$ to produce
missingness patterns of interest.

\appendix
\section{Appendix}\label{sec:appendix}

\subsection{Missingness mechanisms}\label{sec:missingness:mechanisms}
\subsubsection{The classical concepts of MCAR, MAR, and MNAR}
According to \cite{Rubin1976}, \cite{Seaman2013},
\cite{meallirubin2015}, and \cite[Section~1.3]{LittleRubin2020},
missingness mechanisms are often categorised as \emph{missing completely
  at random (MCAR)}, \emph{missing at random (MAR)}, and \emph{missing
  not at random (MNAR)}.  Analytically, based on the conditional
probability mass function
$f_{M|Y}(m\,|\,y;\,\bm{\theta})=\P(M=m\,|\,Y=y;\ \bm{\theta})$ of $M$
given $Y=y$ and depending on a parameter vector
$\bm{\theta}\in\Theta\subseteq\IR^r$, the following definitions of the
missingness mechanisms MCAR, MAR, and MNAR capture the nature of the
relationship between $Y$ and $M$, leading to $X$ via
$X=\ast\,M + Y\odot(1-M)$.%
\begin{definition}[MCAR missingness]\label{def:MCAR}
  The missingness mechansims is \emph{missing completely at random (MCAR)}
  if $f_{M|Y}(m\,|\,y;\,\bm{\theta})=f_{M|Y}(m\,|\,y';\,\bm{\theta})$
  for all $m, \bm{\theta}, y, y'$.
\end{definition}
By Definition~\ref{def:MCAR}, MCAR assumes that the distribution of
$M\,|\,Y=y$ remains invariant under the choice of realisation $y$ of
$Y$, so the distribution of $M$ does not depend on $Y$.

\begin{definition}[MAR missingness]\label{def:MAR}
  The missingness mechanism is \emph{missing at random (MAR)}
  if %
  $f_{M|Y}(m\,|\,y;\,\bm{\theta})=f_{M|Y}(m\,|\,y';\,\bm{\theta})$ for
  all $m, \bm{\theta}, y, y'$ such that
  $x^{\text{obs}}=x'^{\,\text{obs}}$.
\end{definition}
MAR assumes that the distribution of $M\,|\,Y=y$ remains invariant
under the choice of realisation $y$ of $Y$ as long as the observed
part $x^{\text{obs}}$ remains the same, so the distribution of
$M\,|\,Y=y$ does not depend on the missing part $x^{\text{mis}}$ (but
can, and typically does, depend on the observed part
$x^{\text{obs}}$).

Adapted to our notation, the definition of MAR in
\cite[Equation~(1.2)]{LittleRubin2002} involved the equation
$f_{M|Y}(m\,|\,y;\bm{\theta})=f_{M|Y}(m\,|\,y^{\text{obs}};\bm{\theta})$, %
which was mentioned by \cite{Seaman2013} for its recursive nature
since $M$ itself is influenced by
$Y^{\text{obs}}$; %
the definition was later changed in
\cite[Equation~(1.2)]{LittleRubin2020}.

\begin{definition}[MNAR missingness]\label{def:MNAR}
  The missingness mechanism is \emph{missing not at random (MNAR)}
  if %
  there exist $m,\bm{\theta},y,y'$ with
  $x^{\text{obs}}=x'^{\,\text{obs}}$ such that
  $f_{M|Y}(m\,|\,y;\,\bm{\theta})\neq
  f_{M|Y}(m\,|\,y';\,\bm{\theta})$.
\end{definition}
MNAR assumes that the distribution of $M\,|\,Y=y$ depends on
the choice of realisation $y$ of $Y$, in particular also on the
missing part $x^{\text{mis}}$. %

\begin{remark}[``everywhere'' vs.\ ``realised'']\label{rem:everywhere:vs:realised}
  The literature contains different ways to state
  Definitions~\ref{def:MCAR}, \ref{def:MAR}, and \ref{def:MNAR}.  Our
  MCAR definition coincides with the ``everywhere MCAR'' definition of
  \cite[Definition~5]{Seaman2013} and the ``missing always completely
  at random (MACAR)'' definition of
  \cite[Definition~4]{meallirubin2015} or
  \cite[Equation~18]{little2021}. The classic definition of MCAR of
  \cite{Rubin1976} or \cite[Equation~17]{little2021}, named ``realised
  MCAR'' in \cite[Definition~4]{Seaman2013}, is that for fixed $m$
  (instead of all $m$), %
  $f_{M|Y}(m\,|\,y;\,\bm{\theta})=f_{M|Y}(m\,|\,y';\,\bm{\theta})$ for
  all $\bm{\theta}, y, y'$. \cite[Equation~1.1]{LittleRubin2020} also
  follow this definition of MCAR but only define it under~\ref{A1} and
  omit explicitly stating that the equality has to hold for all
  $\bm{\theta}$.

  Similarly, our MAR definition coincides with the ``everywhere MAR''
  definition of \cite[Definition~2]{Seaman2013}, and the ``missing
  always at random (MAAR)'' definition of
  \cite[Definition~2]{meallirubin2015} or
  \cite[Equation~9]{little2021}. The classical definition of MAR of
  \cite{Rubin1976} or \cite[Equation~3]{little2021}, called ``realised
  MAR'' in \cite[Definition~1]{Seaman2013}, is that for fixed $m$ and
  $y'$ (instead of all $m$ and
  $y'$), %
  $f_{M|Y}(m\,|\,y;\,\bm{\theta})=f_{M|Y}(m\,|\,y';\,\bm{\theta})$ for
  all $\bm{\theta}, y$ such that $x^{\text{obs}}=x'^{\,\text{obs}}$;
  \cite[Definition~1]{meallirubin2015} and
  \cite[Equation~3]{little2021} also follow this definition of MAR.

  Neither \cite{Seaman2013} nor \cite{little2021} seem to provide an explicit
  definition of MNAR, understanding it as ``not MAR''. In terms of
  Definition~\ref{def:MNAR}, MNAR is the logical negation of
  Definition~\ref{def:MAR}. Similarly,
  \cite[Definition~5]{meallirubin2015}, translated to our notation,
  define MNAR as there exist $\bm{\theta},y,y'$ with
  $x^{\text{mis}}\neq x'^{\,\text{mis}}$ such that
  $f_{M|Y}(m\,|\,y;\,\bm{\theta})\neq
  f_{M|Y}(m\,|\,y';\,\bm{\theta})$, %
  which is equivalent to
  ours. %

  As the ``realised'' versions are more conducive to imputation (fixed $M=m$),
  we work with the ``everywhere'' concept, having multiple amputation in mind;
  \cite{Seaman2013} also mention that ``Many other authors [...] have used
  `MCAR' to mean everywhere MCAR.''
\end{remark}

\cite{little2021} uses the terms \emph{unit MCAR}, \emph{unit MAR},
and \emph{unit MNAR} to refer to MCAR, MAR, and MNAR under~\ref{A1},
in which case the defining equalities of conditional probability mass
functions have to hold for every row $i=1,\dots,n$. In general, this
assumption is not necessary.

\subsubsection{The concept of structured missingness}\label{app:SM:def}
Real-world datasets often exhibit \emph{structured missingness (SM)}, e.g.,
missingness mechanisms arising from non-random sampling or data linkage; see
\cite{Schouten2018}, %
\cite{Mitra2023} and \cite{Jackson2023}. SM is an umbrella term covering a range
of missingness mechanisms, revolving around the following two notions:
\begin{enumerate}[label={(SM\arabic*)}, labelwidth=\widthof{(SM2)}]
\item\label{SM1} multivariate missingness, where missingness in at
  least one variable of a row influences missingness in other
  variables of the same row;
\item\label{SM2} deterministic missingness, where data (often blocks)
  are almost surely missing or almost surely not missing.
\end{enumerate}
Assumptions made on the columns of $M$ include
the following (unstructured) assumption~\ref{U}:
\begin{enumerate}[label=(U), labelwidth=\widthof{(U)}]
\item\label{U} The columns $\bm{M}_{,1},\dots,\bm{M}_{,d}$ of $M$ are
  independent given
  $Y$. %
\end{enumerate}
Assumption~\ref{U} leads to the following definition of unstructured and
structured missingness concepts related to MCAR, MAR, and MNAR; see also
\cite{Jackson2023}. ``Unstructured'' here is rather a misnomer, though, as $M$, and
thus $X$, can show a lot of structure (also block missingness) if the respective
(block of) $p_{i,j}$'s are all $1$, irrespective of whether \ref{U} is satisfied
or not; see Section~\ref{sec:non:unique}.
\begin{definition}[Unstructured and structured missingness]
  \begin{enumerate}
  \item
  The missingness mechanism is \emph{missing completely at random
    unstructured (MCAR-U)} if it is MCAR and satisfies~\ref{U},
  otherwise it is \emph{missing completely at random structured
    (MCAR-S)}.
  \item
  The missingness mechanism is \emph{missing at random unstructured
    (MAR-U)} if it is MAR and satisfies~\ref{U}, otherwise it is
  \emph{missing at random structured (MAR-S)}.
  \item
  The missingness mechanism is \emph{missing not at random unstructured
    (MNAR-U)} if it is MNAR and satisfies~\ref{U}, otherwise it is
  \emph{missing not at random structured
    (MNAR-S)}. %
  \end{enumerate}
\end{definition}

Drop-outs in longitudinal studies produce monotone missingness in the
following sense; see \cite{libaccinimealli2014} and
\cite[Section~4.1.1]{vanBuuren2018}.
\begin{definition}[Monotone missingness]\label{def:monotone:miss}%
  \emph{Monotone missingness} patterns are of the form
  $M=(\I_{\{j>j_i\}})_{i=1,\dots,n,\ j=1,\dots,d}$, %
  where $j_i\in\{0,1,\dots,d\}$ for $i=1,\dots,n$.
\end{definition}
We included the case $j_i=0$ to allow for completely missing rows in
the resulting $X$; see Figure~\ref{fig:mtcars01:monotone}. The concept
of monotone missingness could be extended to \emph{locally monotone
  missingness}, where there are $j_{i,1}<j_{i,2}$ such that
$M_{i,j_{i,1}}=\ldots=M_{i,j_{i,2}}=1$ almost surely, which includes
\emph{eventually monotone missingness} where $j_{i,2}=d$. Monotone
missingness is SM and falls under~\ref{SM1}. With
$p_{i,j}=\I_{\{j> j_i\}}$, $i=1,\dots,n$, $j=1,\dots,d$, monotone
missingness also falls under~\ref{SM2}, %
and we obtain that $M$ does not depend on the underlying dependence
structure; see also Section~\ref{sec:non:unique}. However, this $P=(p_{i,j})$
only provides one deterministic monotone missingness pattern, so is degenerate
in the distributional sense.

Suppose an amputer wants to simulate random monotone missingness
patterns $M\in\{0,1\}^{n\times d}$ where, in each row $i$,
$j_i\in\{0,1,\dots,d\}$ is chosen randomly and independently of the
other rows. We can set up a mixture over (a subset of or all) monotone
missingness patterns to model $M$. This allows for the straightforward
and easy-to-simulate stochastic representation
$M=(\I_{\{j>\lceil(d+1)U_i\rceil-1\}})_{i=1,\dots,n,\ j=1,\dots,d}$
for $U_1,\dots,U_n\isim\U(0,1)$. As in many areas of statistics,
mixtures can make a model more versatile; here in the context of
monotone missingness patterns. The assignment of probability $1/(d+1)$
to each of the $d+1$ monotone missingness patterns can be
generalised to any discrete distribution on $\{0,1,\dots,d\}$ (see
Section~\ref{sec:vis:SM} for an example), potentially even depending
on $Y$, being dependent on the row number $i$, or being dependent
across different rows (see Section~\ref{sec:main}).

\subsection{On existing amputation approaches}\label{sec:lit}

\subsubsection{Univariate amputation}\label{sec:uni:amp}
\emph{Univariate amputation} focuses on introducing missingness in one of the
$d$ (as opposed to possibly $nd$) variables, say in the $j$th. Missingness in
multiple columns can be achieved one variable at a time %
(\emph{stepwise univariate amputation}), each based on $Y$. %
Univariate amputation is often implemented based on a logistic regression equation
as in~\eqref{eq:regr} under assumptions to simplify the choices of coefficients;
see also \cite{White2010}, \cite{Hu2013}, \cite{Miao2016}, \cite{Schouten2018}.

If one assumes that $I_{1,j}=\ldots=I_{n,j}=I_j$ for each $j=1,\dots,d$, then
each row $i$ has the same variables $(Y_{i,k})_{k\in I_j}$ that influence
$p_{i,j}$. In this case the amputer has to specify
$\bm{\beta}_{1,j}=\ldots=\bm{\beta}_{n,j}=\bm{\beta}_j$, $j=1,\dots,d$, and so
up to $d(d+1)$ %
coefficients. Under MNAR, one could further assume that
$\bm{\beta}_1=\ldots=\bm{\beta}_d=\bm{\beta}$. This is not possible under MAR, though,
because whatever $\bm{\beta}$ is, there exists at least one component
$j\in\{1,\dots,d\}$ such that $\bm{\beta}$ puts weight on $Y_{i,j}$, which is
not allowed under MAR.

A different assumption could be that for each row $i$, $|I_{i,j}|=1$
(degree $1$), so that $p_{i,j}$ only depends on one variable. In this
case, the amputer has to specify $n$ coefficients. Under MAR, this
would need to be $Y_{i,k}$ for $k\neq j$ (for each $i$,
$I_{i,j}=\{k_i\}\subseteq\{1,\dots,d\}$ for some $k_i\neq j$). Under
MNAR, $p_{i,j}$ would need to only depend on $Y_{i,j}$ (for each $i$,
$I_{i,j}=\{j\}$; this is self-MNAR). In the MNAR case,
\eqref{eq:regr} reads
$p_{i,j} =\logiti(\beta_{i,j;0}+\beta_{i,j;j} Y_{i,j})$ for all
$i=1,\dots,n$, $j=1,\dots,d$. Combining this with the aforementioned
assumption of using the same coefficients per row gives
$p_{i,j} =\logiti(\beta_{j,0}+\beta_{j,j} Y_{i,j})$ for all
$i=1,\dots,n$, $j=1,\dots,d$. If all variables $Y_{i,j}$ are
distributed similarly, then choosing the same coefficients per
variable would lead to having to choose just two coefficients
($\beta_{1,0}=\ldots=\beta_{d,0}=\beta_0$, and
$\beta_{1,1}=\ldots=\beta_{d,d}=\beta_1$).

In practice, amputation of more than one variable is of main interest, but
stepwise univariate amputation typically leads to too small probabilities of
joint missingness through neglect of dependence, as well as an
underrepresentation of missingness structures encountered in real-world
datasets. %
For example, under MCAR, if $p_{i,j}=p$ for all $j$ in row $i$, under univariate amputation,
the probability of the whole row $i$ being missing is $p^d$ due to independence, %
whereas, when incorporating dependence among the entries in the $i$th row, this
probability can be up to $p$; this follows from
Proposition~\ref{prop:joint:missingness:prob}. %

\subsubsection{Multivariate amputation}\label{sec:multi:amp}
\emph{Multivariate amputation} by \cite{Schouten2018} %
based on \cite{Brand1999},
works under Assumptions~\ref{U} and also focuses on
modelling $d$ (as opposed to possibly $nd$) variables. It proceeds by specifying
missingness patterns on groups of rows \emph{a priori}; the missingness patterns
could be devised by considering those present in similar datasets or by expert
opinion, see also those we describe at the end of
Section~\ref{sec:missingness:mechanisms}. The following algorithm summarizes the
procedure.
\begin{algorithm}[Multivariate amputation]\label{alg:scen:approach1}%
  Fix the number $K\in\IN$ of missingness patterns.
  \begin{enumerate}
  \item\label{alg:scen:approach1:patterns} For $k=1,\dots,K$, do: Specify missingness pattern
    $\bm{m}_{k,} \in \{0,1\}^d$ potentially applied to any row
    $\bm{Y}_{1,},\dots,\bm{Y}_{n,}$ of $Y$.
  \item Partition the row numbers $\{1,\dots,n\}$ into sets $I_1,\ldots,I_K$; for
    each $k\in\{1,\dots,K\}$, all rows $\bm{Y}_{i,}$, $i\in I_k$, are candidates
    for receiving missingness pattern $\bm{m}_{k,}$ (the size of $I_k$ relative
    to $n$ specifies the relative frequency of rows $\bm{Y}_{i,}$, $i\in I_k$,
    potentially receiving missingness pattern $\bm{m}_{k,}$).
  \item\label{alg:scen:approach1:permutation} Randomly permute the rows in $Y$.
  \item For $k=1,\dots,K$, do:
    \begin{enumerate}
    \item\label{alg:weight:choice:step} Specify weights
      $\bm{w}_k=(w_{k,0},w_{k,1},\dots,w_{k,d})$ affecting the probability that
      each $\bm{Y}_{i,}$, $i\in I_k$, receives missingness pattern
      $\bm{m}_{k,}$ (or no missingness pattern) %
      according to
      \begin{align}
        \P(\bm{M}_{i,}=\bm{m}_{k,}\,|\,\bm{Y}_{i,}) = \logiti(w_{k,0} + w_{k,1}
        Y_{i,1} + \ldots + w_{k,d} Y_{i,d}). \label{eq:prob}
      \end{align}
    \item\label{alg:weight:choice:step:decision} For each $i\in I_k$,
      do independently: Select $\bm{Y}_{i,}$ to receive missingness
      pattern $\bm{m}_{k,}$ according to the
      probability~\eqref{eq:prob}. %
    \end{enumerate}
  \end{enumerate}
\end{algorithm}
To summarise, block $k$ of $K$ blocks of rows of $Y$ is
assigned missingness pattern $\bm{m}_{k,}\in\{0,1\}^d$, and each row
$\bm{Y}_{i,}$, $i\in I_k$, (each row of block $k$) receives
missingness pattern $\bm{m}_{k,}$ (according to~\eqref{eq:prob}) or
remains complete. Randomly permuting the rows of $Y$ in
Step~\ref{alg:scen:approach1:permutation} guarantees that each row has
some positive chance to receive any of the $K$ missingness
patterns. \cite{Schouten2018} provide qualitative advice on how to
choose the weights. %

The main failing point of multivariate amputation is that manually
specifying missingness patterns $\bm{m}_{1,},\dots,\bm{m}_{K,}$ in
Algorithm~\ref{alg:scen:approach1} is burdensome in high-dimensional
datasets (large $d$), often consequently resulting in a
lack of heterogeneity among missingness patterns. In amputation
algorithms, burden on the amputer is not entirely
avoidable, but the amputer's efforts are better spent on
structural decisions (dependence, margins, who-influences-whom under
MAR and MNAR) rather than designing specific patterns. Moreover, if multiple
amputation is required, an amputer cannot
construct multiple missingness matrices manually that can be viewed as
realisations of the same distribution of $M$ given $Y$.

\subsection{Proofs}\label{sec:proofs}

\begin{proof}[Proof of Proposition~\ref{prop:joint:missingness:prob}]\mbox{}
  \begin{enumerate}
  \item Since $F_{i,j}^{-1}(u)=\I_{(1-p_{i,j},1]}(u)$, the stochastic representation~\eqref{eq:stoch:rep:M} implies
    that
    \begin{align*}
      \P(M=m_S)&=\P(F^{-1}_{i,j}(U_{i,j})=1\ (i,j)\in S,\ F^{-1}_{i,j}(U_{i,j})=0\ (i,j)\in S^c)\\
             &=\P(U_{i,j}>1-p_{i,j}\ (i,j)\in S,\ U_{i,j}\le 1-p_{i,j}\ (i,j)\in S^c)=\Delta_{\left(\tb{\bm{1}-\bm{p}_S,\,\bm{1}}{\bm{0},\,\bm{1}-\bm{p}_{S^c}}\right]}C\\
             &=\P(1-U_{i,j}\le p_{i,j}\ (i,j)\in S,\ 1-U_{i,j}> p_{i,j}\ (i,j)\in S^c)=\Delta_{\left(\tb{\bm{0},\,\bm{p}_{S}}{\bm{p}_{S^c},\,\bm{1}}\right]}\check{C}.
    \end{align*}
    The remaining part of the statement holds as an expectation of a discrete
    distribution.
  \item Since $F_{i,j}^{-1}(u)=\I_{(1-p_{i,j},1]}(u)$, the stochastic representation~\eqref{eq:stoch:rep:M} implies
    that
    \begin{align*}
      \P(M_{i,j}=1\ (i,j)\in S)&=\P(F^{-1}_{i,j}(U_{i,j})=1\ (i,j)\in S)=\P(U_{i,j}>1-p_{i,j}\ (i,j)\in S)\\
                               &=\P(1-U_{i,j}\le p_{i,j}\ (i,j)\in S)=\check{C}_S(\bm{p}_S).
    \end{align*}
    In particular, for $S=\{(i_1,j_1),(i_2,j_2)\}$,
    \begin{align*}
      \P(M_{i_1,j_1}=1,M_{i_2,j_2}=1)&=\check{C}_{(i_1,j_1),(i_2,j_2)}(p_{i_1,j_1},p_{i_2,j_2})\\
                                     &=-1+p_{i_1,j_1}+p_{i_2,j_2}+C_{(i_1,j_1),(i_2,j_2)}(1-p_{i_1,j_1},1-p_{i_2,j_2}).
    \end{align*}
  \end{enumerate}
\end{proof}

\begin{proof}[Proof of Proposition~\ref{pro:pairwise:cor}]
  With stochastic representations
  \begin{align*}
    M_{i_1,j_1}=\I_{(1-p_{i_1,j_1},1]}(U_{i_1,j_1}),\quad M_{i_2,j_2}=\I_{(1-p_{i_2,j_2},1]}(U_{i_2,j_2})
  \end{align*}
  for $(U_{i_1,j_1},U_{i_2,j_2})\sim C_{(i_1,j_1),(i_2,j_2)}$, we
  obtain
  \begin{align*}
    \E(M_{i_1,j_1}M_{i_2,j_2})&=\P(U_{i_1,j_1}>1-p_{i_1,j_1},U_{i_2,j_2}>1-p_{i_2,j_2})\\
                              &=\P(1-U_{i_1,j_1}\le p_{i_1,j_1}, 1-U_{i_2,j_2}\le p_{i_2,j_2})=\check{C}_{(i_1,j_1),(i_2,j_2)}(p_{i_1,j_1},p_{i_2,j_2}).
  \end{align*}
  With $\E(M_{i_k,j_k})=p_{i_k,j_k}$ and
  $\var(M_{i_k,j_k})=p_{i_k,j_k}(1-p_{i_k,j_k})$, $k=1,2$, the result
  follows as stated. The second statement follows by the
  Fr\'echet--Hoeffding bounds theorem, see \cite{frechet1951}, by
  which, pointwise,
  $C^{\text{W}}\le\check{C}_{(i_1,j_1),(i_2,j_2)}\le C^{\text{M}}$,
  and the bounds are attained for the provided copulas $C^{\text{W}}$
  and $C^{\text{M}}$, respectively.
\end{proof}

\begin{proof}[Proof of Lemma~\ref{lem:ran:beta}]
  By~\eqref{eq:regr},
  $p_{i,j}=\logiti(\beta_{i,j;0}+\sum_{k\in I_{i,j}}\beta_{i,j;k} Y_{i,k})$ which
  is, for positive $\beta_{i,j;k}$, $k\in I_{i,j}$, almost surely contained in
  $[\beta_{i,j;0}+c_{\text{min}}\sum_{k\in
    I_{i,j}}\beta_{i,j;k},\beta_{i,j;0}+c_{\text{max}}\sum_{k\in
    I_{i,j}}\beta_{i,j;k}]$. Therefore,
  \begin{align*}
    p_{i,j}\in\biggl[\logiti\biggl(\beta_{i,j;0}+c_{\text{min}}\sum_{k\in I_{i,j}}\beta_{i,j;k}\biggr),\ \logiti\biggl(\beta_{i,j;0}+c_{\text{max}}\sum_{k\in I_{i,j}}\beta_{i,j;k}\biggr)\biggr].
  \end{align*}
  This equals $[p-\eps,p+\eps]$ if and only if $\beta_{i,j;0}$ and
  $\sum_{k\in I_{i,j}}\beta_{i,j;k}$ take the values as stated, which
  proves the first statement. The remaining statement follows from the
  special case of equal $\beta_{i,j;k}$, $k\in I_{i,j}$.
\end{proof}

\printbibliography[heading=bibintoc]
\end{document}

